\documentclass[acmsmall,nonacm]{acmart}

\AtBeginDocument{%
  \providecommand\BibTeX{{%
    \normalfont B\kern-0.5em{\scshape i\kern-0.25em b}\kern-0.8em\TeX}}}
\usepackage[utf8]{inputenc}
\usepackage[TS1,T1]{fontenc}

\usepackage{amsmath}

\usepackage{microtype}

\usepackage{subcaption}

\usepackage{placeins}

\usepackage{stmaryrd}

\makeatletter
\newcommand\footnoteref[1]{\protected@xdef\@thefnmark{\ref{#1}}\@footnotemark}
\makeatother

\usepackage{rotating}

\usepackage{pgfplots}
\usepgfplotslibrary{
  units,
  groupplots,
}

\pgfplotsset{compat=newest}

\usepackage{xspace}
\usepackage{booktabs}
\usepackage{etoolbox}

\usepackage{multirow}

\usepackage{bbm}

\usepackage{proba}
\newcommand{\probability}[1]{\prob{[#1]}}

\usepackage[noend]{algpseudocode}
\usepackage{algorithm}

\newenvironment{code}{\noindent%
\begin{tabbing}%
\hspace{2em}\=\hspace{2em}\=\hspace{2em}\=\hspace{2em}\=\hspace{2em}\=%
\hspace{2em}\=\hspace{2em}\=\hspace{2em}\=\hspace{2em}\=\hspace{2em}\=%
\kill}{\end{tabbing}}
\newcommand{\RRem}[1]   {\`{\bf --\hspace{0.5mm}--~}{\rm#1}}

\usepackage[capitalise]{cleveref}

\makeatletter
\newcounter{algorithmicH}
\let\oldalgorithmic\algorithmic
\renewcommand{\algorithmic}{%
  \stepcounter{algorithmicH}%
  \oldalgorithmic}%
\renewcommand{\theHALG@line}{ALG@line.\thealgorithmicH.\arabic{ALG@line}}
\makeatother

\newtheorem{assumption}{Assumption}
\makeatletter
\if@cref@capitalise
\crefname{assumption}{Assumption}{Assumptions}
\else
\crefname{assumption}{assumption}{assumptions}
\fi
\makeatother
\Crefname{assumption}{Assumption}{Assumptions}

\newcommand{\Is}       {:=}

\newcommand{\punkt}{\enspace .}

\newcommand{\Oh}[1]{\mathcal{O}\!\left( #1\right)}

\newcommand{\Th}[1]{\Theta\!\left( #1\right)}
\newcommand{\Om}[1]{\Omega\!\left( #1\right)}
\newcommand{\om}[1]{\omega\!\left( #1\right)}
\newcommand{\Ohbp}[1]{\mathcal{O}\!\!\>\big( #1\big)}
\newcommand{\Thbp}[1]{\Theta\!\!\>\big( #1\big)}
\newcommand{\Ombp}[1]{\Omega\!\!\>\big( #1\big)}

\newbool{arxiv}
\booltrue{arxiv}

\newbool{plots}
\booltrue{plots}

\newcommand{\dissnote}[1]{}
\newcommand{\todo}[1]{}
\newcommand{\answer}[1]{}
\newcommand{\frage}[1]{}
\newcommand{\note}[1]{}

\newcommand{\VarArray}{A\xspace}
\newcommand{\BaseCaseSize}{n_0\xspace}
\newcommand{\BlockSize}{b\xspace}
\newcommand{\OversamplingFactor}{\alpha\xspace}
\newcommand{\tbegin}{\underline{t}}
\newcommand{\tend}{\overline{t}}

\newcommand{\VarBucketCount}{k\xspace}
\newcommand*{\VarBucket}[1][]{\ifthenelse{\equal{#1}{}}{b}{b_{#1}}\xspace}
\newcommand{\TargetBucketIndex}{\text{dest}}

\newcommand*{\writeblock}[1][]{\ifthenelse{\equal{#1}{}}{w}{w_{#1}}\xspace}
\newcommand*{\readblock}[1][]{\ifthenelse{\equal{#1}{}}{r}{r_{#1}}\xspace}
\newcommand*{\delimiterblock}[1][]{\ifthenelse{\equal{#1}{}}{d}{d_{#1}}\xspace}

\newcommand{\VarThreadCount}{t\xspace}

\newcommand{\compssssort}{\textsf{S\textsuperscript{4}o}\xspace}
\newcommand{\compmyssssaxtmann}{\textsf{1S\textsuperscript{4}o}\xspace}
\newcommand{\compssssschneider}{\textsf{S\textsuperscript{4}oS}\xspace}
\newcommand{\compissssort}{\textsf{I1S\textsuperscript{4}o}\xspace}
\newcommand{\compblock}{\textsf{BlockQ}\xspace}
\newcommand{\compspdq}{\textsf{BlockPDQ}\xspace}
\newcommand{\compssort}{\textsf{std::sort}\xspace}

\newcommand{\compsyaros}{\textsf{DualPivot}\xspace}
\newcommand{\compswiki}{\textsf{WikiSort}\xspace}
\newcommand{\compsgrail}{\textsf{GrailSort}\xspace}
\newcommand{\compsmergequick}{\textsf{QMSort}\xspace}
\newcommand{\compstim}{\textsf{Timsort}\xspace}
\newcommand{\radixlearned}{\textsf{LearnedSort}\xspace}
\newcommand{\radixipp}{\textsf{IppRadix}\xspace}

\newcommand{\radixsska}{\textsf{SkaSort}\xspace}
\newcommand{\compissrsort}{\textsf{I1S\textsuperscript{2}Ra}\xspace}

\newcommand{\imsdradix}{\textsf{IMSDradix}\xspace}
\newcommand{\compiparassssort}{\textsf{IPS\textsuperscript{4}o}\xspace}
\newcommand{\compiparassssortnts}{\textsf{IPS\textsuperscript{4}oNT}\xspace}
\newcommand{\compparastringssssort}{\textsf{StringPS\textsuperscript{4}o}\xspace}
\newcommand{\compmyparassssaxtmann}{\textsf{PS\textsuperscript{4}o}\xspace}
\newcommand{\comppsort}{\textsf{MCSTLmwm}\xspace}
\newcommand{\comppbalancedsort}{\textsf{MCSTLbq}\xspace}

\newcommand{\compppbbs}{\textsf{PBBS}\xspace}
\newcommand{\compptbb}{\textsf{TBB}\xspace}
\newcommand{\radixregion}{\textsf{RegionSort}\xspace}
\newcommand{\comppaspas}{\textsf{ASPaS}\xspace}
\newcommand{\radixraduls}{\textsf{RADULS2}\xspace}
\newcommand{\radixparadis}{\textsf{PARADIS}\xspace}

\newcommand{\radixppbbr}{\textsf{PBBR}\xspace}
\newcommand{\compiparassrsort}{\textsf{IPS\textsuperscript{2}Ra}\xspace}

\newcommand{\pcintellargefour}{\mbox{I4x20}\xspace}
\newcommand{\pcamd}{\mbox{A1x16}\xspace}
\newcommand{\pcinteltwo}{\mbox{I2x16}\xspace}
\newcommand{\pcintelfour}{\mbox{A1x64}\xspace}

\newcommand{\distzipf}{Zipf\xspace}
\newcommand{\distuniform}{Uniform\xspace}
\newcommand{\distexpo}{Exponential\xspace}
\newcommand{\distalmostsorted}{AlmostSorted\xspace}
\newcommand{\distsorted}{Sorted\xspace}
\newcommand{\distreversesorted}{ReverseSorted\xspace}
\newcommand{\distones}{Zero\xspace}
\newcommand{\distduplicatesroot}{RootDup\xspace}
\newcommand{\distduplicatestwice}{TwoDup\xspace}
\newcommand{\distduplicateseight}{EightDup\xspace}

\newcommand{\bytes}{100B\xspace}
\newcommand{\pair}{Pair\xspace}
\newcommand{\quartet}{Quartet\xspace}
\newcommand{\double}{double\xspace}
\newcommand{\uint}{uint32\xspace}
\newcommand{\ulong}{uint64\xspace}

\newcommand{\oset}[2]{[#1\>..\>#2]}
\newcommand{\excloset}[2]{[#1\>..\>#2)}

\begin{document}

\newcommand*\SubLabel[1]{
  \node [
  anchor=north west,
  text width=2em,
  text height=2ex,
  text depth=1ex,
  align=left,
  ] at (axis description cs:0.5,1.2)
  {\phantomsubcaption\label{#1}\subref{#1})};
}

\pgfplotsset{
  layers/my layer set/.define layer set={
    background,
    main,
    foreground
  }{},
  set layers=my layer set,
}

\pgfplotscreateplotcyclelist{my subroutine}{%
  black,every mark/.append style={solid,fill=gray},mark=otimes*,mark size=1.75pt\\%
  orange,every mark/.append style={fill=orange!80!black},mark=triangle*,mark size=1.75pt\\%
  yellow!60!black,dashed,every mark/.append style={solid,fill=yellow!80!black},mark=square*,mark size=1.4pt\\%
  teal,every mark/.append style={fill=teal!80!black},mark=star,mark size=1.75pt\\%
  brown!60!black,every mark/.append style={fill=brown!80!black},mark=square*,mark size=1.4pt\\%
  blue,mark=star,every mark/.append style=solid,mark size=1.75pt\\%
  red!70!white,every mark/.append style={solid,fill=red!80!black},mark=*,mark size=1.4pt\\%
  red,dashed,mark=star,mark size=1.75pt\\%
  red,dashed,every mark/.append style={solid,fill=red!80!black},mark=diamond*,mark size=1.75pt\\%
  red!60!black,dashed,every mark/.append style={solid,fill=red!80!black},mark=square*,mark size=1.4pt\\%
}

\pgfplotsset{
  plotstylesubroutine/.style={
    cycle list name=my subroutine,
  }
}

\pgfplotscreateplotcyclelist{my exotic sequential}{%
  black,every mark/.append style={solid,fill=gray},mark=otimes*,mark size=1.75pt\\%
  lime!80!black,every mark/.append style={fill=lime},mark=*,mark size=1.4pt\\%
  teal,every mark/.append style={fill=teal!80!black},mark=triangle*,mark size=1.75pt\\%
  red,every mark/.append style={fill=red!80!black},mark=diamond*,mark size=1.75pt\\%
  red!60!red,every mark/.append style={solid,fill=red!80!black},mark=star,mark size=1.75pt\\%
  red!60!black,dashed,every mark/.append style={solid,fill=red!80!black},mark=square*,mark size=1.4pt\\%
  red,dashed,every mark/.append style={solid,fill=red!80!black},mark=diamond*,mark size=1.75pt\\%
  red,dashed,mark=star,mark size=1.75pt\\%
  red,every mark/.append style={solid,fill=gray},mark=otimes*,mark size=1.75pt\\%
}

\pgfplotsset{
  plotstylesequential/.style={
    cycle list name=my exotic sequential,
  }
}

\pgfplotscreateplotcyclelist{my exotic sequential perf}{%
  black,every mark/.append style={solid,fill=gray},mark=otimes*,mark size=1.75pt\\%
  lime!80!black,every mark/.append style={fill=lime},mark=*,mark size=1.4pt\\%
  red,every mark/.append style={fill=red!80!black},mark=diamond*,mark size=1.75pt\\%
  red!60!black,dashed,every mark/.append style={solid,fill=red!80!black},mark=square*,mark size=1.4pt\\%
  red,dashed,every mark/.append style={solid,fill=red!80!black},mark=diamond*,mark size=1.75pt\\%
  red,dashed,mark=star,mark size=1.75pt\\%
  red,every mark/.append style={solid,fill=gray},mark=otimes*,mark size=1.75pt\\%
}

\pgfplotsset{
  plotstylesequentialperf/.style={
    cycle list name=my exotic sequential perf,
  }
}

\pgfplotscreateplotcyclelist{my exotic parallel}{%
  black,every mark/.append style={solid,fill=gray},mark=otimes*,mark size=1.75pt\\%
  lime!80!black,every mark/.append style={fill=lime},mark=triangle*,mark size=1.75pt\\%
  teal,every mark/.append style={fill=teal!80!black},mark=star,mark size=1.75pt\\%
  red!60!black,dashed,every mark/.append style={solid,fill=red!80!black},mark=square*,mark size=1.4pt\\%
  red,dashed,mark=star,mark size=1.75pt\\%
  red,dashed,every mark/.append style={solid,fill=red!80!black},mark=diamond*,mark size=1.75pt\\%
  red!70!white,every mark/.append style={solid,fill=red!80!black},mark=*,mark size=1.4pt\\%
  red,every mark/.append style={solid,fill=gray},mark=otimes*,mark size=1.75pt\\%
}

\pgfplotsset{
  plotstyleparallel/.style={
    cycle list name=my exotic parallel,
  }
}

\pgfplotscreateplotcyclelist{my exotic parallel perf}{%
  black,every mark/.append style={solid,fill=gray},mark=otimes*,mark size=1.75pt\\%
  orange,every mark/.append style={fill=orange!80!black},mark=triangle*,mark size=1.75pt\\%
  teal,every mark/.append style={solid,fill=teal!80!black},mark=diamond*,mark size=1.75pt\\%
  red!60!black,dashed,every mark/.append style={solid,fill=red!80!black},mark=square*,mark size=1.4pt\\%
  red,dashed,mark=star,mark size=1.75pt\\%
  red,dashed,every mark/.append style={solid,fill=red!80!black},mark=diamond*,mark size=1.75pt\\%
  red!70!white,every mark/.append style={solid,fill=red!80!black},mark=*,mark size=1.4pt\\%
}

\pgfplotsset{
  plotstyleparallelperf/.style={
    cycle list name=my exotic parallel perf,
  }
}

\title{Engineering In-place (Shared-memory) Sorting Algorithms}

\thanks{This work is based on an earlier work: In-Place Parallel Super Scalar Samplesort (IPS$^4$o). In 25th European Symposium on Algorithms (ESA), Vol. 87. 2017. LIPIcs, 9:1–9:14. \url{https://doi.org/10.4230/LIPIcs.ESA.2017.9}}

\author{Michael Axtmann}
\affiliation{%
  \institution{Karlsruhe Institute of Technology}
  \city{Karlsruhe}
  \state{Germany}
}
\email{michael.axtmann@kit.edu}

\author{Sascha Witt}
\affiliation{%
  \institution{Karlsruhe Institute of Technology}
  \city{Karlsruhe}
  \state{Germany}
}
\email{sascha.witt@kit.edu}

\author{Daniel Ferizovic}
\affiliation{%
  \institution{Karlsruhe Institute of Technology}
  \city{Karlsruhe}
  \state{Germany}
}
\email{dani93.f@gmail.com}

\author{Peter Sanders}
\affiliation{%
  \institution{Karlsruhe Institute of Technology}
  \city{Karlsruhe}
  \state{Germany}
}
\email{sanders@kit.edu}

\begin{abstract}
  We present new sequential and
  parallel sorting algorithms that now represent the fastest known
  techniques for a wide range of input sizes, input distributions,
  data types, and machines. Somewhat surprisingly, part of the speed
  advantage is due to the additional feature of the algorithms to work
  in-place, i.e., they do not need a significant amount of space beyond
  the input array. Previously, the in-place feature often implied
  performance penalties. Our main algorithmic contribution is a
  blockwise approach to in-place data distribution that
  is provably cache-efficient.
  We also parallelize this approach taking dynamic load balancing
  and memory locality into account.

  Our new comparison-based algorithm \emph{In-place Superscalar
    Samplesort (\compiparassssort)}, combines this technique with
  branchless decision trees. By taking cases with many equal elements into account
  and by adapting the distribution degree dynamically, we obtain a highly robust
  algorithm that outperforms the
  best previous in-place parallel comparison-based sorting algorithms
  by almost a factor of three. That algorithm also outperforms the best comparison-based
  competitors regardless of whether we consider in-place or not in-place,
  parallel or sequential settings.

  Another surprising result is that \compiparassssort even outperforms
  the best (in-place or not in-place) integer sorting algorithms in a
  wide range of situations. In many of the remaining cases (often
  involving near-uniform input distributions, small keys, or a
  sequential setting), our new \emph{In-place Parallel Super Scalar
    Radix Sort (\compiparassrsort)} turns out to be the best
  algorithm.

  Claims to have the -- in some sense -- ``best'' sorting algorithm
  can be found in many papers which cannot all be true.  Therefore, we
  base our conclusions on an extensive experimental study involving a
  large part of the cross product of 21 state-of-the-art sorting
  codes, 6 data types, 10 input distributions, 4 machines, 4 memory
  allocation strategies, and input sizes varying over 7 orders of
  magnitude. This confirms the claims made about the robust performance of our
  algorithms while revealing major performance problems in many competitors
  outside the concrete set of measurements
  reported in the associated publications. This is particularly true
  for integer sorting algorithms giving one reason to prefer
  comparison-based algorithms for robust general-purpose sorting.
\end{abstract}

\keywords{in-place algorithm, branch prediction}

\maketitle

\section{Introduction}
Sorting an array of
elements according to a total ordering of their keys is a fundamental
subroutine used in many applications.  Sorting is used for index
construction, for bringing similar elements together, or for
processing data in a ``clever'' order.  Indeed, often sorting is the
most expensive part of a program.  Consequently, a huge amount of
research on sorting has been done.  In particular, algorithm
engineering has studied how to make sorting practically fast in
presence of complex features of modern hardware like multi-core
(e.g.,~\cite{TsiZha03,putze2007mcstl,blelloch2010low,shun2012brief,TsiZha03,reinders2007intel,obeya2019theo,martel1989fast,heidelberger1990parallel,francis1992parallel,hou2015aspas,polychroniou2014comprehensive,kokot2017even,bingmann2017engineering,polychroniou2014comprehensive})
instruction parallelism
(e.g.,~\cite{sanders2004super,Intel2020Ipp,hou2015aspas,polychroniou2014comprehensive}),
branch prediction
(e.g.,~\cite{sanders2004super,KS06,edelkamp2016blockquicksort,yaroslavskiy2009dual,skarupke2016skasort,bingmann2017engineering}),
caches
(e.g.,~\cite{sanders2004super,BFV04,Fran04,blelloch2010low,kokot2017even}),
or virtual memory
(e.g.,~\cite{Rah03,jurkiewicz2015model,SanWas11}).
In contrast, the sorting algorithms used in the
standard libraries of programming languages like Java or C++ still use
variants of quicksort -- an algorithm that is more than 50 years old
\cite{hoare1962quicksort}.  A reason seems to be that you have to
outperform quicksort in every respect in order to replace it.  This is
less easy than it sounds since quicksort is a pretty good algorithm --
a careful randomized implementation needs $\Oh{n\log n}$ expected work
independent of the input, it can be
parallelized~\cite{TsiZha03,putze2007mcstl}, it can be implemented to
avoid branch mispredictions~\cite{edelkamp2016blockquicksort}, and it
is reasonably cache-efficient.  Furthermore, quicksort works in-place
which is of crucial importance for very large inputs. These features
rule out many contenders.  Further algorithms are eliminated by the
requirement to work for arbitrary data types and input distributions.
This makes integer sorting algorithms like radix sort
(e.g.,~\cite{kokot2017even}) or using specialized hardware (e.g.,~GPUs
or SIMD instructions) less attractive since these algorithms
are not sufficiently general for a reusable library that has to work for arbitrary
data types.  Another portability issue is that the algorithm should
use no code specific to the processor architecture or the operating
system like non-temporal writes or overallocation of virtual memory
(e.g. \cite{SanWas11,kokot2017even}).  One aspect of making an
algorithm in-place is that such ``tricks'' are not needed.  Hence,
this paper concentrates on portable algorithms with a particular
focus on comparison-based algorithms and how they can be made robust
for arbitrary inputs, e.g., with a large number of repeated keys or
with skewed input distributions. That said, we also contribute to
integer sorting algorithms and we extensively compare ourselves also
to a number of non-portable and non-comparison-based sorting
algorithms.

The main contribution of this paper is to propose a new algorithm -- \emph{{\bf\em I}n-place {\bf\em P}arallel {\bf\em S}uper {\bf\em S}calar {\bf\em S}ample{\bf\em so}rt}~(\compiparassssort)%
\footnote{The Latin word ``ipso'' means ``by itself'', referring to
  the in-place feature of \compiparassssort.}  -- that combines enough
advantages to become an attractive replacement of quicksort.  Our
starting point is \emph{Super Scalar Samplesort}
(\compssssort)~\cite{sanders2004super} which already provides a very
good sequential non-in-place algorithm that is cache-efficient, allows
considerable instruction parallelism, and avoids branch
mispredictions.  \compssssort is a variant of samplesort
\cite{FraMck70}, which in turn is a generalization of quicksort to
multiple pivots.  The main operation is distributing elements of an
input sequence to $\VarBucketCount$ output buckets of about equal
size.  Our two main innovations are that we make the algorithm
in-place and parallel. The first phase of
\compiparassssort\ distributes the elements to $\VarBucketCount$
\emph{buffer blocks}.  When a buffer block becomes full, it is emptied into
a block of the input array that has already been distributed.
Subsequently, the memory blocks are permuted into the globally correct
order.  A cleanup phase handles empty blocks and half-filled buffer
blocks.  The classification phase is parallelized by assigning disjoint
pieces of the input array to different threads.  The block permutation
phase is parallelized using atomic fetch-and-add operations for each
block move.  Once subproblems become smaller, we adjust their
parallelism until they can be solved independently in parallel.  We
also make \compiparassssort more
robust by taking advantage of inputs with many identical keys.

It turns out that most of what is said above
is not specific to samplesort. It also applies to integer sorting,
specifically \emph{most-significant-digit (MSD) radix sort}~\cite{friend56electronic}
where data distribution is based on extracting the most significant
(distinguishing) bits of the input subproblem. We therefore also
present \emph{In-place Parallel Super Scalar Radix Sort
  (\compiparassrsort)} -- a proof-of-concept implementation of MSD
radix sort using the in-place partitioning framework we developed for
\compiparassssort.

After introducing basic tools in \cref{sec:preliminaries} and
discussing related work in \cref{sec:related}, we describe our new
algorithm \compiparassssort\ in \cref{sec:ips4o} and analyze our
algorithm in \cref{sec:analysis}.  In \cref{s:details} we give
implementation details of \compiparassssort and \compiparassrsort.

We then turn to an extensive experimental
evaluation in \cref{s:experiments}. It turned out that there is a
surprisingly large number of contenders for ``the best'' sorting
algorithm, depending on which concrete setting is considered.  We
therefore consider a large part of the cross product of 21 state-of-the-art sorting codes, 6 data types, 10 input
distributions, 4 machines, 4 memory allocation strategies, and input
sizes varying over 7 orders of magnitude.  Overall, more than
500\,000 different configurations where tried.  Our
algorithm \compiparassssort can be called ``winner'' of this
complicated comparison in the following sense: (1) It outperforms all
competing implementations for most cases. (2) In many of these cases,
the performance difference is large.  (3) When \compiparassssort is
slower than a competitor this is only by a small percentage in the
overwhelming majority of cases -- the only exceptions are for easy
instances that are handled very quickly by some contenders that
however perform poorly for more difficult instances.
Our radix sorter \compiparassrsort complements this by being even faster in some cases
involving few processor cores and small, uniformly distributed keys.
\compiparassrsort also outperforms other radix sorters in many cases.
We believe that our approach to experimental evaluation is also a contribution independent from our new algorithms since it helps to better understand which algorithmic measures (e.g., exploiting various features of a processor architecture) have an impact under which circumstances. 
This approach may thus help to improve future studies of sorting algorithms.

An overall
discussion and possible future work is given in \cref{s:conclusion}.
\ifarxiv The appendix provides further experimental data and proofs. \else An extended version of this paper~\cite{axtmann2020ips4oarxiv} provides further experimental data and proofs. \fi
The codes and benchmarks are available at \url{https://github.com/ips4o}.

%%%%%%%%%%%%%%%%%%%%%%%%%%%%%%%%%%%%%%%%%%%%%%%%%%%%%%%%%%%%%%%%%%%%%%
\section{Definitions and Preliminaries}\label{sec:preliminaries}

The input of a sorting algorithm is an array $\VarArray\oset{0}{n - 1}$ of
$n$ elements, sorted by $t$ threads.  We expect that the output of a sorting
algorithm is stored in the input array.
We use the notation $\oset{a}{b}$ as a shorthand for the ordered set $\{a,\ldots,b\}$ and $\excloset{a}{b}$ for $\{a,\ldots,b-1\}$.
We also use $\log x$ for $\log_2 x$.
  
A machine has one or multiple \emph{CPUs}.
A CPU contains one or multiple \emph{cores}, which in turn contain
one, two, or more \emph{hardware threads}.
We denote a machine with multiple CPUs as a Non-Uniform Memory Access
(\emph{NUMA}) machine if the cores of the CPUs can access their
``local main memory'' faster than the memory attached to the other
CPUs.  We call the CPUs of NUMA machines \emph{NUMA nodes}.
  
  In algorithm theory, an algorithm works in-place if it uses only
  constant space in addition to its input.  We use the term
  \emph{strictly in-place} for this case.  In algorithm engineering,
  one is sometimes satisfied if the additional space is logarithmic in
  the input size.  In this case, we use the term \emph{in-place}.  In
  the context of in-place algorithms, we count machine words and
  equate the machine word size with the size of a data element to be
  sorted.  Note that other space complexity measures may count the
  number of used bits.

  The parallel external memory (\emph{PEM}) model~\cite{AGNS08} is a
  cache-aware extension of the parallel random-access machine.  This
  model is used to analyze parallel algorithms if the main issue is
  the number of accesses to the main memory.  In the PEM model, each
  of the $t$~threads has a private cache of size~$M$ and access to
  main memory happens in \emph{memory blocks} of size~$B$.
  The \emph{I/O complexity} of an algorithm is the asymptotic number
  of parallel memory block transfers (I/Os) between the main memory
  and the private caches.  An algorithm is denoted as
  \emph{I/O-efficient} if its I/O complexity is optimal.  In this
  work, we use the term \emph{cache-efficient} as a synonym for
  I/O-efficient when we want to emphasize that we consider memory
  block transfers between the main memory and the private
  cache.

  We adopt an asynchronous variant of the PEM model where we charge
  $t$ I/Os if a thread accesses a variable which is shared with other
  threads.  To make this a realistic assumption, our implementations
  avoid additional delays due to \emph{false sharing} by allocating at
  most one shared variable to each memory block.
  \dissnote{Michael:
    Cite~\cite[p.~29-32]{sanders2019sequential} Sequential models
    ~\cite[p.~31-32]{sanders2019sequential} External Memory Model
    \cite[p.~p.~80]{sanders2019sequential} Parallel External Memory
    Model}

%......................................................................
\subsubsection*{(Super Scalar) Samplesort.}

The $\VarBucketCount$-way \compssssort
algorithm~\cite{sanders2004super} 
starts with allocating two temporary arrays of size $n$ -- one data
array to store the buckets, and one so-called \emph{oracle array}.  The
partitioning routine contains three phases and is executed recursively.
The sampling phase sorts
$\OversamplingFactor \VarBucketCount - 1$~randomly sampled input
elements where the \emph{oversampling factor} $\OversamplingFactor$ is
a tuning parameter.  The splitters~$S=\oset{s_0}{s_{\VarBucketCount-2}}$ are then picked equidistantly from the sorted
sample.  The classification phase classifies each input element,
stores its target bucket in a so-called \emph{oracle array}, and
increases the size of its bucket.  Element $e$ goes to bucket
$\VarBucket[i]$ if~$s_{i-1}< e\leq s_i$ (with $s_{-1}=-\infty$ and
$s_{\VarBucketCount-1}=\infty$).  Then, a prefix sum is used to
calculate the bucket boundaries.  The distribution phase uses the
oracle array and the bucket boundaries to copy the elements from the
input array into their buckets in the temporary data array.

The main contribution of \compssssort to samplesort is to use a
decision tree for element classification which eliminates branch
mispredictions (\emph{branchless decision tree}).  Assuming
$\VarBucketCount$ is a power of two, the splitters are stored in an
array $a$ representing a complete binary search tree:
$a_1=s_{\VarBucketCount/2-1}$, $a_2=s_{\VarBucketCount/4-1}$,
$a_3=s_{3\VarBucketCount/4-1}$, and so on. More generally, the left
successor of $a_i$ is $a_{2i}$ and its right successor is $a_{2i+1}$.
Thus, navigating through this tree is possible by performing a
conditional instruction for incrementing an array index.  \compssssort
completely unrolls the loop that traverses the decision tree to reduce
the instruction count.  Furthermore, the loop for classifying elements
is unrolled several times to reduce data dependencies between
instructions. This allows a higher degree of instruction parallelism.

Bingmann~et~al.~\cite{bingmann2017engineering} apply the branchless
decision tree to parallel~string~sample~sorting
(\compparastringssssort) and add additional buckets for elements identical
to a splitter.
After the decision tree of \compssssort has assigned element $e$ to
bucket $b_i$, \compparastringssssort updates the bucket to introduce
additional equality buckets for elements corresponding to a splitter:
Element $e$ goes to bucket $\VarBucket[2i+\mathbbm{1}_{e = s_i}]$
if~$i<k-1$, otherwise $e$ goes to bucket $\VarBucket[2i]$.\footnote{We
  use $\mathbbm{1}_c$ to express a conversion of a comparison result
  to an integer. When $c$ is true, $\mathbbm{1}_c$ is equal to
  $1$. Otherwise, it is equal to $0$.}  The case distinction $i<k-1$
is necessary as $a_{k-1}$ is undefined.

For \compiparassssort, we adopt (and refine) the approach of element
classification but change the organization of buckets in order to make
\compssssort in-place and parallel.  Our element classification works
as follows: Beginning by the source node $i=1$ of the decision tree,
the next node is calculated by $i \gets 2i+\mathbbm{1}_{a_i<e}$.  When
the leafs of the tree are reached, we update $i$ once more
$i\gets 2i+\mathbbm{1}_{a_i<e} - k$.  For now, we know for $e$ that
$s_{i-1}<e \leq s_i$ if we assume that $s_{-1}=-\infty$ and that
$s_{k-1}=\infty$.  Finally, the bucket of $e$ is
$2i+1-\mathbbm{1}_{e<s_i}$.  Note that we do not use the comparison
$\mathbbm{1}_{e = s_i}$, from $\compparastringssssort$ to calculate
the final bucket.  The reason is that our algorithm accepts a compare
function $<$ and \compparastringssssort compares radices.  Instead, we
use $1-\mathbbm{1}_{e<s_i}$ which is identical to $\mathbbm{1}_{e = s_i}$
since we already know that $e \leq s_i$.
Also, note that we avoid the case distinction $i<k-1$
from the classification of \compparastringssssort which may
potentially cause a branch misprediction.  Instead, we set
$s_{k-1} = s_{k-2}$.  Compared to \compssssort and
\compparastringssssort, we support values of $k$ which are no powers
of two, i.e., when we had removed splitter duplicates in our
algorithm.  In these cases, we round up $k$ to the next power of two
and pad the splitter array $S$ with the largest splitter.  We note
that this does not increase the depth of the decision tree.
\cref{fig:decision tree} depicts our refined decision tree and
\cref{alg:element classification1} classifies elements using the
refined decision tree.  \cref{alg:element classification1} classifies
a chunk of elements in one step.  A single instruction of the decision
tree traversal is executed on multiple elements before the next
operation is executed.  We use loops to execute each instruction on a
constant number of elements.  It turned out that recent compilers
automatically unroll these loops and remove the instructions of the
loops for code optimization.

\begin{figure}[tbp]
  \begin{center}
    \includegraphics[]{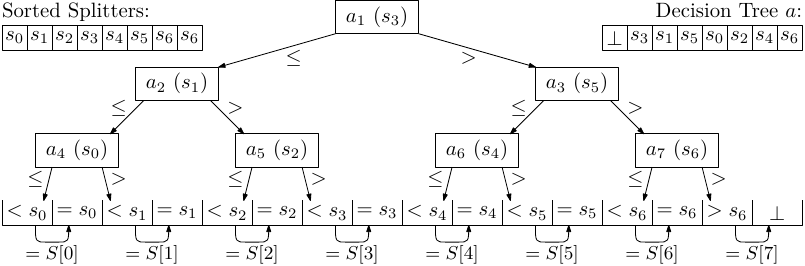}
  \end{center}
  \caption{\label{fig:decision tree}
    Branchless decision tree with $7$ splitters and $15$ buckets, including $7$ equality buckets.
    The first entry of the decision tree array stores a dummy to allow tree navigation.
    The last splitter in the sorted splitter array is duplicated to avoid a case distinction.
  }
\end{figure}

\begin{algorithm}[t]
  \begin{algorithmic}
    \caption{Element classification of the first $u\lfloor n/u\rfloor$ elements}\label{alg:element classification1}
    \State \textbf{Template parameters:} $s$ number of splitters, $u$ unroll factor, $\mathit{equalBuckets}$ boolean value which indicates the use of equality buckets
    \State \textbf{Input:} \begin{minipage}[t]{33em}$\VarArray\oset{0}{n-1}$ an array of $n$ input elements

      $\mathit{tree}\oset{1}{s}$ decision tree, splitters stored in left-to-right breadth-first order

      $\mathit{splitter}\oset{0}{s}$ sorted array of $s$ splitters, last splitter is duplicated

      \Call{compare}{$e_l, e_r$} a comparator function which returns $0$ or $1$

      \Call{output}{$e, t$} an output function which gets an element $e$ and its target bucket $t$\end{minipage}
    \State $l\gets$ \Call{$\log_2$}{$s+1$}\Comment{Log.\ number of buckets (equality buckets are excluded)}
    \State $k\gets 2^{l + 1}$\Comment{Number of buckets}
    \State $b\oset{0}{u-1}$\Comment{Array to store current position in the decision tree}
    \For{$j\gets 0$ {\bf in steps of} $u$ {\bf to } $n - u$}\Comment{Loop over elements in blocks of $u$}
    \For{$i\gets 0$ {\bf to} $u-1$}
    \State $b[i]\gets 1$\Comment{Set position to the tree root}
    \EndFor
    \For{$r\gets 0$ {\bf to} $l$}\Comment{Unrolled by most compilers as $l$ and $u$ are constants}
    \For{$i\gets 0$ {\bf to} $u-1$}
    \State $b[i]\gets 2 \cdot b[i] +$ \Call{compare}{$\mathit{tree}[b[i]], a[j + i]$}\Comment{Navigate through the tree}
    \EndFor
    \EndFor
    \If{$\mathit{equalBuckets}$}
    \For{$i\gets 0$ {\bf to} $u-1$}\Comment{Assign elements identical to the splitter to its equality bucket}
    \State $b[i]\gets 2 \cdot b[i] + 1-$\Call{compare}{$a[j + i], \mathit{splitter}[b[i] - k / 2]$}
    \EndFor
    \EndIf
    \For{$i\gets 0$ {\bf to} $u-1$}
    \State \Call{output}{$b[i] - k, a[j + i]$}
    \EndFor
    \EndFor
  \end{algorithmic}
\end{algorithm}

% ------------------------------------------------------------------------------
\section{Related Work}\label{sec:related}

\subsubsection*{Quicksort.} Variants of Hoare's
quicksort~\cite{hoare1962quicksort,musser1997introspective} are
generally considered some of the most efficient general-purpose
sorting algorithms.  Quicksort works by selecting a \emph{pivot}
element and partitioning the array such that all elements smaller than
the pivot are in the left part and all elements larger than the pivot
are in the right part.  The subproblems are solved recursively.
Quicksort (with recursion on the smaller subproblem first) needs
logarithmic additional space for the recursion stack. 
Strictly in-place variants~\cite{vdurian1986quicksort, bing1986one,
  wegner1987generalized} of quicksort avoid recursion, process the array from
left to right, and use a careful placement of the pivots to find the
end of the leftmost partition.
A variant of quicksort (with a fallback to heapsort to avoid worst-case scenarios) is currently used in the C++ standard library of
GCC~\cite{musser1997introspective}.

Some variants of quicksort use two or three
pivots~\cite{yaroslavskiy2009dual,KLMQ14} and achieve improvements of
around $20$\,\% in running time over the single-pivot case. The basic principle of quicksort remains,
but elements are partitioned into three or four subproblems instead of
two.

Quicksort can be parallelized in a scalable way by parallelizing both
partitioning and
recursion~\cite{martel1989fast,heidelberger1990parallel,francis1992parallel}. Tsigas
and Zhang~\cite{TsiZha03} show in practice how to do this
in-place. Their algorithm scans the input from left to right and from
right to left until the scanning positions meet -- as in most
sequential implementations. The crucial adaptation is to do this in a
blockwise fashion such that each thread works at one block from each
scanning direction at a time. When a thread finishes a block from one
scanning direction, it acquires a new one using an atomic
fetch-and-add operation on a shared pointer.  This process terminates
when all blocks are acquired.  The remaining unfinished blocks are
resolved in a sequential cleanup phase.  Our \compiparassssort\
algorithm can be considered as a generalization of this approach to
$k$ pivots.  This saves a factor $\Th{\log k}$ of passes through the
data.  We also parallelize the cleanup process.

\subsubsection*{Samplesort.}
Samplesort~\cite{FraMck70,blelloch1991comparison,blelloch2010low}\note{first
  sequential, the last two parallel} can be considered as a
generalization of quicksort which uses $\VarBucketCount-1$ splitters
to partition the input into $\VarBucketCount$ subproblems (from now on
called \emph{buckets}) of about equal size.
Unlike single- and dual-pivot quicksort, samplesort is usually not
in-place, but it is well-suited for parallelization and more
cache-efficient than quicksort.

\compssssort~\cite{sanders2004super} improves samplesort by avoiding
inherently hard-to-predict conditional branches linked to element
comparisons.  Branch mispredictions are very expensive because they
disrupt the pipelined and instruction-parallel operation of modern
processors.  Traditional quicksort variants suffer massively from
branch mispredictions~\cite{KS06}.  By replacing
conditional branches with conditionally executed machine instructions,
branch mispredictions can be largely avoided.  This is done
automatically by modern compilers if only a few instructions depend on
a condition.  As a result, \compssssort is up to two times faster than
quicksort~(\texttt{std::sort}), at the cost of $\Oh{n}$ additional
space.  BlockQuicksort~\cite{edelkamp2016blockquicksort} applies
similar ideas to single-pivot quicksort, resulting in a very fast
in-place sorting algorithm with performance similar to \compssssort.

For \compiparassssort, we used a refined version of the branchless
decision tree from \compssssort.  As a starting point, we took the
implementation of the branchless decision tree from
\compssssschneider, an implementation of \compssssort written by
Lorenz Hübschle-Schneider.  \compssssort has also been adapted for
efficient parallel string sorting~\cite{bingmann2017engineering}.  We
apply their approach of handling identical keys to our decision tree.

\subsubsection*{Radix Sort.}
As for samplesort, the core of radix sort is a
$k$-way data partitioning routine which is recursively executed.  In
its simplest way, all elements are classified once to determine the
bucket sizes and then a second time to distribute the elements.  Most
partitioning routines are applicable to samplesort as well as to radix
sort.  Samplesort classifies an element with $\Th{\log k}$ invocations
of the comparator function while radix sort just extracts a digit of
the key in constant time.
In-place $k$-way data partitioning is often done element by element,
e.g., in the sequential in-place radix sorters American
Flag~\cite{mcilroy1993engineering} and
\radixsska~\cite{skarupke2016skasort}. However, these approaches have two
drawbacks. First, they perform the element classification twice.  This
is a particular problem when we apply this approach to samplesort as
the comparator function is more expensive.  Second, a naive
parallelization where the threads use the same pointers and acquire
single elements suffer from read/write
dependencies.

In 2014, Orestis~and~Ross~\cite{polychroniou2014comprehensive} outlined a
parallel in-place radix sorter that moves blocks of elements
in its $k$-way data partitioning routine.  We use the same general approach for \compiparassssort.
However, the paper \cite{polychroniou2014comprehensive} leaves open
how the basic idea can be turned into a correct in-place
algorithm. The published prototypical implementation uses $20$\,\%
additional memory, and does not work for small inputs or a number of threads different from 64.

In 2015,
Minsik~et~al.\ published \radixparadis~\cite{cho2015paradis}, a
parallel in-place radix sorter.  The partitioning routine of
\radixparadis classifies the elements to get bucket boundaries and
each thread gets a subsequence of unpartitioned elements from each
bucket.  The threads then try to move the elements within their
subsequences so that the elements are placed in the subsequence of
their target bucket.  This takes time $\Oh{n/t}$.  Depending on the
data distribution, elements may still be in the wrong bucket.  In this
case, the threads repeat the procedure on the unpartitioned elements.
Depending on the key distribution, the load of the threads in the
partitioning routine differs significantly. No bound better than
$\Oh{n}$ is known for this partitioning
routine~\cite{obeya2019theo}. 

In 2019, Shun et al.~\cite{obeya2019theo} proposed an in-place
$k$-way data partitioning routine for the radix sorter
\radixregion.  This algorithm builds a graph
that models the relationships between element regions and their target
buckets.  Then, the algorithm performs multiple rounds where the
threads swap regions into their buckets.

To the best of our knowledge, the initial version of
\compiparassssort~\cite{axtmann2017confplace}, published in 2017, is
the first parallel $k$-way partitioning algorithm that moves elements
in blocks, works fully in-place, and gives adequate performance
guarantees.  Our algorithm \compiparassssort is more general than
\radixregion in the sense that it is comparison based. To demonstrate
the advantages of our approach, we also propose the radix sorter
\compiparassrsort which adapts our in-place partitioning
routine.

\subsubsection*{(Strictly) In-Place Mergesort}

There is a
considerable amount of theory work on strictly in-place sorting
(e.g.,~\cite{Fran04,FraGef05,langston1991timespace}).
\note{\cite{Fran04,FraGef05,langston1991timespace} are only a theory
  paper -- no implementation found, both strictly in-place} %
\note{\cite{langston1991timespace} is parallel and not mergesort.}
However, there are few -- mostly negative -- results of transferring
the theory work into practice.  Implementations of non-stable in-place
mergesort~\cite{katajainen1996practical,elmasry2012branch,edelkamp2019worst}
are reported to be slower than quicksort from the C++ standard
library.  Katajainen and Teuhola report that their
implementation~\cite{katajainen1996practical} is even slower than
heapsort, which is quite slow for big inputs due to its
cache-inefficiency.  The fastest non-stable in-place mergesort
implementation we have found is QuickMergesort (\compsmergequick) from
Edelkamp~et~al.~\cite{edelkamp2019worst}.  Relevant implementations of
stable in-place mergesort are \compswiki (derived
from~\cite{pokson2008patio}) and \compsgrail (derived
from~\cite{huang1992faststable}).
However,
Edelkamp~et~al.~\cite{edelkamp2019worst} report that \compswiki is a
factor of more than $1.5$ slower than \compsmergequick for large
inputs and that \compsgrail performs similar to \compswiki.
Edelkamp~et~al.\ also state that non-in-place mergesort is considerably
faster than in-place mergesort.
There is previous theoretical
work on
sequential (strictly) in-place multi-way merging \cite{Geffert2009}.
However, this approach needs to allocate very large blocks to become
efficient.  In contrast, the block size of \compiparassssort\ does not
depend on the input size.
The best practical multi-core mergesort algorithm we
found is the non-in-place multi-way mergesort algorithm (\comppsort)
from the MCSTL library~\cite{putze2007mcstl}.
We did not find any practical parallel in-place
mergesort implementation.

% ------------------------------------------------------------------------------
\section{In-Place Parallel Super Scalar Samplesort (\compiparassssort)}\label{sec:ips4o}

 \compiparassssort is a recursive algorithm.
 Each recursion level divides the input into $\VarBucketCount$~buckets (\emph{partitioning step}), such that each element of bucket~$\VarBucket[i]$ is smaller than all elements of $\VarBucket[i+1]$.
 Partitioning steps operate on the input array in-place and are executed with one or more threads, depending on their size.
If a bucket is smaller than a certain base case size, we invoke a base case algorithm on the bucket (\emph{base case}) to sort small inputs fast.
 A scheduling algorithm determines at which time a base case or partitioning step is executed and which threads are involved.
 We describe the partitioning steps in \cref{sec:partitioning step} and the scheduling algorithm in \cref{sec:task scheduling}.

\subsection{Sequential and Parallel Partitioning}\label{sec:partitioning step}

\begin{figure}[tbp]
  \begin{center}
    \includegraphics[]{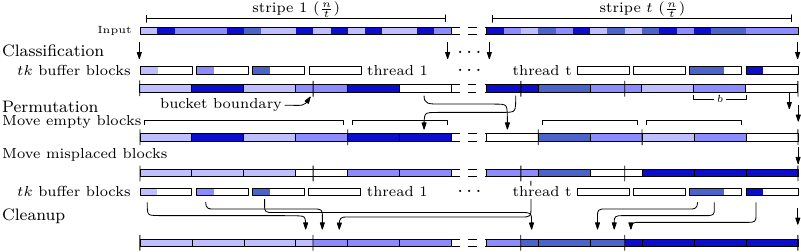}
  \end{center}
  \caption{\label{fig:partitioning step overview}
    Overview of a parallel $k$-way partitioning step ($k=4$) with $t$ threads and blocks of size three.
    Elements with the same color belong into the same bucket.
    The brighter the color, the smaller the bucket index.
    This figure depicts the first and last stripe of the input array, containing $n/t$ elements each.
    In the classification phase, thread $i$ classifies the elements of stripe $i$, moves elements of bucket $j$ into its buffer block $j$, and flushes the buffer block back into its stripe in case of an overflow.
    In the permutation phase, the bucket boundaries are calculated and the blocks belonging into bucket $j$ are placed in the blocks after bucket boundary $j$ in two steps:
    First, the empty blocks are moved to the end of the bucket.
    Then, the misplaced blocks are moved into its bucket.
    The cleanup phase moves elements which remained in the buffer blocks and elements which overlap into the next bucket to their final positions.
    }
\end{figure}

A partitioning step consists of four phases, executed sequentially or by a (sub)set of the input threads.
{\bf Sampling} determines the bucket boundaries.
{\bf Classification} groups the input into blocks such that all elements in a block belong to the same bucket.
{\bf (Block) permutation} brings the blocks into the globally correct order.
Finally, we clean up blocks that cross bucket boundaries or remained partially filled in the {\bf cleanup} phase.
\Cref{fig:partitioning step overview} depicts an overview of a parallel partitioning step.
The following paragraphs will explain each of these phases in more detail.

%......................................................................
\subsubsection{Sampling.}\label{sec:sampling}

Similar to the sampling in \compssssort, the sampling phase of \compiparassssort creates a branchless decision tree  -- the tree follows the description of the decision tree proposed by Sanders and Winkel, extended by equality buckets\footnote{The authors describe a similar technique for handling duplicates, but have not implemented the approach for their experiments.}.
For a description of the decision tree used in \compssssort including our refinements, we refer to \cref{sec:preliminaries}.
In \compiparassssort, the decision tree is used in the classification phase to assign elements to buckets.

The sampling phase performs four steps.
First, we sample $k\OversamplingFactor$ elements of the input.
We swap the samples to the front of the partitioning step to keep the in-place property even if the oversampling factor $\OversamplingFactor$ depends on $n$.
Second, $k-1$ splitters are picked equidistantly from the sorted sample.
Third, we check for and remove duplicates from the splitters.
This allows us to decrease the number of buckets $k$ if the input contains many duplicates.
Finally, we create the decision tree.
The strategy for handling identical keys is enabled conditionally:
The decision tree only creates equality buckets when there are several identical splitters.
Otherwise, we create a decision tree without equality buckets.
Having inputs with many identical keys can be a problem for samplesort, since this might move large fractions of the keys through many recursion levels.
The equality buckets turn inputs with many identical keys into ``easy'' instances as they introduce separate buckets for elements identical to splitters
(keys occurring more than $n/\VarBucketCount$ times are likely to become splitters).

%......................................................................
\subsubsection{Classification}

\begin{figure}[tbp]
  \begin{center}
    \includegraphics[]{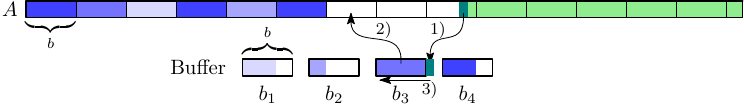}
  \end{center}
  \caption{\label{fig:block gen progress}
  Classification.
  Blue elements have already been classified, with different shades indicating different buckets.
  Unprocessed elements are green.
  Here, the next element (in dark green) has been determined to belong to bucket~$\VarBucket[3]$.
  As that buffer block is already full, we first write it into the array~$\VarArray$, then write the new element into the now empty buffer.
  }
\end{figure}

\begin{figure}[tbp]
  \begin{center}
  \includegraphics[]{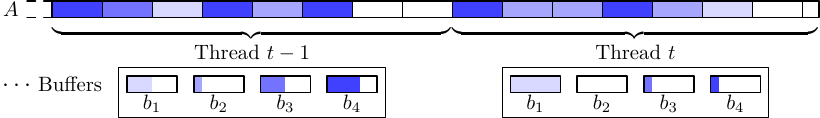}
  \end{center}
  \caption{\label{fig:block generation final}
  Input array and block buffers of the last two threads after classification.
  }
\end{figure}

The input array $\VarArray$ is viewed as an array of blocks each containing $\BlockSize$~elements (except possibly for the last one).
For parallel processing, we divide the blocks of $\VarArray$ into $\VarThreadCount$~stripes of equal size -- one for each thread.
Each thread works with a local array of $\VarBucketCount$ \emph{buffer blocks} -- one for each bucket.
A thread then scans its stripe.
Using the search tree created in the sampling phase, each element in the stripe is classified into one of the $\VarBucketCount$~buckets and
  then moved into the corresponding local buffer block.
If this buffer block is already full, it is first written back into the local stripe, starting at the front.
It is clear that there is enough space to write $\BlockSize$~elements into the local stripe,
  since at least $\BlockSize$ more elements have been scanned from the stripe than have been written back -- otherwise, no full buffer could exist.

In this way, each thread creates blocks of $\BlockSize$~elements belonging to the same bucket.
\Cref{fig:block gen progress} shows a typical situation during this phase.
To achieve the in-place property, we do not track which bucket each block belongs to.
However, we count how many elements are classified into each bucket, since we need this information in the following phases.
This information can be obtained almost for free as a side effect of maintaining the buffer blocks.
\Cref{fig:block generation final} depicts the input array after classification.
Each stripe contains full blocks, followed by empty blocks.
The remaining elements are still contained in the buffer blocks.

%......................................................................
\subsubsection{Block Permutation}

In this phase, the blocks in the input array are rearranged such that they appear in the correct order.
From the classification phase we know, for each stripe, how many elements belong to each bucket.
We first aggregate the per-thread bucket sizes and then compute a prefix sum over the total bucket sizes. This yields the exact boundaries of the buckets in the output.
Roughly, the idea is then that each thread repeatedly looks for a misplaced block $B$ in some bucket $\VarBucket[i]$, finds
the correct destination bucket $\VarBucket[j]$ for $B$, and swaps $B$ with a misplaced block in $\VarBucket[j]$.
If $b_j$ does not contain a misplaced block, $B$ is moved to an empty block in $\VarBucket[j]$.
The threads are coordinated by maintaining atomic read and write pointers for each bucket.
Costs for updating these pointers are amortized by making blocks sufficiently large.

We now describe this process in more detail beginning with the preparations needed before starting the actual block permutation.
We mark the beginning of each bucket $\VarBucket[i]$ with a delimiter pointer~$\delimiterblock[i]$, rounded up to the next block.
We similarly mark the end of the last bucket~$\VarBucket[\VarBucketCount]$ with a delimiter pointer~$\delimiterblock[\VarBucketCount+1]$.
Adjusting the boundaries may cause a bucket to ``lose'' up to $\BlockSize - 1$ elements;
this doesn't affect us, since this phase only deals with full blocks, and
elements outside full blocks
remain in the buffers.
Additionally, if the input size is not a multiple of $\BlockSize$, some of the $\delimiterblock[i]$s may end up outside the bounds of $\VarArray$.
To avoid overflows, we allocate a single empty \emph{overflow block} which the algorithm will use instead of writing to the final (partial) block.

\begin{figure}[tbp]
  \begin{center}
    \includegraphics[]{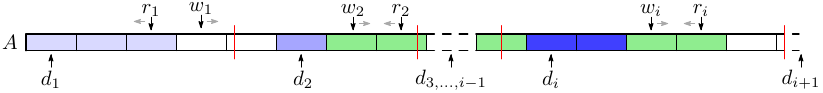}
  \end{center}
  \caption{\label{fig:block perm invariant}
    Invariant during block permutation.
    In each bucket~$\VarBucket[i]$, blocks in $[\delimiterblock[i],\,\writeblock[i])$ are already correct (blue),
      blocks in $[\writeblock[i],\,\readblock[i]]$ are unprocessed (green), and blocks in $[\max(\writeblock[i],\readblock[i]+1),\,\delimiterblock[i+1])$ are empty (white).
  }
\end{figure}

For each $\VarBucket[i]$, a write pointer~$\writeblock[i]$ and a read pointer~$\readblock[i]$ are introduced;
  these will be set such that all unprocessed blocks, i.e., blocks that still need to be moved into the correct bucket,
  are found between $\writeblock[i]$~and~$\readblock[i]$.
During the block permutation, we maintain the following invariant for each bucket $\VarBucket[i]$, visualized in \cref{fig:block perm invariant}:

\begin{itemize}
  \item Blocks to the left of $\writeblock[i]$ (exclusive) are correctly placed, i.e., contain only elements belonging to $\VarBucket[i]$.
  \item Blocks between $\writeblock[i]$ and $\readblock[i]$ (inclusive) are unprocessed, i.e., may need to be moved.
  \item Blocks to the right of $\max(\writeblock[i],\readblock[i]+1)$ (inclusive) are empty.
\end{itemize}

In other words, each bucket follows the pattern of correct blocks followed by unprocessed blocks followed by empty blocks,
with $\writeblock[i]$ and $\readblock[i]$ determining the boundaries.
In the sequential case, this invariant is already fulfilled from the beginning.
In the parallel case, all full blocks are at the beginning of each stripe, followed by its empty blocks.
This means that only the buckets crossing a stripe boundary need to be fixed.

To do so, each thread finds the bucket that starts before the end of its stripe but ends after it.
It then finds the stripe in which that bucket ends (which will be the following stripe in most cases) and
  moves the last full block in the bucket into the first empty block in the bucket.
It continues to do this until either all empty blocks in its stripe are filled or all full blocks in the bucket have been moved.

In rare cases, very large buckets exist that cross multiple stripes.
In this case, each thread will first count how many blocks in the preceding stripes need to be filled.
It will then skip that many blocks at the end of the bucket before starting to fill its own empty blocks.

The threads are then ready to start the block permutation.
Each thread maintains two local swap buffers that can hold one block each.
We define a \emph{primary} bucket~$\VarBucket[p]$ for each thread;
  whenever both its buffers are empty, a thread tries to read an unprocessed block from its primary bucket.
To do so, it decrements the read pointer~$\readblock[p]$ (atomically) and reads the block it pointed~to into one of its swap buffers.
If $\VarBucket[p]$ contains no more unprocessed blocks (i.e., $\readblock[p] < \writeblock[p]$), it switches its primary bucket to the next bucket (cyclically).
If it completes a whole cycle and arrives back at its initial primary bucket, there are no more unprocessed blocks and the whole permutation phase ends.
The starting points for the threads are distributed across that cycle to reduce contention.

\begin{figure}[tbp]
  \begin{center}
    \begin{subfigure}[t]{0.496\textwidth}
      \includegraphics[]{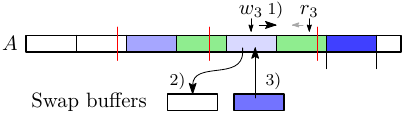}
      \caption{\label{fig:block perm a}\begin{minipage}[t]{0.85\textwidth}
        Swapping a block into its correct position.
      \end{minipage}}
    \end{subfigure}
    \hfill
    \begin{subfigure}[t]{0.496\textwidth}
      \includegraphics[]{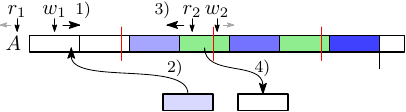}
      \caption{\label{fig:block perm b}\begin{minipage}[t]{0.85\textwidth}
        Moving a block into an empty position, followed by refilling the swap buffer.
      \end{minipage}}
    \end{subfigure}
  \end{center}
  \vspace{-10pt}
  \caption{\label{fig:block perm}
    Block permutation examples.
    The numbers indicate the order of the operations.
  }
\end{figure}

\begin{figure}[tbp]
  \begin{center}
    \includegraphics[]{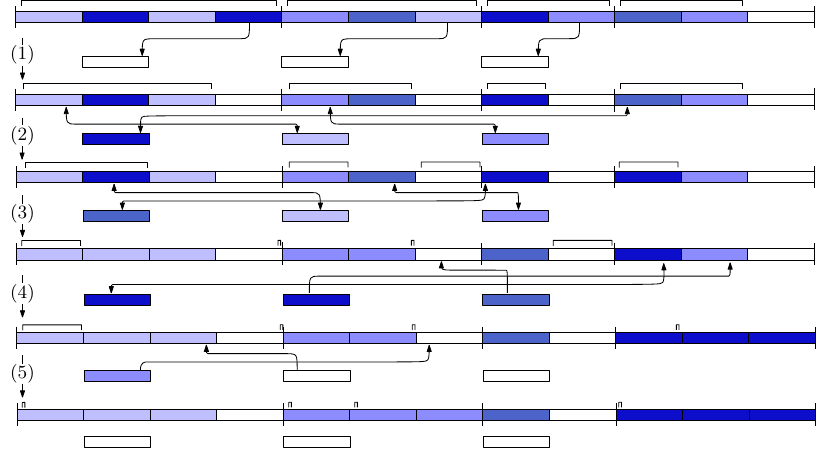}
  \end{center}
  \caption{\label{fig:parallel permutation}
    An example of a permutation phase with $k=4$ buckets and $t=3$ threads.
    The brackets above the buckets mark the unprocessed blocks.
    After five permutation steps, the blocks were moved into their target buckets.
    (1) The buffer blocks are filled with blocks.
    (2-3) Swap buffer block with the leftmost unprocessed block of the buffer block's buckets.
    (4) Thread $0$ and $1$ have a buffer block for the last bucket.
    They increase the write pointer of this bucket concurrently.
    Thread $0$ executes the fetch-and-add operation first, the thread swaps its buffer block with the second block (unprocessed block) of the last bucket.
    Thread $1$ writes its buffer block into the third block (empty block) of the last bucket.
    After step four, threads $1$ and $2$ finished a permutation chain, i.e., flushed their buffer block into an empty block.
    (5) Thread $0$ flushes its buffer block into an empty block.
    Thread $1$ classifies the last unprocessed block of the first bucket but this block is already in its target bucket.
  }
\end{figure}

Once it has a block, each thread classifies the first element of that block to find its destination bucket~$\VarBucket[\TargetBucketIndex]$.
There are now two possible cases, visualized in \cref{fig:block perm}:

\begin{itemize}
  \item As long as $\writeblock[\TargetBucketIndex] \leq \readblock[\TargetBucketIndex]$, write pointer $\writeblock[\TargetBucketIndex]$ still points to an unprocessed block in bucket $\VarBucket[\TargetBucketIndex]$. In this case, the thread increases $\writeblock[\TargetBucketIndex]$,
      reads the unprocessed block into its empty swap buffer, and writes the other one into its place.
  \item If $\writeblock[\TargetBucketIndex] > \readblock[\TargetBucketIndex]$, no unprocessed block remains in bucket $\VarBucket[\TargetBucketIndex]$ but $\writeblock[\TargetBucketIndex]$ now points to an empty block. In this case, the thread increases $\writeblock[\TargetBucketIndex]$, writes its swap buffer to the empty block, and then reads a new unprocessed block from its primary bucket.
\end{itemize}

We repeat these steps until all blocks are processed.
We can skip unprocessed blocks which are already correctly placed:
We simply classify blocks \emph{before} reading them into a swap buffer, and skip as needed.

It is possible that one thread wants to write to a block that another
thread is currently reading from (when the reading thread has just
decremented the read pointer but has not yet finished reading the
block into its swap buffer).  However, threads are only allowed to
write to empty blocks if no other threads are currently reading from
the bucket in question, otherwise, they must wait.  Note that this
situation occurs at most once for each bucket, namely when
$\writeblock[\TargetBucketIndex]$ and $\readblock[\TargetBucketIndex]$
cross each other.  We avoid these data races by keeping track of how
many threads are reading from each bucket.

When a thread fetches a new unprocessed block, it reads and modifies
either $\writeblock[i]$ or $\readblock[i]$.  The thread also needs to
read the other pointer for the case distinctions.  These operations
are performed simultaneously to ensure a consistent view of both
pointers for all threads.  \Cref{fig:parallel permutation} depicts an
example of the permutation phase with three threads and four
buckets.
%......................................................................
\subsubsection{Cleanup}

After the block permutation, some elements may still be in incorrect positions
since blocks may cross bucket boundaries.
We call the partial block at the beginning of a bucket its \emph{head} and the partial block at its end its \emph{tail}.

Thread $i$ performs the cleanup for buckets $\excloset{\lfloor ki/t\rfloor}{\lfloor k(i+1)/t\rfloor}$.
Thread~$i$ first reads the head of the first bucket of thread $i+1$ into one of its swap buffers.
Then, each thread processes its buckets from left to right, moving incorrectly placed elements into empty array entries
The incorrectly placed elements of bucket~$\VarBucket[i]$ can be in four locations:
\begin{enumerate}
\item Elements may be in the head of $\VarBucket[i+1]$ if the last block belonging into bucket $\VarBucket[i]$ overlaps into bucket  $\VarBucket[i+1]$.
\item Elements may be in the partially filled buffers from the classification phase.
\item Elements of the last bucket (in this thread's area) may be in the swap buffer.
\item Elements of one bucket may be in the overflow buffer.
\end{enumerate}
Empty array entries consist of the head of $\VarBucket[i]$ and any (empty) blocks to the right of $\writeblock[i]$ (inclusive).
Although the concept is relatively straightforward, the implementation is somewhat involved, due to the many parts that have to be brought together.
\Cref{fig:cleanup} shows an example of the steps performed during the cleanup phase.
Afterwards, all elements are back in the input array and correctly partitioned, ready for recursion.

\begin{figure}[tbp]
  \begin{center}
    \includegraphics[]{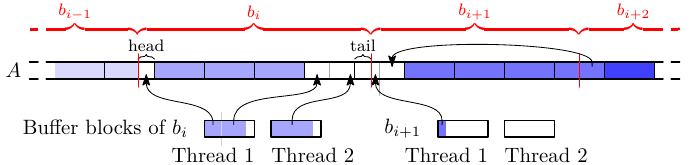}
  \end{center}
  \caption{\label{fig:cleanup}
    An example of the steps performed during cleanup.
  }
\end{figure}

% ------------------------------------------------------------------------------

\subsection{Task Scheduling}\label{sec:task scheduling}

In this section, we describe the task scheduling of \compiparassssort.
We also establish basic properties of the task scheduler.
The properties are used to understand how the task scheduler works.
The properties are also used later on in~\cref{sec:analysis} to analyze the parallel I/O complexity and the local work of \compiparassssort.

In general, \compiparassssort uses static scheduling to apply tasks to threads.
When a thread becomes idle, we additionally perform a dynamic rescheduling of sequential tasks to utilize the computation resources of the idle thread.
Unless stated otherwise, we exclude the dynamic rescheduling from the analysis of \compiparassssort and only consider the static load balancing.
We state that dynamic load balancing --  when implemented correctly -- cannot make things worse asymptotically.

Before we describe the scheduling algorithm, we introduce some definitions and derive some properties from these definitions to understand how the task scheduler works.
A task $T[l, r)$ either partitions the (sub)array $\VarArray[l, r - 1]$ with a partitioning step (\emph{partitioning task}) or sorts the base case $\VarArray[l, r - 1]$ with a base case sorting algorithm (\emph{base case task}).
A partitioning step performed by a group of threads (a \emph{thread group}) is a \emph{parallel partitioning task} and a partitioning step with one thread is a \emph{sequential partitioning task}.
Each thread has a \emph{local stack} to store \emph{sequential tasks}, i.e., sequential partitioning tasks and base cases.
Additionally, each thread $i$ stores a handle $G_i$ to its current thread group and has access to the handles stored by the other threads of the thread group.

To decide whether a task $T[l,r)$ is a parallel partitioning task or a sequential task, we denote $\tbegin$ as $\lfloor lt / n\rfloor$ and $\tend$ as $\lfloor rt/n\rfloor$.
The task $T[l, r)$ is a parallel partitioning task when $\tend - \tbegin > 1$.
In this case, the task is executed by the thread group $\excloset{\tbegin}{\tend}$.
Otherwise, the task is a sequential task processed by thread $\min (\tbegin, \tend - 1)$.
As a parallel partitioning task is the only task type executed in parallel, we use \emph{parallel task} as a synonym for parallel partitioning task.

We use a base case threshold $n_0$ to determine whether a sequential task is a sequential partitioning task or a base case task:
Buckets with at most $2\BaseCaseSize$ elements as well as buckets of a task with at most $k\BaseCaseSize$ elements are considered as base cases.
Otherwise, it is a sequential partitioning task.
We use an adaptive number of buckets $k$ for partitioning steps with less than $k^2n_0$ elements, such that the expected size of the base cases is between $0.5n_0$ and $n_0$ while the expected number of buckets remains equal or larger than $\sqrt{k}$.
To this end, for partitioning steps with $kn_0 < n' < k^2n_0$ elements ($n_0<n' \leq kn_0$ elements), we adjust $k$ to $2^{\lceil(\log_2(n'/n_0)+1)/2\rceil}$ (to $2^{\lceil\log_2(n'/n_0)\rceil}$).
This adaption of $k$ is
important for a robust running time of \compiparassssort.  In the
analysis of \compiparassssort, the adaption of $k$ will allow us to
amortize the sorting of samples.  In practice, our experiments
have shown that for fixed $k$, the running time per element
oscillates with maxima around $k^in_0$.

From these definitions, the \cref{lem:parallel task covers element,lem:sequential subgroup,lem:par task is bucket,lem:par task unique level} follow.
The lemmas allow a simple scheduling of parallel tasks and thread groups.

\begin{lemma}\label{lem:parallel task covers element}
  The parallel task $T[l,r)$ covers position $(i+1)n/t-1, i\in \excloset{0}{t}$ of the input array if and only if thread $i$ executes the parallel task $T[l,r)$.
\end{lemma}

\begin{proof}  
  We first prove that thread $i$ executes the parallel task $T[l,r)$ if $T[l,r)$ covers position $(i+1)n/t-1, i\in \excloset{0}{t}$ of the input array.
  Let the parallel task $T[l,r)$ cover position $w=(i+1)n/t-1, i\in \excloset{0}{t}$ of the input array.
  From the inequalities
  \begin{align*}
    &i = \lfloor ((i+1)n/t - 1)t/n\rfloor \geq \lfloor l_s t/n\rfloor = \tbegin_s\\
    &i = \lfloor ((i+1)n/t-1)t/n \rfloor < \lfloor r_st/n\rfloor = \tend_s
  \end{align*}
  follows that thread $i$ executes the parallel task $T[l,r)$.
  For the ``$\geq$'' and the ``$<$'', we use that task $T[l,r)$ covers position $w$ of the input array, i.e., $l\leq w<r$.
  
  We now prove that a parallel task $T[l,r)$ must cover position $(i+1)n/t-1$ of the input array if it is executed by thread $i$.
  Let us assume that a parallel task $T[l,r)$ is executed by thread $i$.
  From the inequalities
  \begin{align*}
    & (i+1)n/t - 1 \geq (\tbegin + 1)n/t-1 \geq l_s\\
    & (i+1)n/t - 1 < \tend n/t \leq r_s
  \end{align*}
  follows that task $T[l, r)$ covers position $(i+1)n/t-1$ of the input array.
  For the second ``$\geq$'' and for the ``$\leq$'', we use the definition for the thread group of $T[l, r)$, i.e., $\excloset{\tbegin}{ \tend} = \excloset{\lfloor lt/n \rfloor}{\lfloor rt/n \rfloor}$.
\end{proof}

\begin{lemma}\label{lem:sequential subgroup}
  Let the sequential task $T[l_s, r_s)$, processed by thread $i$ be a bucket of a parallel task $T[l,r)$.
  Then, task $T[l,r)$ is processed by thread $i$ and others.
\end{lemma}

\begin{proof}
  Let the parallel task $T[l,r)$ be processed by threads $\excloset{\tbegin}{ \tend}$.
  Task $T[l_s, r_s)$ is processed by thread $i = \min (\lfloor l_st/n\rfloor , \tend - 1)$.
  We have to show that $\tbegin\leq i < \tend$ holds.
  Indeed, inequality $i =  \min (\lfloor l_st/n\rfloor , \tend - 1) < \tend$ holds.
  The inequality $i \geq \tbegin$ is only wrong if $\lfloor l_st/n\rfloor < \tbegin$ or if $\tend- 1 < \tbegin$.
  However, we have $l \leq l_s$ as task $T[l_s,r_s)$ is a bucket of its parent task $T[l,r)$.
  Thus, $ \lfloor l_st/n\rfloor\geq \lfloor lt/n\rfloor = \tbegin$ holds as the first thread $\tbegin$ of $T[l,r)$ is defined as $\lfloor lt/n\rfloor$.
  Also, we have $\tend - 1 > \tbegin$ as task $T[l, r)$ is a parallel task with at least two threads, i.e., $\tend - \tbegin > 1$.
\end{proof}

\begin{lemma}\label{lem:par task is bucket}
  Let the parallel subtask $T[l_s, r_s)$, processed by thread $i$ and others, be a bucket of a task $T[l,r)$.
  Then, task $T[l,r)$ is also a parallel task processed by thread $i$ and others.
\end{lemma}

\begin{proof}
  Task $T[l,r)$ is a parallel task if $\lfloor rt/n \rfloor - \lfloor lt/n \rfloor>1$.
  This inequality is true as
  \begin{align*}
    \lfloor rt/n \rfloor - \lfloor lt/n \rfloor \geq \lfloor r_st/n \rfloor - \lfloor l_st/n \rfloor  > 1\punkt
  \end{align*}
  For the ``$\geq$'' we use that $T[l_l,r_l)$ is a bucket of $T[l,r)$ and for the ``$>$'' we use that $T[l_s, r_s)$ is a parallel task.

  As the parallel task $T[l_s, r_s)$ is processed by thread $i$,  $T[l_s, r_s)$ covers the position $(i+1)n/t - 1$ of the input array (see \cref{lem:parallel task covers element}).
  As task $T[l_s, r_s)$ is a bucket of $T[l,r)$, the parallel task $T[l,r)$ also covers the position $(i+1)n/t - 1$.
  From \cref{lem:parallel task covers element} follows that the parallel task $T[l,r)$ is processed by thread $i$.
\end{proof}

\begin{lemma}\label{lem:par task unique level}
On each recursion level, thread $i$ works on at most one parallel task.
\end{lemma}

\begin{proof}
  Let $S_i^j, i\in \excloset{0}{t}$ be the set of parallel tasks on recursion level $j$ which cover the position $(i + 1) n/t - 1$ of the input array.
  From \cref{lem:parallel task covers element} follows that thread $i$ processes on recursion level $j$ only the tasks $S_i^j$.
  The set $S_i^j$ contains at most one task as tasks on the same recursion level are disjoint.
\end{proof}

\begin{figure}[tbp]
  \begin{center}
  \includegraphics[width=\textwidth]{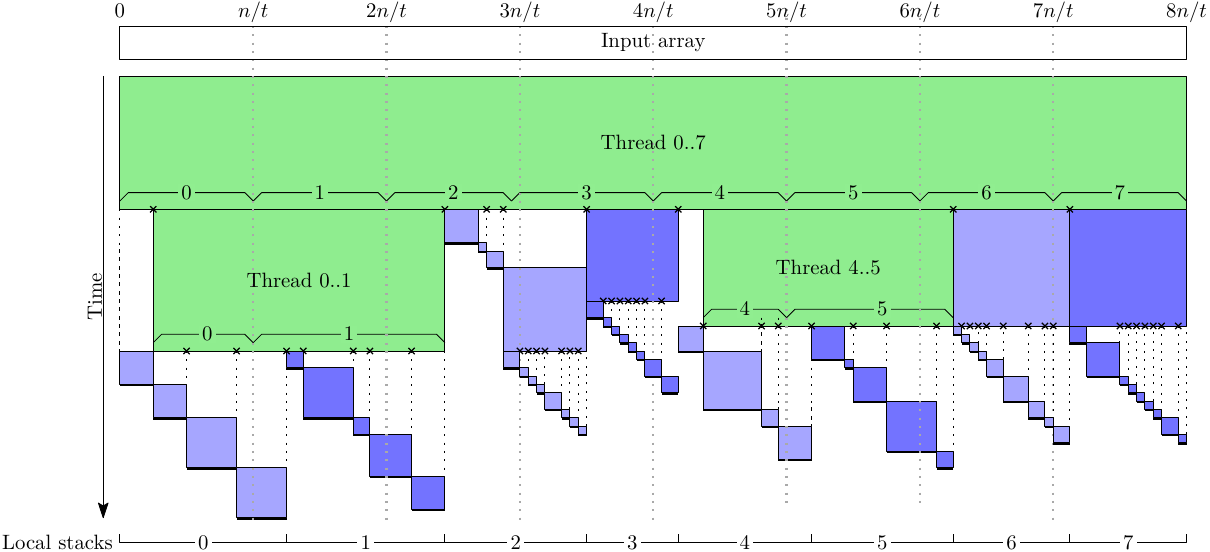}
  \end{center}
  \caption{\label{fig:job scheduling}
    Example schedule of a task execution in \compiparassssort with $8$ threads where partitioning steps split tasks into $8$ buckets.
    Each rectangle represents a task in execution.
    The height of a task is defined by the size of the task divided by the number of threads assigned to this task.
    For parallel tasks (green), the threads processing that task are shown in the rectangles.
    The sequential partitioning tasks (blue) are covered by the local stack which stores the task until processing.
    Base case tasks are omitted for the sake of simplicity.
    The crosses at the bottom of a rectangle indicate bucket boundaries.
    The brackets pointing downwards are used to decide in which local stack the sequential subtasks are inserted.
    Tasks stored in local stack $i$ are executed by thread $i$.
  }
\end{figure}

\subsubsection{Static Scheduling}

We start the description of the task scheduler by describing the static scheduling part.
The idea behind the static scheduling is that each thread executes its tasks in depth-first search order tracing parallel tasks first.
From \cref{lem:par task unique level,lem:par task is bucket} follows, keeping the order of execution in mind, that each thread first executes all of its parallel tasks before it starts to execute its sequential tasks.

\compiparassssort starts by processing a parallel task $T[0, n)$ with threads $\excloset{0}{t}$.
In general, when a parallel task $T[l, r)$ is processed by the thread group $G=\excloset{\tbegin }{ \tend}$, five steps are performed.
\begin{enumerate}
\item A parallel partitioning step is invoked on $\VarArray[l, r - 1]$.
\item The buckets of the partitioning step induce a set of subtasks $S$.
\item If subtask $T[l_s, r_s)\in S$ is a sequential task, thread $i =\min (l_st/n, \tend - 1)$ adds the subtask to its local stack.
From \cref{lem:sequential subgroup}, we know that thread $i$ is actually also processing the current task $T[l, r)$.
This allows threads to add sequential tasks exclusively to their own local stack, so no concurrent stacks are required.
\item Each thread $i\in\excloset{\tbegin}{\tend}$ extracts the subtask $T_s=T[l_s, r_s)$ from $S$ which covers position $(i+1)n/t-1$ of the input array $\VarArray$.
  Also, thread $i$ calculates $\tbegin_s=\lfloor l_s t/n\rfloor$ as well as $\tend_s = \lfloor r_s t/n\rfloor$ and continues with the case distinction  $\tend_s - \tbegin_s \leq 1$ and  $\tend_s - \tbegin_s > 1$.

  If $\tend_s - \tbegin_s \leq 1$, thread $i$ once synchronizes with $G$ and starts processing the sequential tasks on its private stack.
  
  Otherwise, $T_s$ is actually a parallel task that has to be processed by the threads $\excloset{\tbegin_s}{ \tend_s}$.
  From \cref{lem:parallel task covers element,lem:par task is bucket} follows that the threads $\excloset{\tbegin_s}{ \tend_s}$ are currently all processing $T[l,s)$ and exactly these threads selected the same task $T_s$.
  This allows setting up the threads $\excloset{\tbegin_s}{ \tend_s}$ for the next parallel task $T_s$ without keeping the threads waiting:
  The first thread $\tbegin_s$ of task $T_s$ creates the data structure representing the task's new thread group $G'=\excloset{\tbegin_s}{ \tend_s}$ and updates the thread group handles $\excloset{G_{\tbegin_s}}{G_{\tend_s}}$ of the threads $\excloset{\tbegin_s}{ \tend_s}$ to the new data structure.
  Afterwards, all threads of $\excloset{\tbegin_s}{ \tend_s}$ synchronize with thread group $G$ and access their new thread group  $G'$ using the updated thread group handles.
  Finally, the threads $\excloset{\tbegin_s}{ \tend_s}$ start processing task $T_s$ with thread group $G'$.
\end{enumerate}

If a thread no longer processes another parallel task, it starts processing the sequential tasks of its stack until the stack is empty.
Base cases are sorted right away.
When the next task $T[l, r)$ is a sequential partitioning task, three steps are performed.
First, a sequential partitioning step is executed on $\VarArray[l, r - 1]$.
Second, a new sequential subtask is created for each bucket.
Finally, the thread adds these subtasks to its local stack in sorted order.
\Cref{alg:task scheduler} shows the steps of the task scheduling algorithm in detail.
The scheduling algorithm is executed by all threads simultaneously.
\Cref{fig:job scheduling} shows an example schedule of a task execution in \compiparassssort.

\begin{algorithm}[t]
  \begin{algorithmic}
    \caption{Task Scheduler}\label{alg:task scheduler}
    \State \textbf{Input:} $\VarArray\oset{0}{n-1}$ array of $n$ input elements, $t$ number of threads, $i$ current thread
    \State $T[l, r)\gets T[0, n)$\Comment{Current task, initialized with $\VarArray\oset{1}{n}$}
    \State $G_{i}[\tbegin, \tend) = G[0, t)$\Comment{Initialize thread group containing thread $\tbegin=0$ to $\tend=t$ (excl.)}
    \State $D\gets\emptyset$\Comment{Empty local stack}
    \If{$\tend - \tbegin = 1$} \Call{D.pushFront}{$T[l,r)$}\Comment{Initial task is a sequential, go to sequential phase}
    \Else
    \While{{\bf true}}\Comment{Execute current parallel task}
    \State $\oset{b_0}{ b_{k-1}}\gets $\Call{partitionParallel}{$A[l,r-1], G_{i}$}\Comment{Partitioning step; returns buckets}
    \For{$\excloset{l_s }{ r_s}\Is b_{k-1}$ {\bf to} $b_0$}\Comment{Handle the buckets}
    \If{$(i+1) n/t-1 \in [l_s, r_s)$}\Comment{Update current task}
    \State $T[l,r)\gets T[l_s, r_s)$\Comment{It might be $i$'s next parallel task}
    \EndIf
    \If{$\lceil r_st/n\rceil - \lceil l_st/n\rceil \leq 1$ {\bf and} $i = \min (l_st/n, \tend - 1)$}
    \State \Call{D.push}{$\{T[l_s, r_s), l, r\}$}\Comment{Thread $i$ adds sequential task to its local stack}
    \EndIf
    \EndFor
    \State $\tbegin\gets l \cdot t/n$; $\tend\gets r \cdot t/n$\Comment{Range of threads used by current task}
    \If{$\tend - \tbegin \leq 1$} {\bf break}\Comment{Go to sequential phase as $T[l,r)$ is not a parallel task}
    \EndIf
    \If{$i = \tbegin$}\Comment{Thread $i$ creates the thread subgroup as it is the first thread}
    \State $G_{i}\gets$ \Call{createThreadGroup}{$\oset{\tbegin}{ \tend}$}
    \For{$j\Is \tbegin$ {\bf to} $\tend - 1$}\Comment{Set subgroup for all subgroup threads}
    \State $G_j\gets $ \Call{ReferenceOf}{$G_{i}$}
    \EndFor
    \EndIf
    \State \Call{WaitFor}{$\tbegin$}\Comment{Wait until thread subgroup is created}
    \State \Call{JoinThreadGroup}{$G_{i}$}\Comment{Join shared data structures}
    \EndWhile
    \EndIf
    \While{\Call{notEmpty}{$D$}}\Comment{Execute sequential tasks}
    \State $\{T[l_s, r_s), l, r\}\gets$ \Call{pop}{$D$}
    \If{$r_s-l_s \leq 2n_0$ {\bf or} $r-l \leq kn_0$}
    \State \Call{processBaseCase}{$A[l_s, r_s-1]$}
    \Else
    \State $\oset{b_0}{ b_{k-1}}\gets $\Call{partitionSequential}{$A[l_s, r_s-1]$}\Comment{Partitioning step -- returns buckets}
    \For{$b\Is b_{k-1}$ {\bf to} $b_0$}
    \State \textproc{D.push}{($\{T[$\textproc{begin}{(b)}$,\>$\textproc{end}{(b)}$), l_s, r_s\}$)}\Comment{Add seq.\ subtasks}
    \EndFor
    \EndIf
    \EndWhile
  \end{algorithmic}
\end{algorithm}

From \cref{lem:par task range,lem:seq task range} follows that the workload of sequential tasks and parallel tasks is evenly divided between the threads.
This property is used in~\cref{sec:analysis} to analyze the parallel I/O complexity and the local work.

\begin{lemma}\label{lem:par task range}
  Let $T[l,r)$ be a parallel task with thread group $\excloset{\tbegin}{\tend}$ and $t'= \tend - \tbegin$ threads.
  Then, $T[l,r)$ processes a consecutive sequence of elements which starts at position $l\in\oset{\tbegin n/t}{ (\tbegin + 1) n/t - 1}$ and which ends at position $r\in\oset{\tend n/t-1}{ (\tend + 1)n/t - 1}$ of the input array.
  This sums up to $\Th{t'n/t}$ elements in total.
\end{lemma}

Thus, the size of a parallel task is proportional to the size of its thread group.

\begin{proof}[Proof of~\cref{lem:par task range}]
  From \cref{lem:parallel task covers element} follows that $T[l,r)$ covers position $(\tbegin + 1) n/t - 1$ but not position $\tbegin n/t - 1$ of the input array.
  It also follows, that $T[l,r)$ covers position $\tend n/t-1$ but not position $(\tend + 1)n/t - 1$ of the input array.
\end{proof}

\begin{lemma}\label{lem:seq task range}

  Thread $i$ processes sequential tasks only containing elements from $\VarArray[i n/t,(i + 2)n/t-1]$.
  This sums up to $\Oh{n/t}$ elements in total.
\end{lemma}

This lemma shows that the load of sequential tasks is
evenly distributed among the threads.

\begin{proof}[Proof of~\cref{lem:seq task range}]
  We prove the following proposition:
  When a thread $i$ starts processing sequential tasks, the tasks only contain elements from $\VarArray[i n/t,(i + 2)n/t-1]$.
  From this proposition, \cref{lem:seq task range} follows directly as thread $i$ only processes sequential subtasks of these tasks.
  
  Let the sequential subtask $T[l_s, r_s)$ be a bucket of a parallel task $T[l, r)$ with threads $\excloset{\tbegin }{ \tend}$.
  Assume that $T[l_s, r_s)$ was assigned to the stack of thread $i$.
  We show $i n/t \leq l_s < r_s \leq (i + 2)n/t$ with the case distinction $i < \tend - 1$ and $i \geq \tend - 1$.

  Assume $i < \tend - 1$.
  From the calculation of $i$, we know that
  \begin{alignat*}{2}
    &i &&= \min(\lfloor l_s t/n\rfloor, \tend - 1) = \lfloor l_s t/n\rfloor \leq l_s t/n\\
    \implies & l_s &&\geq i n/t\punkt
  \end{alignat*}
  We show that $r_s \leq (i + 2)n/t$ with a proof by contradiction.
  For the proof, we need the inequality $l_s < (i+1) n/t$ which is true as
  \begin{align*}
    i = \min(\lfloor l_s t/n\rfloor, \tend - 1) = \lfloor l_s t/n\rfloor > l t/n - 1\punkt
  \end{align*}
  Now, we assume that $r_s > (i + 2)n/t$.
  As $T[l_s, r_s)$ is  a sequential task, we have $\lfloor r_s t/n\rfloor - \lfloor l_s t/n \rfloor=1$.
  However, this leads to the contradiction
  \begin{align*}
    1 = \lfloor r_s t/n\rfloor - \lfloor l_s t/n \rfloor \geq i + 2 - l_s t/n > (i + 2) - (i + 1) = 1\punkt
  \end{align*}
  Thus, we limited the end of the sequential task to $r_s \leq (i + 2)n/t$ and its start to $l_s \geq i n/t$ for $i < \tend - 1$.

  Assume $i \geq \tend - 1$.
  In this case, $i$ is essentially equal to $\tend - 1$ as \cref{lem:sequential subgroup} tells us that a sequential subtask of a parallel task is assigned to a thread of the parallel task.
  From the calculation of thread $i$, we know that
  \begin{alignat*}{2}
    &i &&= \min(\lfloor l_s t/n\rfloor, \tend - 1) = \tend - 1 \leq l t/n\\
    \implies & l_s &&\geq n/t(\tend - 1)\punkt    
  \end{alignat*}
  The end $r$ of the parent task can be bounded by
  \begin{alignat*}{2}
    &\tend &&= \lfloor r t/n\rfloor \geq r t/n - 1\\
    \implies &r &&\leq (\tend + 1)n/t
  \end{alignat*}
  We can use this inequality to bound the end $r_s$ of the sequential subtask $T[l_s, r_s)$ to
  \begin{align*}
    r_s \leq r \leq (\tend + 1)n/t
  \end{align*}
  as the subtask does not end after the parent task's end $r$.
  Thus, we limited the end of the sequential task to $r_s \leq (i + 2)n/t$ and its start to $l_s \geq i n/t$ for $i \geq \tend - 1$.
\end{proof}

\subsubsection{Dynamic Rescheduling}

The task scheduler is extended to utilize computing resources of threads that no longer have tasks.
We implement a simplified version of \emph{voluntary work sharing} proposed for parallel string sorting~\cite{bingmann2017engineering}.
A global stack is used to transfer sequential tasks to idle threads.
Threads without sequential tasks increase a global atomic counter which tracks the number of idle threads.
Threads with sequential tasks check the counter after each partitioning step and move one task to the global stack if the counter is larger than zero.
Then, an idle thread can consume the task from the global stack by decreasing the counter and processing the task.
The algorithm terminates when the counter is equal to $t$ which implies that no thread has a task left.
We expect that we are able to amortize the additional cost in most cases or even reduce the work on the critical execution path:
As long as no thread becomes idle, the counter remains valid in the thread's private cache and the threads only access their local stacks.
When a thread becomes idle, the local counter-copies of the other threads are invalidated and the counter value is reloaded into their private cache.
In most cases, we can amortize the counter reload by the previously processed task, as the task has typically more than $\Om{k\BaseCaseSize}$ elements.
When a thread adds an own task to the global stack, the task transfer is amortized by the workload reduction.

%......................................................................
\section{Analysis of \compiparassssort}\label{sec:analysis}

In this section, we analyze the additional memory requirement (\cref{ss:inplace property}), the I/O complexity (\cref{ss:io efficiency}), and the local work (\cref{ss:total work}) of  \compiparassssort.
The analysis in \cref{ss:io efficiency,ss:total work} assumes the following constraints for \compiparassssort:

\begin{assumption}\label{as:1}Minimum size of a logical data block of \compiparassssort: $b\in\Om{tB}$\end{assumption}
\begin{assumption}\label{as:3} Minimum number of elements per thread: $n/t\in\Om{\max(M, bt)}$.\end{assumption}
\begin{assumption}\label{as:9} Restrict I/Os while sampling and buffers fit into private cache: $M\in\Om{Bk\log k + bk}$.\end{assumption}
\begin{assumption}\label{as:2} Oversampling factor: $\alpha\in\Th{\log k'}$ where $k'$ is the current number of buckets.\end{assumption}
\begin{assumption}\label{as:6} Restrict maximum size of base cases: $n_0\in\Om{\log k}\cap\Oh{M/k}$. \end{assumption}

Without loss of generality, we assume that an element has the size of one machine word.
In practice, we keep the block size $b$ the same, i.e., the number of elements in a block is inverse proportional to the element size.
In result, we can guarantee that the size of the buffer blocks does not exceed the private cache without adapting $k$.

\ifarxiv \else

We use the following theorem to bound the number of recursion levels of \compiparassssort. For the proof of this theorem, we refer to the extended version of this paper~\cite{axtmann2020ips4oarxiv}.

\begin{theorem}\label{thm:ips4olevel}
  Let $M\geq 1$ be a constant.
  Then, after $\Oh{\log_k \frac{n}{M}}$ recursion levels, all non-equality buckets of \compiparassssort have size $M$ with a probability of at least $1 - n/M$ for an oversampling ratio of $\alpha=\Th{c\log k}$.
\end{theorem}

\fi % arxiv

% ------------------------------------------------------------------------------
\subsection{Additional Memory Requirement}\label{ss:inplace property}

In this section, we show that \compiparassssort can be implemented either strictly in-place if the local task stack is implicitly represented or in-place if the tasks are stored on the recursion stack.

\begin{theorem}\label{thm:space}
  \compiparassssort\ can be implemented with $\Oh{kb}$ additional memory per thread.
\end{theorem}

\begin{proof}[Proof of \cref{thm:space}]
  Each thread has a space overhead of two swap buffers and $k$ buffer blocks of size $b$ (in total $\Oh{kb}$).
    This bound also covers smaller amounts of memory required for the partitioning steps.
  A partitioning step uses a search tree ($\Oh{k}$), an overflow buffer ($\Oh{b}$),
  read and write pointers ($\Oh{kB}$ if we avoid false sharing), end pointers, and bucket boundary pointers ($\Th{k}$ each).
  All of these data structures can be used for all recursion levels.

  The classification phase stores elements only in the buffer blocks and the overflow buffer.
  As each thread reads its elements sequentially into its buffer blocks, there is always an empty block in the input array when a buffer block is flushed.
  When the size of the input array is not a multiple of the block size, a single overflow buffer may be required to store the overflow.
  The permutation phase only requires the swap buffers and the read and write pointers to move blocks into their target bucket.
  In the sampling phase, we do not need extra space as we swap the sample to the front of the input array.
  Nor do we need the local stacks (each of size $\Ohbp{k\log_k \frac{n}{n_0}}$) since we can use an implicit representation of the sequential tasks as described in \cref{ss:strictly}.
\end{proof}

\begin{theorem}\label{thm:space stack}
  With a local stack, \compiparassssort\ can be implemented with $\Ohbp{tk\log_k \frac{n}{n_0}}$ additional memory per thread.
\end{theorem}

\begin{proof}[Proof of \cref{thm:space stack}]
Each recursion level stores at most $k$ tasks on the local stack.
 Only $\Ohbp{\log_k \frac{n}{n_0}}$ levels of parallel recursion are needed to get to the base cases with a probability of at least $1-n_0/n$ (see~\cref{thm:ips4olevel}).
 In the rare case that the memory is exhausted, the algorithm is restarted.
\end{proof}

% ------------------------------------------------------------------------------
\subsection{I/O Complexity}\label{ss:io efficiency}

Apart from the local work, the main issue of a sorting algorithm is the number of accesses to the main memory.
In this section, we analyze this aspect in the PEM model.
First, we show that \compiparassssort is I/O-efficient if the constraints we state at the beginning of this chapter apply. Then, we discuss how the I/O efficiency of \compiparassssort relates to practice.

\begin{theorem}\label{thm:io}
  \compiparassssort\ has an I/O complexity of
  $\Oh{\frac{n}{tB}\log_k\frac{n}{M}}$ memory block transfers with a probability of at least $1 - M/n$.
\end{theorem}

Before we prove \cref{thm:io}, we prove that sequential partitioning steps exceeding the private cache are I/O-efficient (\cref{lem:io seq}) and that parallel partitioning steps are I/O-efficient (\cref{lem:io par}).

\begin{lemma}\label{lem:io seq}
  A sequential partitioning task with $n'\in \Om{M}$ elements transfers $\Th{n'/B}$ memory blocks.
\end{lemma}

\begin{proof}
  A sequential partitioning task performs a partitioning step with one thread.
  The \emph{sampling phase} of the partitioning step requires $\Th{k\log k}$ I/Os for sorting the random sample (\cref{as:2}).
  We have $k\log k\in\Oh{n'/B}$ as $M \in \Om{Bk\log k}$ (\cref{as:9}).
  During the \emph{classification phase}, the thread reads $\Oh{n'}$ consecutive elements, writes them to its local buffer blocks, and eventually moves them blockwise back to the main memory.
  This requires $\Oh{n'/B}$ I/Os in total.
  As $M\in\Om{kb}$, the local buffer blocks fit into the private cache.
  The same asymptotic cost occurs for moving blocks during the \emph{permutation phase}.
  In the \emph{cleanup phase}, the thread has to clean up $k$ buckets.
  To clean up bucket $i$, the thread moves the elements from buffer block $i$ and, if necessary, elements from a block which overlaps into bucket $i+1$ to bucket boundary $i$.
  The elements from these two blocks are moved consecutively.
  We can amortize the transfer of full memory blocks with the I/Os from the classification phase as these blocks have been filled in the classification phase.
  We account $\Oh{1}$ I/Os for potential truncated memory blocks at the ends of the consecutive sequences.
  For $k$ bucket boundaries, this sums up to $\Oh{k}\in \Oh{n'/B}$ as $n' \in \Om{M}\in\Om{Bk}$ (\cref{as:1,as:9}).
\end{proof}

\begin{lemma}\label{lem:io par}
  A parallel task with $\Th{t'n/t}$ elements and $t'$ threads transfers $\Th{\frac{n}{tB}}$ memory blocks per thread.
\end{lemma}

\begin{proof}
  A parallel task performs a partitioning step.
  The \emph{sampling phase} of the partitioning step requires $\Oh{k\log k}$ I/Os for loading the random sample (\cref{as:2}).
  We have $k\log k\in\Oh{\frac{n}{tB}}$ as $n/t\in\Om{Bk\log k}$ (\cref{as:3,as:9}).
  During the \emph{classification phase}, each thread reads $\Oh{n/t}$ consecutive elements, writes them to its local buffer blocks, and eventually moves them blockwise back to the main memory.
  As $M\in\Omega(kb)$, the local buffer blocks fit into the private cache.
    In total, the classification phase transfers $\Oh{\frac{n}{tb}}$ logical data blocks causing $\Oh{\frac{n}{tB}}$ I/Os per thread.
  
  The same asymptotic cost occurs for moving blocks during the \emph{permutation phase}.
  Each thread performs $\Oh{\frac{n}{tb}}$ successful acquisitions of the next block in a bucket.
  The successful acquisitions require $\Oh{\frac{n}{tb}}$ reads and writes of the read pointers $r_i$ and the write pointers $w_i$ -- for each read and write, we charge $\Oh{t}$ I/Os for possible contention with other threads.
  Thus, the successful acquisitions sum up to $\Oh{t\cdot\frac{n}{tb}}\in \Oh{\frac{n}{tB}}$ I/Os (\cref{as:1}).
 For the block permutations of $\Oh{\frac{n}{tb}}$ elements, we get the overall cost of $\Oh{\frac{n}{tB}}$ I/Os.
 Furthermore, an additional block is loaded for each of the $k$ buckets to recognize that the bucket does not contain any unprocessed blocks.
   Similar to the successful acquisitions we charge an overall cost of $\Oh{kt}$ I/Os.
   Since $n/t\in\Om{bk}$ (\cref{as:3,as:9}), we have $k\in\Ohbp{\frac{n}{t^2B}}$ and hence $\Oh{kt}\in\Oh{\frac{n}{tB}}$.

  In the \emph{cleanup phase}, $t'$ threads have to clean up $k$ buckets.
  To clean up a single bucket, elements from $t'+2$ buffer blocks and bucket boundaries are moved.
  This takes $\Oh{t'b/B}$ I/Os for cleaning a bucket.
  We consider a case distinction with respect to $k$ and $t'$.
  If $k\leq t'$, then each thread cleans at most one bucket.
  This amounts to a cost of $\Oh{t'b/B}\in\Oh{\frac{n}{tB}}$ since $n/t\in \Om{tb}$ (\cref{as:3}).
  If $k>t'$, then each thread cleans $\Oh{k/t'}$ buckets with a total cost
  of $\Oh{k/t'\cdot t'b/B}\in\Oh{kb/B}$ I/Os.
  We have $\Oh{kb/B}\in\Oh{\frac{n}{tB}}$ since $\Th{n/t}\in\Om{kb}$ (\cref{as:3,as:9}).
\end{proof}

Now, we can prove that \compiparassssort is I/O-efficient if the constraints we state at the beginning of this chapter apply.

\begin{proof}[Proof of \cref{thm:io}]
  In this proof, we can assume that \compiparassssort\ performs $\Oh{\log_k\frac{n}{M}}$ recursion levels until the tasks have at most $M$ elements.
  According to \cref{thm:ips4olevel}, this assumption holds with a probability of at least $1 - M/n$.
  We do not consider the situation of many identical keys since the elements with these identical keys will not be processed at later recursion levels anymore.
  
  From \cref{thm:space} we know that \compiparassssort uses additional data structures that require $\Oh{kb}$ additional memory.
  In addition to the accesses to these data structures, a task $T[l,r)$ only accesses $\VarArray\oset{l}{r-1}$.
  As $M\in\Om{bk}$ (\cref{as:9}) we can keep the additional data structures in the private cache.
  Thus, we only have to count the memory transfers of tasks from and to the input array.

  In this analysis, we consider a case distinction with respect to the task type, its size, and the size of its parent task.
  Each configuration requires at most $\Oh{\frac{n}{tB}\log_k \frac{n}{M}}$ I/Os per thread.

  \emph{Parallel tasks:}
  \compiparassssort\ processes the parallel tasks first.
  Parallel tasks transfer $\Th{\frac{n}{tB}}$ memory blocks per thread (see \cref{lem:io par}).
  As a thread performs at most one parallel task on each recursion level (see \cref{lem:par task unique level}), parallel tasks on the first $\Oh{\log_k\frac{n}{M}}$ recursion levels perform $\Oh{\frac{n}{tB}\log_k \frac{n}{M}}$ I/Os per thread.
  On subsequent recursion levels, no parallel tasks are executed:
  After $\Oh{\log_k\frac{n}{M}}$ recursion levels, the size of tasks is at most $M$.
  However, parallel tasks have more than $M$ elements.
  This follows from \cref{lem:par task range} and \cref{as:3}.

  \emph{Large tasks (Sequential partitioning task with $\om{M}$ elements):}
  A large task with $n'$ elements takes $\Th{n'/B}$ I/Os (see \cref{lem:io seq}).
  A thread processes sequential tasks covering a continuous stripe of $\Oh{n/t}$ elements of the input array (see \cref{lem:seq task range}).
  Thus, the large tasks of a thread transfer $\Oh{\frac{n}{tB}}$ memory blocks on each recursion level.
  This sums up to $\Oh{\frac{n}{tB}\log_k\frac{n}{M}}$ I/Os per thread for the first $\log_k\frac{n}{M}$ recursion levels.
After $\Oh{\log_k\frac{n}{M}}$ recursion levels, the size of tasks fits into the main memory, i.e., their size is $\Oh{M}$.

  \emph{Small tasks (Sequential tasks containing $\Oh{M}\cap \Om{B}$ elements with parent tasks having $\om{M}$ elements or with parallel parent tasks having $\Om{M}$ elements):}
  In the first step of small tasks, the classification phase, the thread reads the elements of the task from left to right.
  As the task fits into the private cache, the task does not perform additional I/Os after the classification phase.
  For a small task of size $n'$, we have $\lfloor n'/B\rfloor$ I/Os as $n'\in\Om{B}$.
  Buckets of sequential partitioning tasks are sequential subtasks which again have $\Oh{M}$ elements.
  Thus, each input element is only once part of a small task and small tasks cover disjoint parts of the input array.
  Additionally, we know from \cref{lem:seq task range} that the sequential tasks of a thread contain $\Oh{n/t}$ different elements.
  From this follows that a thread transfers $\Oh{\frac{n}{tB}}$ memory blocks for small tasks.

  \emph{Tiny tasks (Sequential tasks with $\Oh{B}$ elements whose parent tasks have $\om{M}$ elements or with parallel parent tasks having $\Om{M}$ elements):}
  A tiny task needs $\Oh{1}$ I/Os.
  We account these I/Os to its parent task.
  A parent task gets at most $\Oh{k}$ additional I/Os in the worst-case, $\Oh{1}$ for each bucket.
  The parent task has $\Om{tBk}$ elements:
  By definition, the parent task has $\Om{M}$ elements and we have $M\in \Om{tBk}$ (\cref{as:1,as:9}).
  We have already accounted $\Om{tBk/B}$ I/Os ($\Om{Bk/B}$ I/Os) for this sequential (parallel) parent task previously (see I/Os of large tasks and parallel tasks).
  Thus, the parent task can amortize the I/Os of its tiny subtasks.

  \emph{Middle tasks (Sequential tasks with sequential parent tasks containing $\Oh{M}$ elements):}
  Let the middle task $T[l_s, r_s)$ processed by thread $i$ be a bucket of a sequential task $T[l,s)$ contained $\Oh{M}$ elements.
  When thread $i$ processed task $T[l, r)$, the subarray $\VarArray\oset{l}{r-1}$ was loaded into the private cache of thread $i$.
  As the thread processes the sequential tasks from its local stack in depth-first search order,  $\VarArray\oset{l}{r-1}$ remains in the thread's private cache until the middle task $T[l_s, r_s)$ is executed.
  The middle task does not require any memory transfers -- it only accesses $\VarArray\oset{l_s}{r_s-1}$ which is a subarray of $\VarArray\oset{l}{r-1}$.
\end{proof}

\ifarxiv In \cref{app:io volume analysis} \else In the extended version of this paper~\cite{axtmann2020ips4oarxiv}\fi, we analyze the constant factors of
the I/O volume (i.e., data flow between cache and main memory) for the
sequential algorithms \compissssort~(\compiparassssort\ with $t=1$)
and \compssssort. To simplify the discussion, we assume a single recursion level,
$k=256$ and $8$-byte elements.  We show that
\compissssort needs about $48n$ bytes of I/O volume, whereas
\compssssort needs between $67n$ and $84n$, depending on whether we
use a conservative calculation or not.  This is surprising since, at
the first glance, the partitioning algorithm of \compissssort\ writes
the data twice, whereas \compssssort does this only once.  However,
this is more than offset by ``hidden'' overheads of \compssssort like
memory management, allocation misses, and associativity misses.

\subsection{Branch Mispredictions and Local Work}\label{ss:total work}

Besides the latency of loading and writing data, which we analyze in the previous section, branch mispredictions and the (total) work of an algorithm can limit its performance.
In the next paragraph, we address branch mispredictions of \compiparassssort and afterwards, we analyze the total work of \compiparassssort.

Our algorithm \compiparassssort has virtually no branch mispredictions during element classification.
The main exception is when the algorithm detects that a bucket has to be flushed.
A bucket is flushed on average after $b$ element classifications (after $b\log k$ element comparisons).

We now analyze the local work of \compiparassssort.
We neglect delays introduced by thread synchronizations and accesses to shared variables as we accounted for those in the I/O analysis in the previous section.
For the analysis, we assume that the base case algorithm performs partitioning steps with $k=2$ and a constant number of samples until at most one element remains.
Thus, the local work of the base case algorithm is identical to the local work of quicksort.
We also assume that the base case algorithm is used to sort the samples.

Actually, our implementation of \compiparassssort sorts the base cases with insertion sort.
The reason is that a base case with more than $2n_0$ elements can only occur if its parent task has between $n_0$ and $kn_0$ elements.
In this case, the average base case size is between $0.5n_0$ and $n_0$.
Our experiments have shown that insertion sort is more efficient than quicksort for these small inputs.
Also, base cases with much more than $2n_0$ elements are very rare.
To sort the samples, our implementation recursively invokes \compiparassssort.

\begin{theorem}\label{thm:local work}
  When using quicksort to sort the base cases and the samples, \compiparassssort\ has a local work of $\Oh{n/t\log n}$ with a probability of at least $1-n^{-4}$.
\end{theorem}

For the proof of \cref{thm:local work}, we need \crefrange{lem:work single partitioning step}{lem:work base cases}.
These lemmas use the term \emph{small task} for tasks with at most $kn_0$ elements and the term \emph{large task} for tasks with at least $kn_0$ elements.

\begin{lemma}\label{lem:work single partitioning step}
  A partitioning task with $n'$ elements and $t'$ threads has a local work of $\Oh{n'/t'\log k}$ excluding the work for sorting the samples.
\end{lemma}

\begin{proof}
  In the \emph{classification phase} of the partitioning step, the comparisons in the branchless decision tree dominate.
  Each thread classifies $n/t'$ elements with takes $\Oh{\log k}$ comparisons each.
  This sums up to $\Oh{n'/t'\log k}$.
  The element classification dominates the remaining work of this phase, e.g., the work for loading each element once, and every $b$ elements, the work for flushing a local buffer.

  In the \emph{permutation phase}, each block in the input array is swapped into a buffer block once and swapped back into the input array once.
  As each thread swaps at most $\lfloor \frac{n'}{tb}\rfloor$ blocks of size $b$, the phase has $\Oh{n'/t}$ local work.

  When the \emph{cleanup phase} is executed sequentially, the elements of the local buffers are flushed into blocks that overlap into the next bucket.
  This may displace elements stored in these blocks.
  The displaced elements are written into empty parts of blocks.
  Thus, each element is moved at most once which sums up to $\Oh{n'}$ work.
  For the cleanup phase of a parallel partitioning step, we conclude from the proof of \compiparassssort's I/O complexity (see \cref{thm:io}) that the local work is in $\Oh{n'/t'}$:
  The proof of the I/O complexity shows that a parallel cleanup phase is bounded by $\Oh{\frac{n'}{t'B}}$ I/Os.
  Also, each element that is accessed in the cleanup phase is moved at most once and no additional work is performed.
  We account a local work of $B$ for each memory block which a thread accesses.
  Thus, we can derive from $\Oh{\frac{n'}{t'B}}$ I/Os a local work of $\Oh{n'/t'}$ for the parallel cleanup phase.
\end{proof}

\begin{lemma}\label{lem:num parallel small tasks}
  At most one parallel task which is processed by a thread is a small task.
\end{lemma}

\begin{proof}
  Assume that a thread processes at least two small parallel tasks $p_1$ and $p_2$.
  According to \cref{lem:par task unique level}, a thread processes at most one of these tasks per recursion level.
  A parallel subtask of thread $i$ is a subtask of a parallel task of thread $i$ and represents a bucket of this task (see \cref{lem:par task is bucket}).
  Thus, $p_1$ and $p_2$ must be on different recursion levels, and $p_1$ processes a subset of elements processed by $p_2$ or vice versa.
  However, this is a contradiction as buckets of small tasks are base cases.
\end{proof}

\begin{lemma}\label{lem:local work small task}
  The local work for all small partitioning tasks is in total $\Oh{n/t\log k}$ excluding the work for sorting the samples.
\end{lemma}

\begin{proof}
  In this proof, we neglect the work for sorting the sample.
  \Cref{lem:work single partitioning step} tells us that a partitioning task with $n'$ elements and $t'$ threads has a local work of $\Oh{n'/t'\log k}$.
  Also, parallel tasks with $t'$ threads have $\Th{t'n/t}$ elements (see \cref{lem:par task range}).
  Thus, we have $\Oh{n/t\log k}$ local work for a small parallel task.
  Overall, we have $\Oh{n/t\log k}$ local work for small parallel tasks as each thread processes at most one of these tasks (see \cref{lem:num parallel small tasks}).

  We now consider small sequential partitioning tasks.
  small sequential partitioning tasks that are processed by a single thread cover in total at most $\Oh{n/t}$ different elements (see \cref{lem:seq task range}).
  Each of these elements passes at most one of the small partitioning tasks since buckets of these tasks are base cases.
  Thus, a thread processes small partitioning tasks of size $\Oh{n/t}$ in total.
  As small partitioning tasks with $n'$ elements require $\Oh{n'\log k}$ local work (see \cref{lem:work single partitioning step}), the local work of all small partitioning tasks is $\Oh{n/t\log k}$.
\end{proof}

\begin{lemma}\label{lem:local work one rec level}
  The partitioning tasks of one recursion level require $\Oh{n/t\log k}$ local work excluding the work for sorting the samples.
\end{lemma}

\begin{proof}
  \Cref{lem:work single partitioning step} tells us that a partitioning task with $n'$ elements and $t'$ threads has a local work of $\Oh{n'/t'\log k}$ excluding the work for sorting the samples.
  Parallel tasks with $t'$ threads have $\Th{t'n/t}$ elements (see \cref{lem:par task range}).
  Thus, we have $\Oh{n/t\log k}$ local work on a recursion level for parallel tasks.
  A thread processes sequential partitioning tasks covering a continuous stripe of $\Oh{n/t}$ elements of the input array (see \cref{lem:seq task range}).
  Thus, we also have $\Oh{n/t\log k}$ local work on a recursion level for sequential partitioning tasks.
\end{proof}

\begin{lemma}\label{lem:work all partitioning steps}
  All large partitioning tasks have in total $\Oh{n/t\log n}$ local work with a probability of at least $1-n^{-1}$ excluding the work for sorting the samples.
\end{lemma}

\begin{proof}
  In this proof, we neglect the work for sorting the samples of a partitioning task.
  Large partitioning tasks create between $\sqrt{k}$ and $k$ buckets (see \cref{sec:sampling}).
  According to \cref{thm:ips4olevel}, \compiparassssort\ performs at most $\Thbp{\log_{\sqrt{k}}\frac{n}{kn_0}} = \Thbp{\log_k\frac{n}{kn_0}}$ recursion levels with a probability of at least $1-kn_0/n$ until all partitioning tasks have less than $kn_0$ elements.
  However, this probability is not tight enough.
  Instead, we can perform up to $\Oh{\log_k n}$ recursion levels and still have $\Oh{n/t\log n}$ local work as each recursion level requires $\Oh{n/t\log k}$ local work (see \ref{lem:local work one rec level}).
  In this case, \cref{thm:ips4olevel} tells us that all partitioning tasks have at most one element with a probability of $1-n^{-1}$.
  This probability also holds if we stop partitioning buckets with less than $kn_0$ elements.
\end{proof}

\begin{lemma}\label{lem:work base cases}
  The local work of all base case tasks is in total $\Oh{n/t\log n}$ with a probability of at least $1-n^{-1}$.
\end{lemma}

\begin{proof}
  Sorting $\Oh{n}$ elements with the base case algorithm quicksort does not exceed $\Oh{\log n}$ recursion levels with probability $1-n^{-1}$~\cite{jaja2000perspective}.
  Thus, an execution of quicksort with $\Oh{n/t}$ elements requires $\Oh{n/t\log n}$ local work with a probability of at least $1-n^{-1}$ as it would also not exceed $\Oh{\log n}$ recursion levels with at least the same probability.
  Sorting all base case of a thread is asymptotically at least as efficient as sorting $\Oh{n/t}$ elements at once:
  The base cases have in total at most $\Oh{n/t}$ elements (see \cref{lem:seq task range}) but the input is already prepartitioned.
\end{proof}

\begin{lemma}\label{lem:local work 2n0}
  The local work for sorting the samples of all small partitioning tasks is $\Oh{n/t\log n}$ in total with a probability of at least $1-n^{-1}$
\end{lemma}

\begin{proof}
  Small partitioning tasks with $n'$ elements have $n'/n_0$ buckets and a sample of size $\Ohbp{n'/n_0 \log \frac{n'}{n_0}}$ (see \cref{as:2} for the oversampling factor).
  The size of a sample is in particular bounded by $\Oh{n'}$.
  For this, we use $n_0 \in \Om{\log k}$ (\cref{as:6}) and $k \geq n'/n_0$ from which follows that $n_0 \in \Ombp{\log \frac{n'}{n_0}}$ holds.
  Furthermore, small sequential partitioning tasks which are processed by a single thread cover in total at most $\Oh{n/t}$ different elements (see \cref{lem:seq task range}).
  Thus, the total size of all samples from small sequential partitioning tasks sorted by a thread is limited to $\Oh{n/t}$.
  We count one additional sample from a potential small parallel task (see \cref{lem:num parallel small tasks}).
  We can also limit its size to $\Oh{n/t}$ using \cref{as:3,as:6}.
  We can prove that the work for sorting samples of a total size of $\Oh{n/t}$ is in $\Oh{n/t\log n}$ with a probability of at least $1-n^{-1}$.
  We refer to the proof of \cref{lem:work base cases} for details.
\end{proof}

\begin{lemma}\label{lem:work sampling}
  The local work for sorting the samples of all large partitioning tasks is $\Oh{n/t\log n}$ in total with a probability of at least $1-n^{-1}$
\end{lemma}

\begin{proof}
  A sequential large partitioning task with $\Om{kn_0}$ elements has $\Om{k\log k}$ elements (see \cref{as:6}) and a parallel large partitioning task has $\Om{t'n/t}$ elements (see \cref{lem:par task range}) with $n/t\in \Om{k\log k}$ which is a result from \cref{as:3,as:9}.
  Thus, a large partitioning task with $t'$ threads has $\Om{t'k\log k}$ elements.
  From \cref{lem:work single partitioning step} follows that a large partitioning task excluding the work for sorting the samples has $\Om{k\log^2 k}$ local work.

  Each thread invokes at most $r=l\frac{n\log n}{tk\log^2 k}$ large partitioning tasks for a constant $l$:
  These tasks require $\Om{k\log^2 k}$ local work each and all partitioning tasks performed by a single thread require $\Oh{n/t\log n}$ local work in total (see ~\cref{lem:work all partitioning steps}) -- excluding the work for sorting the samples.
  Thus, when we execute large partitioning tasks, each thread performs at most $r$ sample sorting routines -- one for each task.

  For the local work analysis, we consider a modified sample sorting algorithm.
  Instead of using quicksort, we use an adaption that restarts quicksort when it exceeds $\Oh{k\log^2k}$ local work until the sample is finally sorted.
  The bounds for the local work which we obtain from this variant also hold when \compiparassssort executes quicksort until success instead of restarting the algorithm:
  A restart means to neglect the prepartitioned buckets which makes the problem unnecessarily difficult.

  For the sample sorting routines of large partitioning tasks, each thread can spend $\Oh{rk\log^2 k}$ local work in total.
  As we restart quicksort after $\Oh{k\log^2k}$ local work, we can amortize even $xr$ (re)starts of quicksort for any constant $x$.
  We show that $xr$ (re)starts are sufficient to successfully sort $r$ samples with a probability of at least $1-n^{-1}$.

  We observe that one execution of quicksort unsuccessfully sorts a sample with a probability of at most $p=k^{-3}\log_2^{-3}k$ as the size of the samples is bounded by $\Oh{k\log k}$.
  For this approximation, we use that sorting $n$ elements with quicksort takes $\Oh{n\log n}$ work with high probability~\cite{jaja2000perspective}.
  Each execution of quicksort is a Bernoulli trial as we have exactly two possible outcomes, ``successful sorting in time'' and ``unsuccessful sorting in time'', and the probability of failure is bounded by $p$ each time.
  When we consider all quicksort invocations of \compiparassssort, we need $r$ successes.
  We define a binomial experiment which repeatedly invokes quicksort on the sample of the first large partitioning task until success and then continues with the second large partitioning task, until the sample of each partitioning step of a thread is sorted.
  Asymptotically, we can spend $xr$ (re)starts of quicksort for any constant $x\geq 1$ such that the binomial experiment does not exceed $\Oh{n/t\log n}$ local work.

  Let the random variable $X$ be the number of unsuccessful sample sorting executions and assume that $x\geq 2$.
    Then, the probability $I$
  \begin{equation}
    \begin{aligned}
      I &= \probability{X > (x-1)r}\\
      &\leq \sum_{j>(x-1)r} \binom{xr}{j}p^j\left(1-p\right)^{xr - j}
      \leq \sum_{j>(x-1)r} \left(\frac{xre}{j}\right)^jp^j\\
      &\leq \sum_{j>(x-1)r} \left(\frac{xre}{(x-1)r}\right)^jp^j
      = \sum_{j>(x-1)r} \left(\frac{pex}{x-1}\right)^j\\
      &= \frac{\left(\frac{pex}{x-1}\right)^{(x-1)r + 1}}{1 - \frac{pex}{x-1}}\\
      &\leq \left(\frac{pex}{x-1 - pex}\right)\left(\frac{pex}{x-1}\right)^{\frac{(x-1)lk\log k\log n}{k\log^2 k}}\\
      &\leq \left(\frac{pex}{x-1 - pex}\right)\left(n\right)^{\frac{(x-1)l\log \left(\frac{pex}{x-1}\right)}{\log k}}\\
      &\leq 2.13n^{-\frac{1}{2}\left(x-1\right)}
    \end{aligned}
  \end{equation}
  defines an upper bound of the probability that $xr$ (re)starts of the the sample sorting algorithm execute less than $r$ successful runs.
  The second ``$\leq$'' uses $\binom{n}{k} \leq \left(en/k\right)^k$, the third ``$=$'' uses $\sum_{k=n}^\infty r^k= \frac{r^n}{1-r}$, derived from the geometric series, the second ``$\leq$'' and the third ``$=$'' use $\frac{pex}{x-1} < 1$, and the last ``$\leq$'' uses $\frac{\log \left(\frac{pex}{x-1}\right)}{\log k}<-1/2$ and $\frac{pex}{x-1 - pex}<2.13$.
\end{proof}

\begin{proof}[Proof of~\cref{thm:local work}]
  According to \cref{lem:local work small task}, the small partitioning tasks excluding the work for sorting the samples require $\Oh{n/t\log k}$ local work.
  For large partitioning tasks, we have $\Oh{n/t\log n}$ local work with a probability of at least $1-n^{-1}$ (see \cref{lem:work all partitioning steps}).
  The same holds for the base cases (see \cref{lem:work base cases}).
  \Cref{lem:local work 2n0,lem:work sampling} bound the local work for sorting the samples of small and large partitioning tasks to $\Oh{n/t\log n}$, each with a probability of at least $1-n^{-1}$.
  This sums up to a total local work of $\Oh{n/t\log n}$ with a probability of at least  $1-4/n$.
\end{proof}

\section{Implementation Details}\label{s:details}

\compiparassssort has several parameters that can be used for tuning and adaptation.
We performed our experiments using (up to) $\VarBucketCount=256$ buckets,
an oversampling factor $\OversamplingFactor=0.2 \log n$,
a base case size $\BaseCaseSize=16$ elements,
and a block size of $\BlockSize=\max(1, 2^{\lfloor11 - \log_2 D\rfloor})$ elements, where $D$ is the size of an element in bytes (i.e., about 2\;KiB).
In the sequential case, we avoid the use of atomic operations on pointers and we use the recursion stack to store the tasks.
On the last level, we perform the base case sorting immediately after the bucket
has been completely filled in the cleanup phase, before processing the other buckets.
This is more cache-friendly, as it eliminates the need for another sweep over the data.
Furthermore, we use insertion sort as the base case sorting algorithm.

\compiparassssort (\compissssort) detects sorted inputs.
If these inputs are detected, our algorithm reverses the input in the case that the input was sorted in decreasing order, and returns afterward.
Note that such heuristics for detecting ``easy'' inputs are quite common \cite{obeya2019theo,reinders2007intel}.

For parallelization, we support OpenMP or \texttt{std::thread} transparently.
If the application is compiled with OpenMP support, \compiparassssort employs the existing OpenMP threads.
Otherwise, \compiparassssort uses C++ threads and determines $t$ by invoking the function \texttt{std::thread::hardware\_concurrency}.
Additionally, the application can use its own custom threads to execute \compiparassssort.
For that, the application creates a thread pool object provided by \compiparassssort, adds its threads to the thread pool, and passes the thread pool to \compiparassssort.

We store each read pointer~$\readblock[i]$ and its corresponding write
pointer~$\readblock[i]$ in a single $128$-bit word which we read and
modify atomically.  We use $128$-bit atomic compare-and-swap
operations from the GNU Atomic library \texttt{libatomic} if the CPU
supports these operations.  Otherwise, we guard the pointer pair with
a mutex.  We did not measure a difference in running time between these
two approaches except for some very special corner cases.
We align the thread-local data to 4\;KiB which is typically the memory page size in systems.
The alignment avoids false sharing and simplifies the migration of memory pages if a thread is moved to a different NUMA node.

We decided to implement our own version of the (non-in-place) algorithm \compssssort from Sanders and Winkel~\cite{sanders2004super}.
This has two reasons.
First, the initial implementation is only of explorative nature, e.g., the implementation does not handle duplicate keys, and, the implementation is highly tuned for outdated hardware architectures.
Second, the reimplementation \compssssschneider~\cite{Lorenz2016ssss}
seemed to be unreasonably slow.
We use \compiparassssort as an algorithmic framework to implement \compmyparassssaxtmann, our version of \compssssort.
For \compmyparassssaxtmann, we had to implement the three main parts of \compssssort: (1) the partitioning step, (2) the decision tree, and (3), the base case algorithm and additional parameters of the algorithm, e.g., for the maximum base case size, the number of buckets, and the oversampling factor.
We replace the partitioning step of \compiparassssort with the one described by Sanders and Winkel.
For the element classification, we reuse the branchless decision tree of \compiparassssort.
We also reuse the base case algorithm and the parameters from \compiparassssort which seem to work very well for \compmyparassssaxtmann.
As we use \compiparassssort as an algorithmic framework, we can execute \compmyparassssaxtmann in parallel or with only one thread.
If we refer to \compmyparassssaxtmann in its sequential form, we use the term \compmyssssaxtmann.

We also use \compiparassssort as an algorithmic framework to implement \emph{In-place Parallel Super Scalar Radix Sort} (\compiparassrsort).
For \compiparassrsort, we replaced the branchless decision tree of \compiparassssort with a simple radix extractor function that accepts unsigned integer keys.
For tasks with less than $2^{12}$ elements, we use \radixsska as a base case sorting algorithm.
For small inputs ($n \leq 2^7$), \radixsska then falls back to quicksort which again uses insertion sort for $n\leq 2^5$.
If we refer to \compiparassrsort in its sequential form, we use the term \compissrsort.
\compiparassrsort only sorts data types with unsigned integer keys. The author of \radixsska~\cite{skarupke2016github,skarupke2016skasort} demonstrates
that a radix sorter can be extended to sort inputs with
floating-point keys and even compositions of primitive data types.  We
note that \compissrsort can be extended to sort these data types as
well.

Our algorithms \compiparassssort, \compiparassrsort , and  \compmyparassssaxtmann are written in C++ and the implementations can be found on the official website \url{https://github.com/ips4o}.
The website also contains the benchmark suite used for this publication and a description of how the experiments can be reproduced.

\section{Experimental Results}\label{s:experiments}

In this section, we present results from ten data distributions,
generated for four different data types obtained on four different
machines with one, two, and four processors and $21$ different sorting algorithms.  We extensively compare
our in-place parallel sorting algorithms \compiparassssort and \compiparassrsort as well as their
sequential counterparts \compissssort and \compissrsort to various
competitors\footnote{Since several algorithms were implemented by
  third parties, we may cite their publication and implementation
  separately.}:

\begin{itemize}
\item {\bf Parallel in-place comparison-based sorting}
    \begin{itemize}
    \item \comppbalancedsort (OpenMP): Two
      implementations (balanced and unbalanced) from the GCC
      STL~library~\cite{putze2007mcstl} based on quicksort proposed by
      Tsigas~and~Zhang~\cite{TsiZha03}.
    \item \compptbb~\cite{wenzel2019tbb} (TBB): Quicksort from the
      Intel\textregistered~TBB library~\cite{reinders2007intel}.
    \end{itemize}
\item {\bf Parallel non-in-place comparison-based sorting}
    \begin{itemize}
    \item \compppbbs~\cite{shun2012brief} (Cilk): $\sqrt{n}$-way
      samplesort~\cite{blelloch2010low} implemented in the so-called problem based
      benchmark suite.
    \item \comppsort (OpenMP): Stable multiway mergesort from the GCC
      STL~library~\cite{putze2007mcstl}.
    \item \compmyparassssaxtmann~\cite{GithubMichael2020ssss} (OpenMP or
      \texttt{std::thread}): Our parallel and stable implementation of
      \compssssort from \cref{s:details}.
    \item \comppaspas~\cite{synergy2019aspas} (POSIX Threads): Stable
      mergesort which vectorizes the merge function with
      AVX2~\cite{hou2015aspas}. \comppaspas only sorts \texttt{int},
      \texttt{float}, and \texttt{double} inputs and uses the
      comparator function~``$<$''.
    \end{itemize}
\item {\bf Parallel in-place radix sort}
  \begin{itemize}
    \item \radixregion~\cite{obeya2019github} (Cilk): \emph{Most
        Significant Digit} (MSD) radix sort~\cite{obeya2019theo} which
      only sorts keys of unsigned integers. \radixregion skips the
      most significant bits which are zero for all keys.
    \item \imsdradix~\cite{GithubOrestis2020comprehensive} (POSIX
      Threads): MSD radix sort~\cite{polychroniou2014comprehensive}
      from Orestis~and~Ross with blockwise redistribution.  The
      implementation, published by Orestis~and~Ross, however, requires
      $20$\,\% of additional memory in addition to the input array and
      is very explorative.
    \end{itemize}
\item {\bf Parallel non-in-place radix sort}
    \begin{itemize}
    \item \radixppbbr~\cite{shun2012brief} (Cilk): A simple
      implementation of stable MSD radix sort from the so-called problem based
      benchmark suite.
    \item \radixraduls~\cite{refresh2017raduls2}
      (\texttt{std::thread}): MSD radix sort which uses non-temporal
      writes to avoid write allocate misses~\cite{kokot2017even}.
      \radixraduls requires $256$-bit array alignment. Both keys and
      objects have to be aligned at 8-byte boundaries. We partially
      compiled the code with the flag \texttt{-O1} as recommended by
      the authors. The algorithm does not compile with the Clang compiler.
    \end{itemize}
\item {\bf Sequential in-place comparison-based sorting}
    \begin{itemize}
    \item \compblock~\cite{edelkamp2016github}: An implementation of
      BlockQuicksort~\cite{edelkamp2016blockquicksort} provided by the
      authors of the sorting algorithm publication.
    \item \compspdq~\cite{peters2015pdq}: Pattern-Defeating Quicksort
      which integrated the approach of BlockQuicksort in
      2016. \compspdq has similar running times as BlockQuicksort
      using Lomuto's Partitioning~\cite{aumuller2018simple}, published
      in 2018.
    \item \compsyaros~\cite{edelkamp2016yaroslavskiy}: A \texttt{C++}
      port of Yaroslavskiy's Dual-Pivot
      Quicksort~\cite{yaroslavskiy2009dual}.  Yaroslavskiy's
      Dual-Pivot Quicksort is the default sorting routine for
      primitive data types in Oracle's Java runtime library since
      version~7.
    \item \compssort: Introsort from the GCC STL~library.
    \item \compswiki~\cite{mcfadden2014wikisort}: An implementation of
      stable in-place mergesort~\cite{pokson2008patio}.
    \end{itemize}
\item {\bf Sequential non-in-place comparison-based sorting}
    \begin{itemize}
    \item \compstim~\cite{goro2011timsort}: A \texttt{C++} port of
      Timsort~\cite{peters2015timsort}. Timsort is an implementation
      of stable mergesort which takes advantage of presorted sequences
      of the input. \compstim is part of Oracle's Java runtime
      library since version~7 to sort non-primitive data types.
    \item \compssssschneider~\cite{Lorenz2016ssss}: A recent
      implementation of non-in-place
      \compssssort~\cite{sanders2004super} optimized for modern
      hardware.
    \item \compmyssssaxtmann~\cite{GithubMichael2020ssss}: Our
      implementation of \compssssort which we describe in
      \cref{s:details}.
    \end{itemize}
\item {\bf Sequential in-place radix sort}
    \begin{itemize}
    \item \radixsska~\cite{skarupke2016github}: MSD radix
      sort~\cite{skarupke2016skasort} which accepts a key-extractor
      function returning primitive data types or pairs, tuples,
      vector, and arrays containing primitive data types. The latter
      ones are sorted lexicographically.
    \end{itemize}
\item {\bf Sequential non-in-place radix sort}
    \begin{itemize}
    \item \radixipp~\cite{Intel2020Ipp}: Radix sort from the Intel®
      Integrated Performance Primitives library optimized with the
      AVX2 and AVX-512 instruction set.
    \end{itemize}
\item {\bf Sequential non-in-place sorting with machine learning models}
    \begin{itemize}
    \item \radixlearned~\cite{kristo2020caseimpl}: An algorithm~\cite{kristo2020case} for sorting numeric data
      using a learned representation of the cumulative key distribution
      function as a piecewise linear function. This can be viewed as a generalization of
      radix sort.
    \end{itemize}
\end{itemize}

\begin{figure}[!t]
\begin{minipage}[t]{0.48\linewidth}
  \begin{algorithm}[H]
      \caption{\quartet comparison}\label{alg:comp quartet}
    \begin{algorithmic}
      \Function{lessThan}{$l$, $r$}
      \If{$l.a \neq r.a$}
      \State \Return $l.a < r.a$
      \ElsIf{$l.b \neq r.b$}
      \State \Return $l.b < r.b$
      \Else~\Return $l.c < r.c$
      \EndIf
      \EndFunction
    \end{algorithmic}
  \end{algorithm}
\end{minipage}
\hfill
\begin{minipage}[t]{0.48\linewidth}
  \begin{algorithm}[H]
      \caption{\bytes comparison}\label{alg:com bytes}
    \begin{algorithmic}
      \Function{lessThan}{$l$, $r$}
      \For{$i\gets 0$ {\bf to} $9$}
      \If {$l.k[i] \neq r.k[i]$}
      \State \Return $l.k[i] < r.k[i]$
      \EndIf
      \EndFor
      \State \Return \textbf{False};
      \EndFunction
    \end{algorithmic}
  \end{algorithm}
\end{minipage}
\hfill
\end{figure}

We do not compare our algorithm to \radixparadis as its source code is
not publicly available.  However, Omar~et.~al.~\cite{obeya2019theo}
compare \radixregion to the numbers reported in the publication of
\radixparadis and conclude that \radixregion is faster than
\radixparadis.  Additionally, \radixregion and our algorithm have
stronger theoretical guarantees (see~\cref{sec:related}).

Most radix sorters, i.e., \radixlearned, \radixipp,
\imsdradix, \radixraduls, \radixppbbr, and
\radixregion, do not support all data
types used in our experiments.  Only the radix sorter
\radixsska supports all data types as the types used here are either
primitives or compositions of primitives, which are sorted
lexicographically.  All algorithms are written in C++ and compiled
with version 7.5.0 of the GNU compiler collection, using the
optimization flags ``\texttt{-march=native -O3}''. We have not
found a more recent compiler that
supports Cilk threads, needed for \radixregion, \compppbbs, and
\radixppbbr.

We ran benchmarks with $64$-bit floating-point elements, $64$-bit unsigned integers, $32$-bit unsigned integers, and \emph{\pair}, \emph{\quartet}, and \emph{\bytes} data types.
\pair (\quartet) consists of one (three) $64$-bit unsigned integers as key and one $64$-bit unsigned integer of associated information.
\bytes consists of $10$~bytes as key and $90$~bytes of associated information.
The keys of \quartet and \bytes are compared lexicographically.
\Cref{alg:comp quartet,alg:com bytes} show the lexicographical compare function which we used in our benchmarks.
We want to point out that lexicographical comparisons can be implemented in different ways.
We also tested \texttt{std::lexicographical\_compare} for \quartet and \texttt{std::memcmp} for \bytes.
However, it turned out that these compare functions are (much) less efficient for all competitive algorithms.
\radixsska is the only radix sorter that is able to sort keys lexicographically.
For \quartet and \bytes data types, we invoke \radixsska with a key-extractor function that returns the values of the key stored in a \texttt{std::tuple} object.

We ran benchmarks with ten input distributions: Uniformly
distributed~(\emph{\distuniform}), exponentially
distributed~(\emph{\distexpo}), and almost
sorted~(\emph{\distalmostsorted}), proposed by
Shun~et.~al.~\cite{shun2012brief}; \emph{\distduplicatesroot},
\emph{\distduplicatestwice}, and \emph{\distduplicateseight} from
Edelkamp~et.~al.~\cite{edelkamp2016blockquicksort}; and
\emph{\distzipf} (Zipf distributed input), \emph{\distsorted} (sorted
\distuniform input), \emph{\distreversesorted}, and \emph{\distones}
(just zeros).  The input distribution \distexpo generates and hashes
numbers selected uniformly at random from $[2^i, 2^{i+1})$ with
$i\in \mathbb{N} \wedge i \leq \log n$, \distduplicatesroot sets
$\VarArray[i] = i \mod \lfloor\sqrt{n}\rfloor$, \distduplicatestwice
sets $\VarArray[i] = i^2 + n/2 \mod n$, and
\distduplicateseight sets $\VarArray[i] = i^8+n/2 \mod n$.
The input distribution \distzipf generates the integer number
$k\in[1,10^2]$ with probability proportional to $1/k^{0.75}$.
\Cref{fig:distr} illustrates the nontrivial input distributions
\distduplicatesroot, \distzipf, \distexpo, \distduplicatestwice,
\distduplicateseight, and \distalmostsorted.

\input{extern/ips4o-benchmark-suite-plots/benchmark/distributions/dist.tex}

We ran our experiments on the following machines:

\begin{itemize}
\item Machine \emph{\pcamd} with one AMD Ryzen $9$ $3950X$ 16-core processor and $32$~GiB of memory.
\item Machine \emph{\pcintelfour} with one AMD EPYC Rome 7702P 64-core processor and $1024$~GiB of memory.
\item Machine \emph{\pcinteltwo} with two Intel Xeon E5-2683 v4 16-core processors and $512$~GiB of memory.
\item Machine \emph{\pcintellargefour} with four Intel Xeon Gold 6138 20-core processors and $768$~GiB of memory.
\end{itemize}

% 2-channel and 1 controller avx2
% 8 mem channels and 8 controllers avx2
% 4-channel and 1 controller avx2
% 6 channels and 2 controller avx512

Each algorithm was executed on all machines with all input
distributions and data types.  The parallel (sequential) algorithms
were executed for all input sizes with $n=2^i, i \in \mathbb{N}^+$,
until the input array exceeds $128$~GiB ($32$~GiB).  
For $n<2^{33}$ ($n<2^{30}$), we perform each parallel (sequential)
measurement $15$ times and for $n\geq 2^{33}$ ($n\geq 2^{30}$), we
perform each measurement twice.  Unless stated otherwise, we report
the average over all runs except the first one\footnote{The first run is excluded because we do not want to overemphasize time effects introduced by memory management, instruction cache warmup, side effects of different benchmark configurations, \ldots}.
We note
that non-in-place sorting algorithms which require an additional array
of $n$ elements will not be able to sort the largest inputs on
\emph{\pcamd}, the machine with $32$~GiB of memory.  We also want to
note that some algorithms did not support all data types because their
interface rejects the key.
In our figures, we emphasize algorithms that are ``not general-purpose''
with red lines, i.e., because they assume integer keys
(\compissrsort, \radixipp, \radixraduls, \radixregion, and
\radixppbbr), make additional assumptions on the data type
(\radixraduls), or at least because they do not accept a comparator
function (\radixsska).
All non-in-place algorithms except
\radixraduls return the sorted output in the input array.  These
algorithms copy the input back into the input array if the algorithm
has executed an even number of recursion levels.  Only \radixraduls
returns the output in a second ``temporary'' array.  To be fair, we
copy the data back into the input array in parallel and include the
time in the measurement.

We tested all parallel algorithms on \distuniform input with and
without hyper-threading.  Hyper-threading did not slow down any
algorithm.  Thus, we give results of all algorithms with
hyper-threading.  Overall, we executed more than $500\,000$
combinations of different algorithms, input distributions, input sizes,
data types, and machines.  We now present an interesting selection of our
measurements and discuss our results.

This chapter is divided as follows.
\Cref{sec:statistical evaluation} introduces and discusses the statistical measurement \emph{average slowdown} which we use to compare aggregated measurements.
We present the results of \compissssort and \compissrsort and their sequential competitors in \cref{sec:sequential comparison}.
In \cref{sec:numa allocation policy}, we discuss the influence of different memory allocation policies on the running time of parallel algorithms.
\Cref{sec:results task scheduler} compares our parallel algorithm \compiparassssort to its implementation presented in the conference version of this article~\cite{axtmann2017confplace}.
We compare the results of \compiparassssort and \compiparassrsort to their parallel competitors in \cref{sec:parallel comparison}.
Finally, \cref{sec:algorithm phases} evaluates the subroutines of \compiparassssort, \compiparassrsort, and their sequential counterparts.

%......................................................................
\subsection{Statistical Evaluation}\label{sec:statistical evaluation}
Many methods are available to compare algorithms. In our case, the cross product of machines, input distributions, input sizes, data types, and array types describes the possible inputs of our benchmark.
In this work we consider the result of a benchmark input always averaged over all executions of the input, using the arithmetic mean.
A common approach of presenting benchmark results is to fix all but two variables of the benchmark set and show a plot for these two variables, e.g., plot the running time of different algorithms over the input size in a graph for a specific input distribution, data type, and array type, executed on a specific machine.
Often, an interesting subset of all possible graphs is presented as the benchmark instances have too many parameters.
However, in this case, a lot of data is not presented at all and general propositions require further interpretation and aggregation of the presented, and possibly incomplete, data.
Besides running time graphs and speedup graphs, we use average slowdown factors (\emph{average slowdowns}) and \emph{performance profiles} to present our benchmark results.\\

Let $\mathcal{A}$ be a set of algorithms, let $\mathcal{I}$ be a set of inputs, let $S_A(\mathcal{I})$ be the inputs of $\mathcal{I}$ which algorithm $A$ sorts successfully, and let $r(A,I)$ be the running time of algorithm $A$ for input $I$.
Furthermore, let $r(A,I,T)$ be the running time of an algorithm $A$ for an input $I$ with array type $T$.
Note that $A$ might not sort $I$ successfully i.e., because its interface does not accept the data type or because $A$ does not return sorted input.
In this case, the running time of $A$ is not defined.

To obtain \textbf{average slowdowns}, we first define the \emph{slowdown factor} of an algorithm $A\in\mathcal{A}$ to the algorithms $\mathcal{A}$ for the input $I$

\[
  f_{\mathcal{A},I}(A)=
  \begin{cases}
    r(A,I)/\min(\{r(A',I)\,|\,A'\in\mathcal{A}\}) &\quad\text{$I\in S_A(\mathcal{I})$, i.e., $A$ successfully sorts $I$}\\
    \infty &\quad\text{otherwise.}
  \end{cases}
\]

as the slowdown using algorithm $A$ to sort input $I$ instead of using the fastest algorithm for $I$ from the set of algorithms $\mathcal{A}$.
Then, the \emph{average slowdown of algorithm $A\in\mathcal{A}$ to the algorithms $\mathcal{A}$ for the inputs $\mathcal{I}$}

\[
  s_{\mathcal{A},\mathcal{I}}(A)=\sqrt[|S_A(\mathcal{I})|]{\prod_{I\in S_A(\mathcal{I})}f_{\mathcal{A},I}(A)}
\]

is the geometric mean of the slowdown factors of algorithm $A$ to the algorithms $\mathcal{A}$ for the inputs of $\mathcal{I}$ which $A$ sorts successfully.

Besides the average slowdown of algorithms, we present average slowdowns of an input array type to compare its performance to a set $\mathcal{T}$ of array types.
The slowdown factor of an array $T\in\mathcal{T}$ to the arrays $\mathcal{T}$ for the input $I$ and an algorithm $A$

\[
  f_{\mathcal{T},A,I}(T) =
  \begin{cases}
    r(A,I,T)/\min(\{r(A,I,T')\,|\,T'\in\mathcal{T}\}) &\quad\text{$I\in S_A(\mathcal{I})$, i.e., $A$ successfully sorts $I$}\\
    \infty &\quad\text{otherwise.}
  \end{cases}
\]

is defined as the slowdown of using array type $T$ to sort input $I$ with algorithm $A$ instead of using the best array from the set of array types $\mathcal{T}$.

Then, the \emph{average slowdown of an array $T\in\mathcal{T}$ to the array types $\mathcal{T}$ for the inputs $\mathcal{I}$ and the algorithm $A$}

\[
  s_{A,\mathcal{T},\mathcal{I}}(T)=\sqrt[|S_A(\mathcal{I})|]{\prod_{I\in S_A(\mathcal{I})}f_{\mathcal{T},A,I}(T)}
\]

is the geometric mean of the slowdown factors of $T$ to the arrays $\mathcal{T}$ for the inputs $\mathcal{I}$ and algorithm $A$ which $A$ sorts successfully.

Average slowdown factors are heavily used by Timo Bingmann~\cite{Bingmann2018diss} to compare parallel string sorting algorithms.
We want to note that the average slowdown could also be defined as the arithmetic mean of the slowdown factors, instead of using the geometric mean.
In this case, the average slowdown would have a very strong meaning:
The average slowdown of an algorithm $A$ over a set of inputs is the expected average slowdown of A when an input of the benchmark set is picked at random to the fastest algorithm for this particular input.
However, Timo Bingmann used in his work the geometric mean for the average slowdowns to ``emphasize small relative differences of the fastest algorithms''.
Additionally, the geometric mean is more robust against outliers and skewed measurements than the arithmetic mean~\cite[p.~229]{mcgeoch2012guide}.
Furthermore, the arithmetic mean of ratios is ``meaningless'' in the general case.
For example, Fleming and Wallance~\cite{fleming1986not} state that the arithmetic mean is meaningless when different machine instances are compared relative to a baseline machine.
In this case, ratios smaller than one and larger than one can occur.
However, combining those numbers is ``meaningless'' as ratios larger than one depend linearly on the measurements but ratios smaller than one do not.
Note that the slowdown factors in this work will never be smaller than one.\\

A \textbf{performance profile}~\cite{dolan2002benchmarking} is the cumulative distribution function of an algorithm for a specific performance metric.
We use the slowdown factors to quantify the relative performance distribution of an algorithm to a set of competitors on a set of inputs.
The \emph{performance profile of algorithm $A\in\mathcal{A}$ to the algorithms $\mathcal{A}$ for the inputs $\mathcal{I}$}

$$p_{\mathcal{A}, \mathcal{I}, A}(\tau) = \frac{|\{I\in\mathcal{I}\,|\,f_{\mathcal{A}}(A, I)\leq \tau\}|}{|\mathcal{I}|}$$
is the probability that algorithm $A$ sorts a random input $I\in\mathcal{I}$ at most a factor $\tau$ slower than the fastest algorithm for input $I$.\\

To avoid skewed measurements we only use input sizes above a certain threshold for the calculation of average slowdowns and performance profiles.
This is an obvious decision as it is very common that algorithms switch to a base case algorithm for small inputs.
By restricting the input size, the results are not affected by the choice of different thresholds and the choice of the algorithm for small inputs.
Additionally, the algorithms generally sort ``very small'' inputs inefficiently (except algorithms with are designed for small inputs) and we do not want those measurements to dominate the average performance.
For sequential algorithms, we use inputs with at least $2^{18}$ bytes and for parallel algorithms, we use inputs with at least $2^{21}t$ bytes.

When an algorithm uses a heuristic to detect easy inputs and quickly transforms these inputs into sorted output, the slowdown factors of algorithms that do not use such heuristics are usually orders of magnitude larger than the remaining ratios of our benchmark.
When aggregating those large ratios with ratios of inputs that are not easy, the large ratios would dominate the result.
Therefore, we exclude the inputs \distones, \distsorted, and \distreversesorted when we average over all input distributions to obtain average slowdowns and performance profiles.

%......................................................................
\subsection{Sequential Algorithms}\label{sec:sequential comparison}

In this section, we compare sequential algorithms for different
machines, input distributions, input sizes, and data types.  We begin
with a comparison of the average slowdowns of \compissrsort,
\compissssort, and their competitors for ten input distributions
executed with six different data types (see \cref{sec:seq slowdown}).
This gives a first general view of the performance of our algorithms
as the presented results are aggregated across all machines.
Afterwards, we compare our algorithms to their competitors on
different machines by scaling with input sizes for input distribution
\distuniform and data type \ulong (see \cref{sec:seq uniform input}).
Finally, we discuss the performance profiles of the algorithms
in \cref{sec:seq perf profiles}.

\subsubsection{Comparison of Average Slowdowns}\label{sec:seq slowdown}

\begin{table}
  \centering
  \resizebox*{!}{0.93\textheight}{

\begin{tabular}{ll|rrrrrrrrr|rrrr}
  Type
  & Distribution
  & \rotatebox[origin=c]{90}{\compissssort}
  & \rotatebox[origin=c]{90}{\compspdq}
  & \rotatebox[origin=c]{90}{\compblock}
  & \rotatebox[origin=c]{90}{\compmyssssaxtmann}
  & \rotatebox[origin=c]{90}{\compsyaros}
  & \rotatebox[origin=c]{90}{\compssort}
  & \rotatebox[origin=c]{90}{\compstim}
  & \rotatebox[origin=c]{90}{\compsmergequick}
  & \rotatebox[origin=c]{90}{\compswiki}
  & \rotatebox[origin=c]{90}{\radixsska}
  & \rotatebox[origin=c]{90}{\radixipp}
  & \rotatebox[origin=c]{90}{\radixlearned}
  & \rotatebox[origin=c]{90}{\compiparassrsort}\\\hline
  \double &        \distsorted &          1.05 & 1.70 & 25.24 & \textbf{1.05} & 12.90 & 20.49 & 1.09 & 62.42 & 2.81 & 21.83 & 62.61 & 56.19 &  \\
  \double & \distreversesorted & \textbf{1.04} & 1.71 & 14.28 &          1.06 &  5.09 &  5.93 & 1.07 & 25.34 & 5.89 &  9.41 & 25.22 & 21.24 &  \\
  \double &          \distones & \textbf{1.07} & 1.77 & 21.20 &          1.10 &  1.20 & 14.98 & 1.08 &  2.72 & 3.58 & 16.36 & 24.23 & 15.08 &  \\

  \hline\hline
  
  \double &            \distexpo & \textbf{1.02} & 1.13 & 1.28 & 1.27 & 2.30 & 2.57 &          4.23 & 4.04 & 4.20 &          1.29 & 1.38 & 2.08 &  \\
  \double &            \distzipf & \textbf{1.08} & 1.25 & 1.42 & 1.37 & 2.66 & 2.87 &          4.63 & 4.21 & 4.72 &          1.17 & 1.28 & 2.30 &  \\
  \double &  \distduplicatesroot & \textbf{1.10} & 1.50 & 1.83 & 1.65 & 1.44 & 2.30 &          1.32 & 6.01 & 3.12 &          1.90 & 2.69 & 3.18 &  \\
  \double & \distduplicatestwice &          1.17 & 1.33 & 1.37 & 1.41 & 2.48 & 2.65 &          2.96 & 3.42 & 3.20 & \textbf{1.07} & 1.22 & 2.52 &  \\
  \double & \distduplicateseight & \textbf{1.01} & 1.13 & 1.41 & 1.30 & 2.42 & 2.84 &          4.43 & 4.69 & 4.40 &          1.31 & 1.60 & 2.46 &  \\
  \double &    \distalmostsorted &          2.33 & 1.15 & 2.21 & 2.99 & 1.68 & 1.80 & \textbf{1.14} & 6.76 & 2.57 &          2.39 & 4.53 & 4.88 &  \\
  \double &         \distuniform &          1.08 & 1.21 & 1.22 & 1.28 & 2.35 & 2.43 &          3.59 & 2.98 & 3.58 & \textbf{1.08} & 1.29 & 2.07 &  \\

  \hline
  Total  & &

  \textbf{1.20} & 1.24 & 1.50 & 1.54 & 2.15 & 2.47 & 2.81 & 4.42 & 3.61 & 1.40 & 1.77 & 3.00 &  \\

  Rank & &
  1 & 2 & 4 & 5 & 7 & 8 & 9 & 12 & 11 & 3 & 6 & 10 &  \\\hline\hline
           
  \ulong &        \distsorted &          1.17 & 1.78 & 23.69 & \textbf{1.02} & 11.94 & 19.96 & 1.11 & 55.40 & 2.88 & 26.64 & 76.86 & 95.80 & 13.33 \\
  \ulong & \distreversesorted & \textbf{1.03} & 1.63 & 12.93 &          1.04 &  4.47 &  5.51 & 1.04 & 21.01 & 5.93 & 10.46 & 28.99 & 34.39 &  5.97 \\
  \ulong &          \distones &          1.17 & 1.69 & 21.43 & \textbf{1.06} &  1.14 & 14.02 & 1.11 &  2.42 & 3.74 & 17.40 & 25.30 & 14.90 &  1.35 \\

  \hline\hline
  
  \ulong &            \distexpo & 1.06 &          1.22 & 1.37 & 1.37 & 2.28 & 2.64 & 4.52 & 3.82 & 4.51 & 1.21 & 1.74 & 2.09 & \textbf{1.05} \\
  \ulong &            \distzipf & 1.53 &          1.86 & 2.13 & 2.06 & 3.62 & 4.04 & 6.65 & 5.53 & 6.79 & 1.73 & 1.99 & 2.56 & \textbf{1.01} \\
  \ulong &  \distduplicatesroot & 1.25 &          1.73 & 2.19 & 2.07 & 1.60 & 2.60 & 1.70 & 6.34 & 3.91 & 2.08 & 2.88 & 3.60 & \textbf{1.13} \\
  \ulong & \distduplicatestwice & 1.73 &          2.07 & 2.11 & 2.17 & 3.56 & 3.88 & 4.54 & 4.65 & 4.93 & 1.58 & 2.66 & 3.32 & \textbf{1.00} \\
  \ulong & \distduplicateseight & 1.26 &          1.39 & 1.74 & 1.64 & 2.75 & 3.29 & 5.46 & 5.12 & 5.38 & 1.71 & 2.97 & 2.40 & \textbf{1.02} \\
  \ulong &    \distalmostsorted & 2.34 & \textbf{1.11} & 2.19 & 3.28 & 1.68 & 1.81 & 1.24 & 6.11 & 2.79 & 2.79 & 6.67 & 8.55 &          1.28 \\
  \ulong &         \distuniform & 1.35 &          1.60 & 1.60 & 1.71 & 2.85 & 3.02 & 4.62 & 3.49 & 4.63 & 1.20 & 2.19 & 4.46 & \textbf{1.04} \\

  \hline
  Total  & &

  1.46 & 1.54 & 1.88 & 1.97 & 2.51 & 2.95 & 3.56 & 4.90 & 4.56 & 1.69 & 2.74 & 4.87 & \textbf{1.07} \\

  Rank & &
  2 & 3 & 5 & 6 & 7 & 9 & 10 & 13 & 11 & 4 & 8 & 12 & 1 \\\hline\hline

  \uint &        \distsorted & 2.44 & 3.89 & 57.73 & 2.42 & 28.63 & 53.53 & \textbf{1.96} & 139.13 & 6.34 & 46.41 & 44.93 & 275.54 & 29.91 \\
  \uint & \distreversesorted & 1.40 & 2.06 & 17.70 & 1.47 &  6.09 &  8.37 & \textbf{1.03} &  29.37 & 5.57 & 10.08 & 20.92 &  53.61 &  7.28 \\
  \uint &          \distones & 2.30 & 3.71 & 59.44 & 2.28 &  2.28 & 37.19 & \textbf{2.06} &   6.19 & 8.98 & 24.29 & 14.03 &  40.53 &  3.05 \\

  \hline\hline
  
  \uint &            \distexpo & 1.49 & 1.77 & 2.03 & 1.82 & 3.66 & 4.04 &          6.67 & 5.91 & 6.51 & 1.38 & \textbf{1.08} &  3.44 &          1.09 \\
  \uint &            \distzipf & 1.82 & 2.33 & 2.75 & 2.37 & 4.97 & 5.46 &          8.68 & 7.55 & 8.81 & 1.41 &          1.27 &  3.93 & \textbf{1.12} \\
  \uint &  \distduplicatesroot & 1.41 & 1.92 & 2.46 & 2.15 & 1.84 & 2.97 &          1.48 & 7.54 & 3.78 & 1.58 &          1.77 &  4.43 & \textbf{1.18} \\
  \uint & \distduplicatestwice & 2.09 & 2.56 & 2.67 & 2.52 & 4.82 & 5.11 &          5.59 & 5.94 & 5.95 & 1.34 &          1.44 &  5.08 & \textbf{1.09} \\
  \uint & \distduplicateseight & 1.40 & 1.68 & 2.09 & 1.76 & 3.67 & 4.19 &          6.47 & 6.23 & 6.45 & 1.35 &          1.77 &  3.05 & \textbf{1.02} \\
  \uint &    \distalmostsorted & 3.07 & 1.45 & 2.79 & 4.24 & 2.15 & 2.58 & \textbf{1.06} & 8.24 & 2.97 & 2.66 &          5.45 & 12.04 &          1.51 \\
  \uint &         \distuniform & 1.67 & 2.01 & 2.05 & 2.04 & 3.85 & 4.02 &          5.92 & 4.55 & 5.79 & 1.39 & \textbf{1.08} &  5.22 &          1.20 \\

  \hline
  Total  & &

  1.78 & 1.93 & 2.39 & 2.32 & 3.37 & 3.93 & 4.09 & 6.47 & 5.45 & 1.54 & 1.67 & 6.71 & \textbf{1.16} \\

  Rank & &
  4 & 5 & 7 & 6 & 8 & 9 & 10 & 12 & 11 & 2 & 3 & 13 & 1 \\\hline\hline
  
  \pair &        \distsorted & 1.06 & 1.62 & 16.88 & \textbf{1.04} & 9.36 & 14.67 & 1.04 & 34.54 & 2.30 & 17.51 &  &  & 10.48 \\
  \pair & \distreversesorted & 1.13 & 1.21 &  8.47 & \textbf{1.08} & 3.65 &  4.19 & 1.12 & 13.71 & 6.60 &  6.86 &  &  &  4.87 \\
  \pair &          \distones & 1.09 & 1.61 & 13.30 & \textbf{1.03} & 1.07 & 11.63 & 1.08 &  1.94 & 2.71 & 11.09 &  &  &  1.21 \\

  \hline\hline
  
  \pair &            \distexpo & 1.10 &          1.92 & 1.20 & 1.36 & 1.84 & 2.12 & 3.87 & 3.11 & 4.14 & 1.16 &  &  & \textbf{1.05} \\
  \pair &            \distzipf & 1.48 &          2.72 & 1.64 & 1.86 & 2.64 & 2.83 & 5.02 & 3.87 & 5.50 & 1.46 &  &  & \textbf{1.01} \\
  \pair &  \distduplicatesroot & 1.27 &          1.44 & 1.78 & 1.84 & 1.42 & 2.16 & 1.83 & 4.70 & 4.05 & 1.69 &  &  & \textbf{1.03} \\
  \pair & \distduplicatestwice & 1.63 &          2.81 & 1.69 & 1.92 & 2.71 & 2.84 & 3.62 & 3.45 & 4.35 & 1.41 &  &  & \textbf{1.01} \\
  \pair & \distduplicateseight & 1.27 &          2.19 & 1.45 & 1.59 & 2.14 & 2.47 & 4.50 & 3.95 & 4.81 & 1.56 &  &  & \textbf{1.00} \\
  \pair &    \distalmostsorted & 3.20 & \textbf{1.01} & 2.79 & 4.00 & 2.18 & 2.39 & 2.34 & 6.56 & 4.55 & 3.24 &  &  &          1.74 \\
  \pair &         \distuniform & 1.37 &          2.46 & 1.45 & 1.66 & 2.40 & 2.45 & 3.89 & 2.88 & 4.26 & 1.17 &  &  & \textbf{1.03} \\

  \hline
  Total  & &

  1.52 & 1.97 & 1.66 & 1.91 & 2.15 & 2.45 & 3.40 & 3.94 & 4.50 & 1.57 &  &  & \textbf{1.10} \\

  Rank & &
  2 & 6 & 4 & 5 & 7 & 8 & 9 & 10 & 11 & 3 &  &  & 1 \\\hline\hline
  
  \quartet & \distuniform & 1.14 & 1.85 & 1.29 & 1.49 & 1.89 & 1.86 & 3.14 & 2.15 & 3.52 & \textbf{1.02} &  &  &  \\

  \hline

  Rank & &
  2 & 5 & 3 & 4 & 7 & 6 & 9 & 8 & 10 & 1 &  &  &  \\\hline\hline
           
  \bytes & \distuniform & 1.41 & 1.27 & 1.27 & 1.64 & 1.83 & 1.33 & 2.22 & 1.78 & 3.17 & \textbf{1.06} &  &  &  \\

  \hline

  Rank & &
  5 & 2 & 3 & 6 & 8 & 4 & 9 & 7 & 10 & 1 &  &  &  \\\hline\hline
\end{tabular}

  }
  \caption{
    Average slowdowns of sequential algorithms for different data types and input distributions.
    The slowdowns average over the machines and input sizes with at least $2^{18}$ bytes.
    \label{tab:slowdown seq all}
  }
\end{table}

In this section, we discuss the average slowdowns of sequential
algorithms for different data types and input distributions aggregated
over all machines and input sizes with at least $2^{18}$ bytes shown
in \cref{tab:slowdown seq all}. The results indicate that
a sorting algorithm performs similarly good for inputs with
``similar'' input distributions. Thus, we divide the inputs
into four groups: The largest group, \emph{Skewed inputs}, contains
inputs with duplicated keys and skewed key occurrences, i.e.,
\distexpo, \distzipf, \distduplicatesroot, \distduplicatestwice, and
\distduplicateseight. The second group, \emph{Uniform inputs},
contains \distuniform distributed inputs. For these inputs,
each bit of the key has maximum entropy.  Thus, we can expect that
radix sort performs the best for these inputs.  The third group,
\emph{Almost Sorted inputs}, are \distalmostsorted distributed inputs.
The last group, \emph{(Reverse) Sorted inputs}, contains
``easy inputs'', i.e., \distsorted, \distreversesorted, and
\distones.
\ifarxiv For the average slowdowns separated by machine, we
refer to \cref{app:more measurements},
\cref{tab:slowdown seq
  128,tab:slowdown seq 132,tab:slowdown seq 133,tab:slowdown seq 135}.
\else For the average slowdowns separated by machine, we
refer to the extended version of this paper~\cite{axtmann2020ips4oarxiv}.
\fi

In this section, an \emph{instance} describes the inputs of a
specific data type and input distribution.
We say that ``algorithm A is faster than algorithm B
(by a factor of C) for some instances'' if the average slowdown
of B is larger than the average slowdown of A (by a factor of C) for these instances.

The subsequent paragraph summarizes the performance of our algorithms.
Then, we
compare our competitors to our radix sorter
\compissrsort. Finally, we compare our competitors to
our samplesort algorithm \compissssort.

Overall,  \compissrsort is significantly
faster than our fastest radix sort competitor \radixsska. For example,
\compissrsort is for all  instances at least a factor of $1.10$
faster than \radixsska and for even $63$\,\% of the instances more
than a factor of $1.40$. The radix sorter \radixipp is faster
than \compissrsort in some special cases.
However, \radixipp is even
slower than our competitor \radixsska for the remaining inputs.  Our algorithm
\compissrsort also outperforms the comparison-based sorting algorithms
for Uniform inputs and Skewed inputs significantly. For example,
\compissrsort is faster for all of these  instances and for
$56$\,\% of these instances even a factor of $1.20$ or more. Only for
Almost Sorted inputs and the ``easy'' (Reverse) Sorted inputs, the
comparison-based algorithms \compspdq and \compstim are faster than
\compissrsort.  For the remaining inputs -- Uniform inputs and Skewed
inputs -- not only our radix sorter \compissrsort but also our
samplesort algorithm \compissssort is faster than all comparison-based
competitors (except for one  instance). For example,
\compissssort is faster than our fastest comparison-based competitor,
\compspdq by a factor of $1.10$ and $1.20$ for $25$
respectively $15$ out of $26$ instances with Uniform and Skewed input.
\compissssort is on average
also faster than the fastest radix sorter.\\

\emph{Comparison to \compissrsort.}
\compissrsort outperforms \textbf{\radixsska} by a factor of $1.10$,
$1.20$, $1.30$, and $1.40$ for respectively $100$\,\%, $83$\,\%,
$73$\,\%, and $63$\,\% of the instances.  Also, \compissrsort
performs much better than \textbf{\radixlearned} for any data type and
input distribution.
\textbf{\radixipp} is the only non-comparison-based algorithm
that is able to outperform
\compissrsort for at least one  instance, i.e., \radixipp is
faster by a factor of $1.01$ (of $1.11$) for \distexpo (\distuniform) distributed inputs with the \uint data type.  However,
\radixipp is (significantly) slower than \compissrsort for other \distexpo (\distuniform)
 instances, e.g., a factor of $1.66$ (of $2.11$) for
the \ulong data type.  For the remaining instances, \radixipp is (much)
slower than \compissrsort.

For the comparison of \compissrsort to comparison-based algorithms, we
first consider Uniform and Skewed inputs. For these  instances,
\compissrsort is significantly faster than all comparison-based
algorithms (including \compissssort).  For example,
\compissrsort is faster than any of these algorithms by a factor of
more than $1.00$, $1.10$, $1.20$, and $1.30$ for respectively
$100$\,\%, $89$\,\%, $78$\,\%, and $56$\,\% of the instances.
We now consider Almost Sorted inputs which are sorted the
fastest by \textbf{\compspdq}. For all $3$ instances,
\compissrsort is slower than \compspdq, i.e., by a factor of
respectively $1.04$, $1.16$, and $1.72$.
The reason is that \compspdq
heuristically detects and skips presorted input sequences.
Even though \compspdq is our fastest comparison-based competitor, it
is significantly slower than \compissrsort for
many instances which are not (almost) sorted.  For example, \compspdq
is for $8$ of these  instances even more than a factor of $2.00$
slower than \compissrsort.  For (Reverse) Sorted inputs, \compissrsort
is slower than at least one comparison-based sorting algorithm for all
$9$  instances.  However, \compissrsort could easily detect these
instances by scanning the input array once. We want to note that
\compissrsort already scans the input array
to detect the significant bits of the input keys.\\

\emph{Comparison to  \compissssort.} Our algorithm \compissssort is faster than
any comparison-based competitor for $28$  instances and slower
for only $15$  instances.  However, when we exclude Almost Sorted
inputs and (Reverse) Sorted inputs, \compissssort is still faster for
the same number of instances but the number of instances for which
\compissssort is slower drops to one  instance.  When we only
exclude (Reverse) Sorted inputs, \compissssort is still only slower
for $5$  instances.

\textbf{\compspdq} is a factor of $1.10$, $1.15$, $1.20$, and $1.25$ slower than
\compissssort for respectively $100$\,\%, $71.43$\,\%, $42.85$\,\%,
and $28.57$\,\% out of $21$ instances with Uniform input and Skewed
input.
\compspdq is also much slower for (Reverse) Sorted inputs.
Only for Almost Sorted inputs, \compspdq is significantly
faster than \compissssort, e.g., by a factor of $2.03$ to $3.17$. Again, the
reason is that \compspdq takes advantage of presorted sequences in the
input.

\textbf{\compblock} shows similar performance as \compspdq for \distuniform
inputs.  The reason is that \compspdq reimplemented the partitioning
routine proposed \compblock~\cite{edelkamp2016blockquicksort}.
However, \compblock does not take advantage of presorted sequences and
\compblock handles duplicate keys less efficient.  Thus, \compblock is
slower than \compspdq for (Reverse) Sorted inputs and Almost Sorted
inputs.

\compissssort outperforms \textbf{\radixsska}
for Skewed inputs by a factor of at least $1.10$ for $50$\,\% out of
$20$  instances whereas \radixsska is faster by a factor of at
least $1.10$ for only $15$\,\% of the instances. \compissssort is also
faster than \radixsska for all $12$ (Reverse) Sorted inputs.  For
Almost Sorted inputs, both algorithms are for one  instance at
least a factor of $1.10$ faster than the other algorithm (out of $4$  instances).  Only
for Uniform inputs, \radixsska is the better algorithm.  I.e.,
\radixsska is faster by a factor of at least $1.10$ on $83$\,\% out of
$6$ Uniform  instances whereas \compissssort is not faster on one
of these  instances.

As expected, \textbf{\compmyssssaxtmann} is slower than
\compissssort for all  instances except (Reverse) Sorted
inputs. For (Reverse) Sorted inputs, both algorithms execute the
same heuristic to detect and sort ``easy'' inputs. Also, as
\compmyssssaxtmann is not in-place, \compmyssssaxtmann can sort only
about half the input size as \compissssort can sort.  The results
strongly indicate that the I/O complexity of \compissssort has smaller constant factors than the I/O complexity of
\compmyssssaxtmann as both algorithms share the same sorting framework
including the same sampling routine, branchless decision tree, and base
case sorting algorithm.

The algorithms \textbf{\compssort}, \textbf{\compsyaros}, and
\compspdq are adaptions of quicksort.  However, \compssort
and \compsyaros do not avoid branch mispredictions by classifying and
moving blocks of elements.  The result is that these algorithms are
always significantly slower than \compspdq.

We also compare \compissssort to the mergesort algorithms
\compstim, \compsmergequick, and
\compswiki.  The in-place versions of mergesort,
\textbf{\compsmergequick} and \textbf{\compswiki}, are significantly slower than
\compissssort for all  instances.
\textbf{\compstim} is also much slower than \compissssort for almost all input
distributions -- in most cases even more than a factor of three.  Only
for Almost Sorted inputs, \compstim is faster than \compissssort and
for (Reverse) Sorted inputs, \compstim has similar running times as
\compissssort.
\\ 

We did not present the results of \compssssschneider, an implementation of
Super Scalar Samplesort~\cite{sanders2004super}.  We made this
decision as \compssssschneider is for all instances except \bytes
instances slower or significantly slower than \compmyssssaxtmann, our
implementation of Super Scalar Samplesort.  For further details, we
refer to \ifarxiv \cref{tab:slowdown seq s4o ssss} in \cref{app:more
  measurements} \else the extended version of this paper~\cite{axtmann2020ips4oarxiv} \fi which shows average slowdowns of
\compmyssssaxtmann and \compssssschneider for different data types and
input distributions. We did not the present results of the sequential
version of \comppaspas for three reasons.  First, \comppaspas performs
worse than \compmyssssaxtmann for all instances.  Second, \comppaspas
only sorts inputs with the data type \double.  Finally, \comppaspas
returns unsorted output for inputs with at least $2^{31}$
elements.

\input{extern/ips4o-benchmark-suite-plots/benchmark/running_times/running_time_random_sequential_some_machines_ulong.tex}

\subsubsection{Running Times for \distuniform Input}\label{sec:seq uniform input}

In this section, we compare \compissrsort and \compissssort to their
closest sequential competitors for \distuniform distributed \ulong
inputs.
\Cref{fig:rt seq ulong} depicts the running times of our algorithms \compissssort and \compissrsort as well as their fastest competitors \compspdq and \radixsska separately
for each machine.
Additionally, we include measurements obtained from \compmyssssaxtmann (which we used as a starting point to develop \compissssort) and \radixipp (which is fast for \uint data types with \distuniform distribution).
We decided to present results for \ulong inputs as
our radix sorter does not support \double inputs. We note that this
decision is not a disadvantage for our fastest competitors as they
show similar running times relative to our algorithms for
both data types (see slowdowns in \cref{tab:slowdown seq all}, \cref{sec:seq
  slowdown}).

Overall, \compissrsort outperforms its radix sort competitors on all
but one machine, and \compissssort significantly outperforms its
comparison-based competitors.  In particular, \compissrsort is a
factor of up to $1.40$ faster than \radixsska, and \compissssort is a
factor of up to $1.44$ ($1.60$) faster than \compspdq
(\compmyssssaxtmann) for the largest input size.  As expected,
\compissrsort is in almost all cases significantly faster than
\compissssort on all machines, e.g., a factor of $1.10$ to $1.52$ for
the largest input size.  \compissrsort shows the fastest running times
on the two machines with the most recent CPUs, \pcamd and
\pcintelfour.  On these two machines, the gap between \compissrsort
and \compissssort is larger than on the other machines.  This
indicates that sequential comparison-based algorithms are not memory
bound in general, and, on recent CPUs, radix sorters may benefit even
more from their reduced number of instructions (for
uniformly distributed inputs).  In the following, we compare our algorithms to their competitors in more detail.\\

\emph{Comparison to \compissrsort.} Our algorithm \compissrsort outperforms
\textbf{\radixsska} on two machines significantly (\pcamd and \pcintelfour), on
one machine slightly (\pcintellargefour), and on one machine (\pcinteltwo),
\compissrsort is slightly slower than \radixsska.  For example,
\compissrsort is on average a factor of respectively
$2.06$, $2.13$, $1.12$, and $0.82$ faster
than \radixsska for $n\geq 2^{15}$ on \pcamd, \pcintelfour, \pcintellargefour, and \pcinteltwo.
According to the
performance measurements, obtained with the Linux tool \texttt{perf},
\radixsska performs more cache misses (factor $1.25$) and
significantly more branch mispredictions (factor $1.52$
for $n=2^{28}$ on \pcintellargefour).

On machine \pcamd, \pcinteltwo, and \pcintelfour, we see that the running times
of \radixsska and \compissrsort vary -- with peaks at $2^{15}$, $2^{23}$ and $2^{31}$.
We assume that the running time peaks as these radix
sorters perform an additional $k$-way partitioning step with
$k=256$.  We have seen the same behavior with our algorithm
\compissssort when we do not adjust $k$ at the last recursion levels.
However, with our adjustments, the large running time peaks disappear
for \compissssort.

We also compare our algorithm against \textbf{\radixipp} which
takes advantage of the \emph{Advanced Vector Extensions} (AVX).  All
machines support the instruction extension \texttt{AVX2}.
\pcintellargefour additionally provides \texttt{AVX-512} instructions.
We expected that \radixipp is competitive, at least on
\pcintellargefour.  However, \radixipp is significantly slower than
\compissrsort on all machines.  For example, \compissrsort outperforms
\radixipp by a factor of $1.76$ to $1.88$ for the largest input size
on \pcintelfour, \pcamd, and \pcintellargefour.  On
\pcinteltwo, \compissrsort is even a factor of $3.00$ faster.
We want to note that
\radixipp is surprisingly fast for (mostly small) \distuniform
distributed inputs with \ifarxiv data type \uint (see \cref{fig:rt seq uint}
in \cref{app:more measurements}). \else data type \uint. For running times obtained with the data type \uint, we refer to the extended version of this paper~\cite{axtmann2020ips4oarxiv}. \fi
Unfortunately, \radixipp fails
to sort \uint inputs with more than $2^{28}$ elements.  In conclusion,
it seems that \texttt{AVX} instructions only help for inputs whose
data type size is very small, i.e., $32$-bit unsigned integers in our
case.

Contrary to the experiments presented by
Kristo~et~al.~\cite{kristo2020case}, the running times of
\radixlearned are very large and would break the running time limits
of \cref{fig:rt seq ulong}.  Our experiments have shown that the
performance of \radixlearned degenerates by orders of magnitude for
input sizes which are not a multiple of $10^6$.  This problem has been
identified by others and was still an open
issue~\cite{kristo2020caseimpl} at the time when we finished our experiments.\\

\emph{Comparison to \compissssort.} For most medium and large
input sizes, \textbf{\compspdq} and \textbf{\compmyssssaxtmann} are significantly slower
than \compissssort.  For example, on \pcamd,
\compissssort is a factor of $1.29$ faster than \compmyssssaxtmann and
a factor of $1.44$ faster than \compspdq for the largest
input size ($n=2^{32}$).  On the other machines, \compspdq is our
closest competitor: \compissssort is a factor of $1.23$ to $1.44$ (of $1.29$ to $1.60$)
faster than \compspdq (\compmyssssaxtmann) for $n=2^{32}$.
According to the performance measurements,
obtained with the Linux tool \texttt{perf}, there may be several
reasons why \compissssort outperforms \compspdq
and \compmyssssaxtmann. Consider the machine \pcintellargefour
and $n=2^{38}$: \compspdq performs
significantly more instructions (factor $1.30$), more cache
misses (factor $2.05$), and more branch mispredictions (factor $1.69$) compared to \compissssort.  Also, \compmyssssaxtmann performs significantly more total
cache misses (factor $1.79$), more L3-store operations (factor
$2.68$), and more L3-store misses (factor $9.71$).
We note that the comparison-based competitors \compsyaros, \compssort,
\compstim, \compsmergequick, and \compswiki perform significantly more
branch mispredictions than \compmyssssaxtmann,
\compspdq, and \compblock. We think that this is the reason for their
poor performance.

\subsubsection{Comparison of Performance Profiles}\label{sec:seq perf profiles}

\begin{figure}[tbp]
  \begin{tikzpicture}
    \begin{groupplot}[
      group style={
        group size=4 by 1,
        x descriptions at=edge bottom,
        y descriptions at=edge left,
        horizontal sep=0.4cm,
      },
      width=0.315\textwidth,
      height=0.2\textheight,
      xmin=1,
      xmax=3,
      ymin=0.0,
      ymax=1.05,
      plotstylesequentialperf,
      xmajorgrids=true,
      ymajorgrids=true,
      ytick = {0, 0.25, 0.5, 0.75, 1},
      xtick = {1, 1.5, 2, 2.5, 3},
      ]
      
      \nextgroupplot[
      ylabel={Fraction of inputs},
      ]
      \addplot coordinates { (1,0.795487) (1.125,0.843912) (1.25,0.863658) (1.375,0.87118) (1.5,0.879643) (1.625,0.882464) (1.75,0.886695) (1.875,0.890926) (2,0.897508) (2.125,0.90362) (2.25,0.915374) (2.375,0.930889) (2.5,0.948284) (2.625,0.955806) (2.75,0.96615) (2.875,0.973672) (3,0.976963) (4,5) };
      \addlegendentry{\compissssort};
      \addplot coordinates { (1,0.204513) (1.125,0.381288) (1.25,0.595205) (1.375,0.673249) (1.5,0.74283) (1.625,0.800658) (1.75,0.845322) (1.875,0.890926) (2,0.932299) (2.125,0.961918) (2.25,0.975552) (2.375,0.989187) (2.5,0.995769) (2.625,0.99859) (2.75,0.99953) (2.875,1) (3,1) (4,5) };
      \addlegendentry{\compspdq};
      \addplot coordinates { (4,5) };
      \addlegendentry{\radixsska};
      \addplot coordinates { (4,5) };
      \addlegendentry{\compissrsort};

      \legend{}
      
      \coordinate (c1) at (rel axis cs:0,1);
      
      \nextgroupplot[
      every legend/.append style={at=(ticklabel cs:1.1)},
      legend style={at={($(0,0)+(1cm,1cm)$)},legend columns=6,fill=none,draw=black,anchor=center,align=center},
      ]
      \addplot coordinates { (1,0.544899) (1.125,0.702398) (1.25,0.832628) (1.375,0.916314) (1.5,0.946874) (1.625,0.962858) (1.75,0.974142) (1.875,0.985896) (2,0.992478) (2.125,0.996709) (2.25,0.99906) (2.375,0.99953) (2.5,1) (2.625,1) (2.75,1) (2.875,1) (3,1) (4,5) };
      \addlegendentry{\compissssort};
      \addplot coordinates { (4,5) };
      \addlegendentry{\compspdq};
      \addplot coordinates { (1,0.455101) (1.125,0.630465) (1.25,0.763987) (1.375,0.841561) (1.5,0.889516) (1.625,0.923366) (1.75,0.947344) (1.875,0.961918) (2,0.972732) (2.125,0.981194) (2.25,0.986366) (2.375,0.991537) (2.5,0.993888) (2.625,0.995769) (2.75,0.996709) (2.875,0.99859) (3,0.99906) (4,5) };
      \addlegendentry{\radixsska};
      \addplot coordinates { (4,5) };
      \addlegendentry{\compissrsort};

      \legend{}
      
      \nextgroupplot[
      ]
        \addplot coordinates { (4,5) };
      \addlegendentry{\compissssort};
      \addplot coordinates { (1,0.152247) (1.125,0.199866) (1.25,0.260228) (1.375,0.312542) (1.5,0.37827) (1.625,0.43662) (1.75,0.498323) (1.875,0.551979) (2,0.604963) (2.125,0.652582) (2.25,0.693494) (2.375,0.733065) (2.5,0.768612) (2.625,0.804158) (2.75,0.834339) (2.875,0.871227) (3,0.895372) (4,5) };
      \addlegendentry{\compspdq};
      \addplot coordinates { (4,5) };
      \addlegendentry{\radixsska};
      \addplot coordinates { (1,0.847753) (1.125,0.885983) (1.25,0.918176) (1.375,0.936955) (1.5,0.95171) (1.625,0.965124) (1.75,0.973172) (1.875,0.977197) (2,0.979879) (2.125,0.984574) (2.25,0.988598) (2.375,0.99061) (2.5,0.991952) (2.625,0.992622) (2.75,0.993293) (2.875,0.994634) (3,0.995305) (4,5) };
      \addlegendentry{\compissrsort};

      \legend{}
      
      \nextgroupplot[
      every legend/.append style={at=(ticklabel cs:1.1)},
      legend style={at={($(0,0)+(1cm,1cm)$)},legend columns=6,fill=none,draw=black,anchor=center,align=center},
      legend to name=legendplotscseq
      ]
      \addplot coordinates { (4,5) };
      \addlegendentry{\compissssort};
      \addplot coordinates { (4,5) };
      \addlegendentry{\compspdq};
      \addplot coordinates { (1,0.122736) (1.125,0.215292) (1.25,0.336687) (1.375,0.452716) (1.5,0.535211) (1.625,0.62173) (1.75,0.706908) (1.875,0.761234) (2,0.810865) (2.125,0.847082) (2.25,0.894702) (2.375,0.926224) (2.5,0.94165) (2.625,0.948357) (2.75,0.957747) (2.875,0.967807) (3,0.975855) (4,5) };
      \addlegendentry{\radixsska};
      \addplot coordinates { (1,0.877264) (1.125,0.948357) (1.25,0.976526) (1.375,0.989269) (1.5,0.992622) (1.625,0.993293) (1.75,0.993964) (1.875,0.997317) (2,1) (2.125,1) (2.25,1) (2.375,1) (2.5,1) (2.625,1) (2.75,1) (2.875,1) (3,1) (4,5) };
      \addlegendentry{\compissrsort};

      \coordinate (c2) at (rel axis cs:1,1);

    \end{groupplot}

    \coordinate (c23s) at ($(group c2r1.south)!.5!(group c3r1.south)$);
    \node[yshift=-0.75cm] at (c23s) {{Running time ratio relative to best $-$ $1$}};

    \coordinate (c12n) at ($(group c1r1.north)!.5!(group c2r1.north)$);
    \node[yshift=0.3cm] at (c12n) {{All Data Types}};

    \coordinate (c34n) at ($(group c3r1.north)!.5!(group c4r1.north)$);
    \node[yshift=0.3cm] at (c34n) {{Unsigned Integer Key}};

    \coordinate (c3) at ($(c1)!.5!(c2)$);
    \node[below] at (c3 |- current bounding box.south)
    {\pgfplotslegendfromname{legendplotscseq}};
  \end{tikzpicture}
  \caption{Pairwise performance profiles of our algorithms
    \compissssort and \compissrsort to 
    \compspdq and \radixsska. The performance plots with \compissssort
    use all data types. The performance plots with the radix sort
    algorithm \compissrsort use inputs with with unsigned
    integer keys (\uint, \ulong, and \pair data types). The results
    were obtained on all machines for all input distributions with at
    least $2^{18}$ bytes except \distsorted, \distreversesorted, and
    \distones.}\label{fig:perf seq}
\end{figure}
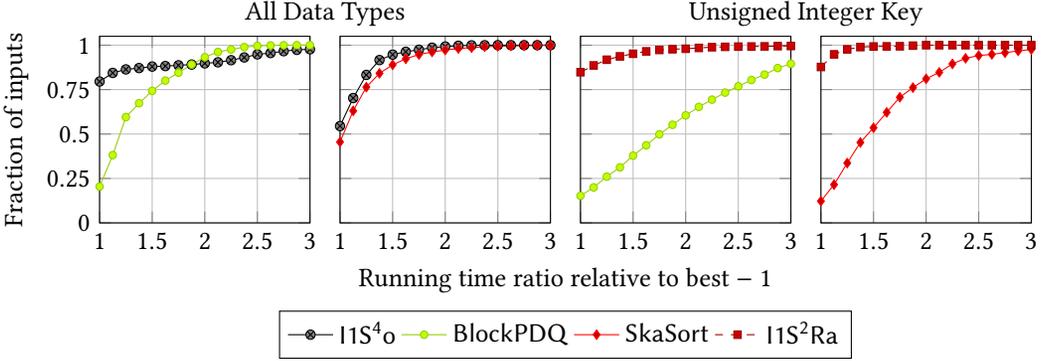

In this section, we discuss the pairwise performance profiles, shown
in \cref{fig:perf seq}, of our algorithms \compissssort and
\compissrsort to the fastest comparison-based competitor (\compspdq)
and the fastest radix sort competitor (\radixsska).

Overall, \compissrsort has a significantly better profile than
\compspdq and \radixsska.  The performance profile of \compissssort is
slightly better than the profile of \radixsska and significantly
better than the one of \compspdq. Exceptions are \distalmostsorted
inputs for which \compissssort is much slower than \compspdq.

\emph{Comparison to \compissrsort.} For the profiles containing
\compissrsort, we used only inputs with unsigned integer keys. The
performance profile of \compissrsort is significantly better than the
profile of \textbf{\radixsska}.  \compissrsort is much faster for most
of the inputs and for the remaining inputs only slightly
slower.  For example, \compissrsort sorts $84$\,\% of the
inputs faster than \radixsska.  Also, \compissrsort sorts $91$\,\%
of the inputs at least a factor of $1.25$ faster than
\radixsska.  \radixsska on the other hand sort only $34$\,\% of the
inputs at most a factor of $1.25$ faster.
The performance profile of \textbf{\compspdq} is even worse than the profile of
\radixsska.  For example, \compissrsort sorts $97$\,\% of the
inputs at least a factor of $1.25$ faster than \compspdq.
\compspdq on the other hand sort only $26$\,\% of the inputs
at most a factor of $1.25$ faster.

\emph{Comparison to \compissssort.} The performance profile of
\compissssort is in most ranges significantly better than the profile
of \textbf{\compspdq}.  For example, \compissssort sorts $79$\,\% of
the inputs faster than \compspdq.  Also, \compspdq sorts only
$60$\,\% of the inputs at least a factor of $1.25$ faster
than \compissssort whereas \compissssort sorts $86$\,\% of the
inputs at least a factor of $1.25$ faster than \compspdq.  We note
that \compissssort is significantly slower than \compspdq for some
inputs.  These inputs are \distalmostsorted inputs.  The
performance profile of \compissssort is slightly better than the
profile of \textbf{\radixsska}.  For example, \compissssort sorts
$54$\,\% of the inputs faster than \radixsska.  Also,
\compissssort (\radixsska) sorts $83$\,\% ($77$\,\%) of the
inputs at least a factor of $1.25$ faster.

%......................................................................
\subsection{Influence of the Memory Allocation Policy}\label{sec:numa allocation policy}

On NUMA
machines, access to memory attached to the local NUMA node is
faster than memory access to other nodes.  Thus, the memory access
pattern of a shared-memory algorithm may highly influence its
performance.  For example, an algorithm can greatly benefit by minimizing the
memory access of its threads to other NUMA nodes.  However, we
cannot avoid access to non-local NUMA nodes for shared-memory
sorting algorithms: For example, when the input array is distributed among
the NUMA nodes, the input- and output-position of elements may be on
different nodes.  In this case, it can be an advantage to distribute
memory access evenly across the NUMA nodes to utilize the full
memory bandwidth.  Depending on the access patterns of an algorithm, a
memory layout may suit a parallel algorithm better than another.  If
we do not consider different memory layouts of the input array in our
benchmark, the results may wrongly indicate that one algorithm is
better than another.

The memory layout of the input array depends on the \emph{NUMA
  allocation policy} of the input array and former access to the
array.
The \emph{local allocation policy} allocates memory pages at the
thread's local NUMA node if memory is available.  This memory policy
is oftentimes the default policy.  Note that after the user has
allocated memory with this policy, the actual memory pages are not
allocated until a thread accesses them the first time.  A memory page
is then allocated on the NUMA node of the accessing thread.  This
principle is called \emph{first touch}.  The \emph{interleaved
  allocation policy} pins memory pages round-robin to (a defined set
of) NUMA nodes.  The \emph{bind allocation policy} binds memory to a
defined set of NUMA nodes and \emph{preferred allocation} allocates
memory on a preferred set of NUMA nodes.  For example, a user could
create an array with the bind allocation policy such that the $i$'th
stripe of the array is pinned to NUMA node $i$.

Benchmarks of sequential algorithms usually allocate and initialize
the input array with a single thread with the default allocation
policy (local allocation).  The memory pages of the array are thus all
allocated on a single NUMA node (\emph{local arrays}).  Local arrays
are slow for many parallel algorithms because the NUMA node holding
the array becomes a bottleneck.  It is therefore recommended to use a
different layout for parallel (sorting) algorithms.  For example, the
authors of \radixraduls recommend to use an array where the $i$'th
stripe of the array is first touched by thread $i$ (\emph{striped
  array}).  Another example is \radixregion for which the authors
recommend to invoke the application with the interleaved allocation
policy.  We call arrays of those applications \emph{interleaved
  arrays}.  Orestis~and~Ross allocate for their
benchmarks~\cite{polychroniou2014comprehensive} on machines with $m$
NUMA nodes $m$ subarrays where subarray $i$ is pinned to NUMA node
$i$.

We execute the benchmark of each algorithm with the following four
input array types.
\begin{itemize}
\item For the \emph{local array}, we allocate the array with the
  function \texttt{malloc} and a single thread initializes the array.
\item For the \emph{striped array}, we allocate the array with
  \texttt{malloc} and thread $i$ initializes the $i$'th stripe of the
  input array.
\item For the \emph{interleaved array}, we activate the process-wide
  interleaved allocation policy using the Linux tool \texttt{numactl}.
\item The \emph{NUMA array}~\cite{Michael2020Numa} uses a refined
  NUMA-aware array proposed by Lorenz
  Hübschle-Schneider\footnote{\href{https://gist.github.com/lorenzhs}{https://gist.github.com/lorenzhs}}. The
  NUMA array pins the stripe $i$ of the array to NUMA node $i$. This
  approach is similar to the array used by Orestis~and~Ross except
  that the NUMA array is a continuous array.
\end{itemize}

\Cref{tab:numa comparison} shows the average slowdown of each array
type for each algorithm on our machines.  As expected, the machines
with a single CPU, \pcamd and \pcintelfour, do not benefit from NUMA
allocations.  On the NUMA machines, the local array performs
significantly worse than the other arrays.  Depending on the
algorithm, the average slowdown of the local array is a factor of up
to $1.49$ larger than the average slowdown of the
respectively best array on \pcinteltwo.  On \pcintellargefour, the
local array performs even worse: Depending on the algorithm, the
average slowdown of the local array is a factor of $1.12$ to $4.88$
larger.

The interleaved array (significantly) outperforms the other arrays for
most algorithms or shows similar slowdown factors ($\pm 0.02$) on the
NUMA machines, i.e. \pcinteltwo and \pcintellargefour.  Only
\comppaspas is on these machines with the striped array noticeable
faster than with the interleaved array.  However, \comppaspas shows
large running times in general.  On \pcinteltwo, the NUMA machine with
$32$~cores, the average slowdown ratios of the striped array and the
NUMA array to the interleaved array are relatively small (up to
$1.12$).  On \pcintellargefour, which is equipped with
$80$~cores, the average slowdown ratio of the striped array (NUMA
array) to the interleaved array increases in the worst case to
$1.44$ (to $2.83$).

Our algorithm \compiparassssort has almost the same average slowdowns
when we execute the algorithm with the interleaved or the NUMA array.
Other algorithms, e.g., our closest competitors \radixraduls and
\radixregion, are much slower on \pcintellargefour when executed with
the NUMA array.  The reason is that a thread of our algorithm
predominantly works on one stripe of the input array allocated on a
single NUMA node.

In conclusion, the local array should not be used on NUMA machines.
The interleaved array is the best array on these machines with just a
few minor exceptions.  The NUMA array and the striped array perform
better than the local array on NUMA machines and in most cases worse
than the interleaved array.  Unless stated otherwise, we report
results obtained with the interleaved input array.

\begin{table}
  \centering
  \resizebox{\textwidth}{!}{%
    \begin{tabular}{@{}l@{}|@{}c@{}|@{}c@{}|@{}c@{}|@{}c@{}}
  & \pcamd & \pcintelfour & \pcinteltwo & \pcintellargefour \\ \hline
  \begin{tabular}{l}
    \\
               \comppaspas \\
                \comppsort \\
        \comppbalancedsort \\
         \compiparassssort \\
                \compppbbs \\
    \compmyparassssaxtmann \\
                 \compptbb \\
         \compiparassrsort \\
               \radixppbbr \\
              \radixraduls \\
              \radixregion \\
    \end{tabular}
    &
    \begin{tabular}{rrrr}
      LA & IA & SA & NA \\ 
           1.01 &          1.00 &          1.00 & \textbf{1.00} \\
  \textbf{1.00} &          1.01 &          1.01 &          1.01 \\
           1.02 &          1.03 &          1.02 & \textbf{1.01} \\
           1.01 & \textbf{1.00} &          1.01 &          1.03 \\
           1.01 & \textbf{1.00} &          1.00 &          1.00 \\
           1.01 & \textbf{1.00} &          1.01 &          1.02 \\
           1.02 &          1.02 & \textbf{1.01} &          1.02 \\
           1.01 &          1.01 & \textbf{1.01} &          1.02 \\
  \textbf{1.00} &          1.01 &          1.01 &          1.00 \\
  \textbf{1.00} &          1.01 &          1.01 &          1.10 \\
           1.00 &          1.01 &          1.01 & \textbf{1.00} \\
  \end{tabular}
   & 
    \begin{tabular}{rrrr}
      LA & IA & SA & NA \\ 
  \textbf{1.00} &          1.01 &          1.01 &          1.01 \\
  \textbf{1.00} &          1.02 &          1.01 &          1.02 \\
           1.04 &          1.05 & \textbf{1.02} &          1.04 \\
           1.01 &          1.01 &          1.01 & \textbf{1.00} \\
           1.01 &          1.00 & \textbf{1.00} &          1.01 \\
  \textbf{1.00} &          1.00 &          1.01 &          1.01 \\
           1.02 & \textbf{1.01} &          1.01 &          1.02 \\
           1.01 & \textbf{1.00} &          1.01 &          1.02 \\
           1.03 & \textbf{1.01} &          1.01 &          1.02 \\
           1.08 &          1.09 &          1.09 & \textbf{1.00} \\
  \textbf{1.00} &          1.01 &          1.01 &          1.01 \\
  \end{tabular}
     & 
    \begin{tabular}{rrrr}
      LA & IA & SA & NA \\ 
  1.22 &          1.05 & \textbf{1.01} &          1.02 \\
  1.13 & \textbf{1.03} &          1.06 &          1.03 \\
  1.49 & \textbf{1.00} &          1.08 &          1.02 \\
  1.27 & \textbf{1.00} &          1.12 &          1.01 \\
  1.09 & \textbf{1.00} &          1.03 &          1.01 \\
  1.13 & \textbf{1.00} &          1.11 &          1.04 \\
  1.10 &          1.02 &          1.09 & \textbf{1.01} \\
  1.45 &          1.02 &          1.14 & \textbf{1.00} \\
  1.11 & \textbf{1.01} &          1.03 &          1.04 \\
  1.23 & \textbf{1.01} &          1.03 &          1.09 \\
  1.28 & \textbf{1.00} &          1.07 &          1.05 \\
  \end{tabular}
     & 
    \begin{tabular}{rrrr}
      LA & IA & SA & NA \\ 
  4.59 &          1.11 & \textbf{1.00} & 1.39 \\
  2.28 & \textbf{1.02} &          1.16 & 1.18 \\
  3.67 & \textbf{1.01} &          1.28 & 1.30 \\
  3.43 & \textbf{1.00} &          1.32 & 1.08 \\
  1.47 & \textbf{1.00} &          1.17 & 1.12 \\
  2.28 & \textbf{1.01} &          1.19 & 1.23 \\
  1.12 & \textbf{1.03} &          1.12 & 1.05 \\
  4.88 & \textbf{1.01} &          1.45 & 1.04 \\
  2.33 & \textbf{1.01} &          1.27 & 1.44 \\
  4.80 & \textbf{1.01} &          1.53 & 2.86 \\
  4.18 & \textbf{1.04} &          1.22 & 1.36 \\
  \end{tabular}

\end{tabular}

  }
  \caption{\label{tab:numa comparison}%
    Average slowdowns of the local array (LA), the interleaved array
    (IA), the striped array (SA), and the NUMA array (NA) for
    different parallel sorting algorithms on different machines. 
    We only consider \ulong data types with at least
    $2^{21}t$~bytes and input distribution \distuniform.}
\end{table}
  
\subsection{Evaluation of the Parallel Task Scheduler}\label{sec:results task scheduler}

The version of \compiparassssort proposed in the
conference article of this publication
(\compiparassssortnts)~\cite{axtmann2017confplace} uses a very simple
task scheduling. I.e., tasks with more than $n/t$ elements are all
executed with $t$ threads (so-called parallel tasks) and sequential
tasks are assigned to threads greedily in descending order according
to their size.  The task scheduler of \compiparassssort, described
in~\cref{sec:task scheduling}, has three advantages. First, the
number of threads processing a parallel task decreases as the size of
the task decreases.  This means that we can process small parallel
subtasks more efficiently.  Second, voluntary work sharing is used
to balance the load of sequential tasks between threads.  Finally,
thread $i$ predominantly accesses elements from
$\VarArray\oset{i n/t }{ (i + 2)n/t - 1}$ in sequential tasks and in
classification phases (see \cref{lem:seq task range,lem:par
  task range}).
Thus, the access
pattern of \compiparassssort significantly reduces memory access to
the nonlocal NUMA nodes when the striped array or the NUMA array is
used.

\Cref{tab:ips4o comparison} compares \compiparassssort with
\compiparassssortnts.
On machines with multiple NUMA nodes, i.e., \pcinteltwo and
\pcintellargefour, both algorithms are much slower when the local
array is used. This is not surprising as the input is read in this
case from a single NUMA node.  On machine \pcintellargefour,
\pcinteltwo, and \pcamd, \compiparassssort shows a slightly smaller
average slowdown than \compiparassssortnts for the same array
type. The improvements come from the voluntary work sharing.  Both
algorithms do not execute parallel subtasks as $t \ll k$.

% ------------------------------------------------------------------------------

In the remainder of this section, we discuss the results obtained on
machine \pcintellargefour. These results are perhaps the most
interesting: Compared to the other machines, on \pcintellargefour
tasks with more than $n/t$ elements occur regularly on the second
recursion level of the algorithms as the number of threads is only
slightly smaller than $k$.  Thus, both algorithms actually perform
parallel subtasks.  In contrast to \compiparassssortnts,
\compiparassssort uses thread groups whose size is proportional to the
size of parallel tasks. Thus, we expect \compiparassssort to be
faster than \compiparassssortnts for any array type.

We want to point out that the advantage of \compiparassssort is caused
by the handling of its parallel subtasks, not by the voluntary work
sharing: When no parallel subtasks are executed, the running times
do not differ much.  However, our experiments show that
\compiparassssort performs much better than \compiparassssortnts in
cases where parallel tasks occur on the second recursion level.  We
now discuss the running time improvements separately for each array
type.

With the interleaved array, \compiparassssort reports the fastest
running times. For this array, the average slowdown ratio of
\compiparassssortnts to \compiparassssort is $1.13$. For
the interleaved array, we expect that parallel subtasks oftentimes
cover multiple memory pages. Thus, both algorithms can utilize the
bandwidth of multiple NUMA nodes when executing parallel subtasks. We
assume that this is the reason that \compiparassssortnts is not much
slower than \compiparassssort with interleaved arrays. For the NUMA
array the average slowdown ratio increases to $2.54$ --
\compiparassssortnts becomes much slower.  The reason for this
slowdown is that the subarray associated with a parallel subtask will
often reside on a single NUMA node. \compiparassssortnts executes such tasks with all $t$
threads which then leads to a severe memory bottleneck. Additionally, subtasks of this task can be assigned to any of these threads.
\compiparassssort on the other hand executes the task with a thread group of appropriate size.
Threads of this thread group also process resulting subtasks (unless they are rescheduled to other threads).

Let us now compare the striped array with the NUMA array.  While
\compiparassssortnts exhibits about the same (bad) performance with
both arrays, \compiparassssort becomes $22$\,\% slower when executed
with the striped array (but still almost twice as fast as
\compiparassssortnts). A reason for the slowdown of \compiparassssort
might be that the striped array does not pin memory pages. Thus,
during the block permutation, many memory pages are moved to
other NUMA nodes. This is counterproductive since they are later
accessed by threads on yet another NUMA node.

If a local array is used, the NUMA node holding it becomes a severe
bottleneck -- both \compiparassssortnts and \compiparassssort become
several times slower.  \compiparassssort suffers less from this
bottleneck (slowdown factor $1.09$ rather than $1.51$ for
\compiparassssortnts), possibly because a thread $i$ of
\compiparassssort accesses a similar array stripe in a child task $T'$
as in a parent task $T$.  Thus, during the execution of $T$, some memory
pages used by $T'$ might be migrated to the NUMA node of $i$ (recall
that local arrays are not pinned).

In conclusion, \compiparassssort is (much) faster than
\compiparassssortnts for any array type tested here.
\compiparassssort shows the best performance for the interleaved array
and the NUMA array, with the interleaved array performing slightly better. Both arrays allocate memory pages distributed
among the NUMA nodes, and, compared to the striped array, pin the
memory pages to NUMA nodes. For these arrays, the average slowdown
ratio of \compiparassssortnts to \compiparassssort is between $1.13$
and $2.54$.

\begin{table}
  \centering
  \centering
  \resizebox{\textwidth}{!}{%
      \begin{tabular}{l|ccl|ccl|ccl|ccl}
    & \multicolumn {3}{c|}{local array} &  \multicolumn {3}{c|}{interleaved array} &  \multicolumn {3}{c|}{striped array} &  \multicolumn {3}{c}{NUMA array} \\ 
    & ips$^4$oNT & \multicolumn {2}{c|}{ips4o} & ips$^4$oNT &  \multicolumn {2}{c|}{ips4o} & ips$^4$oNT & \multicolumn {2}{c|}{ips4o} & ips$^4$oNT &  \multicolumn {2}{c}{ips4o} \\\hline
               \pcamd & 1.07 & 1.00 &  \hspace{-0.8em}{(}3.62{)} & 1.06 & 1.00 & \hspace{-0.8em}{(}3.59{)} & 1.06 & 1.00 & \hspace{-0.8em}{(}3.61{)} & 1.05 & 1.00 & \hspace{-0.8em}{(}3.67{)} \\
         \pcintelfour & 1.04 & 1.00 &  \hspace{-0.8em}{(}4.47{)} & 1.04 & 1.00 & \hspace{-0.8em}{(}4.46{)} & 1.04 & 1.00 & \hspace{-0.8em}{(}4.47{)} & 1.05 & 1.00 & \hspace{-0.8em}{(}4.44{)} \\
          \pcinteltwo & 1.03 & 1.01 &  \hspace{-0.8em}{(}5.65{)} & 1.02 & 1.01 & \hspace{-0.8em}{(}4.47{)} & 1.04 & 1.04 & \hspace{-0.8em}{(}5.01{)} & 1.02 & 1.01 & \hspace{-0.8em}{(}4.52{)} \\
    \pcintellargefour & 1.51 & 1.09 & \hspace{-0.8em}{(}17.46{)} & 1.13 & 1.00 & \hspace{-0.8em}{(}5.22{)} & 1.84 & 1.00 & \hspace{-0.8em}{(}6.90{)} & 2.54 & 1.00 & \hspace{-0.8em}{(}5.66{)} \\
\end{tabular}

  }
  \caption{\label{tab:ips4o comparison}%
    This table shows the average slowdown of \compiparassssort and
    \compiparassssortnts to the best of both algorithms for different
    array types and machines. The
    numbers in parentheses show the average running times of
    \compiparassssort divided by $n/t \log_2 n$ in
    nanoseconds.  We only consider \ulong data types with at least
    $2^{21}t$~bytes and input distribution \distuniform.}
\end{table}

%......................................................................
\subsection{Parallel Algorithms}\label{sec:parallel comparison}

In this section, we compare parallel algorithms for different
machines, input distributions, input sizes, and data types.  We begin
with a comparison of the average slowdowns of \compiparassssort,
\compiparassrsort, and their competitors for ten input distributions
executed with six different data types (see \cref{sec:par slowdown}).
This gives a first general view of the performance of our algorithms
as the presented results are aggregated across all machines.
Afterwards, we compare the algorithms for input distribution
\distuniform with data type \ulong on different machines: We consider
scaling with input sizes in \cref{sec:par uniform input} and scaling
with the number of utilized cores in \cref{sec:par uniform speedup}.
Then, we discuss in \cref{sec:par distr types} the running times for
an interesting set of input distributions and data types, again by
scaling the input size. In \cref{sec:par perf profiles}, we discuss
the performance profiles of our algorithms and their most promising
competitors. Finally, we separately compare \compiparassssort to
\imsdradix, which is only implemented in a very explorative manner and
thus only works in some special cases (see \cref{sec:comp inpl MSD
  radix sort}).

\subsubsection{Comparison of Average Slowdowns}\label{sec:par slowdown}

\begin{table}
  \resizebox*{!}{0.93\textheight}{

\begin{tabular}{ll|rrrrrr|rrrrrrr}
    Type
  & Distribution
  & \rotatebox[origin=c]{90}{\compiparassssort} 
  &  \rotatebox[origin=c]{90}{\compppbbs}
  & \rotatebox[origin=c]{90}{\compmyparassssaxtmann} 
  & \rotatebox[origin=c]{90}{\comppsort}
  & \rotatebox[origin=c]{90}{\comppbalancedsort} 
  & \rotatebox[origin=c]{90}{\compptbb} 
  & \rotatebox[origin=c]{90}{\radixregion}  
  & \rotatebox[origin=c]{90}{\radixppbbr}
  & \rotatebox[origin=c]{90}{\radixraduls}
  & \rotatebox[origin=c]{90}{\comppaspas}
  & \rotatebox[origin=c]{90}{\compiparassrsort} \\\hline
  \double &        \distsorted &          1.42 & 10.96 & 2.02 & 15.47 &  13.36 & \textbf{1.06} &  &  &  & 42.23 &  \\
  \double & \distreversesorted & \textbf{1.06} &  1.34 & 1.98 &  1.76 &  11.00 &          3.01 &  &  &  &  5.34 &  \\
  \double &          \distones &          1.54 & 12.83 & 1.80 & 14.55 & 166.67 & \textbf{1.06} &  &  &  & 41.78 &  \\

  \hline\hline
  
  \double &            \distexpo & \textbf{1.00} & 1.82 & 1.97 & 2.60 & 3.20 & 10.77 &  &  &  & 4.97 &  \\
  \double &            \distzipf & \textbf{1.00} & 1.96 & 2.12 & 2.79 & 3.55 & 11.56 &  &  &  & 5.33 &  \\
  \double &  \distduplicatesroot & \textbf{1.00} & 1.54 & 2.22 & 2.52 & 3.88 &  5.54 &  &  &  & 6.28 &  \\
  \double & \distduplicatestwice & \textbf{1.00} & 1.93 & 1.88 & 2.45 & 2.99 &  5.52 &  &  &  & 4.44 &  \\
  \double & \distduplicateseight & \textbf{1.00} & 1.82 & 2.01 & 2.48 & 3.19 & 10.37 &  &  &  & 5.02 &  \\
  \double &    \distalmostsorted & \textbf{1.00} & 1.73 & 2.40 & 5.12 & 2.18 &  3.54 &  &  &  & 6.37 &  \\
  \double &         \distuniform & \textbf{1.00} & 2.00 & 1.85 & 2.53 & 2.99 &  9.16 &  &  &  & 4.39 &  \\

  \hline
  Total  & &

  \textbf{1.00} & 1.82 & 2.06 & 2.83 & 3.10 & 7.46 &  &  &  & 5.21 &  \\

  Rank & &
  1 & 2 & 3 & 4 & 5 & 7 &  &  &  & 6 &  \\\hline\hline
  
  \ulong &        \distsorted &          1.45 & 10.56 & 1.80 & 15.65 &  13.50 & \textbf{1.09} & 6.72 & 56.24 & 33.08 &  & 8.83 \\
  \ulong & \distreversesorted & \textbf{1.17} &  1.42 & 2.23 &  2.01 &  12.27 &          3.40 & 1.34 &  8.07 &  4.65 &  & 1.76 \\
  \ulong &          \distones &          1.69 & 13.58 & 1.87 & 15.02 & 171.86 & \textbf{1.13} & 1.36 & 51.61 & 32.50 &  & 1.16 \\

  \hline\hline
  
  \ulong &            \distexpo & \textbf{1.04} & 1.74 & 2.10 & 2.62 & 3.41 & 10.38 & 1.79 &  1.58 & 2.58 &  &          1.20 \\
  \ulong &            \distzipf & \textbf{1.00} & 1.82 & 2.16 & 2.69 & 3.60 & 10.48 & 1.61 & 16.80 & 6.04 &  &          1.68 \\
  \ulong &  \distduplicatesroot & \textbf{1.00} & 1.47 & 2.24 & 2.52 & 3.84 &  5.78 & 1.59 &  9.89 & 7.00 &  &          1.54 \\
  \ulong & \distduplicatestwice & \textbf{1.07} & 1.91 & 2.04 & 2.54 & 3.20 &  5.83 & 1.30 & 10.00 & 3.89 &  &          1.34 \\
  \ulong & \distduplicateseight & \textbf{1.02} & 1.69 & 2.06 & 2.42 & 3.25 &  9.54 & 1.37 & 12.45 & 5.00 &  &          1.44 \\
  \ulong &    \distalmostsorted & \textbf{1.11} & 1.88 & 2.73 & 5.75 & 2.54 &  4.15 & 1.36 &  9.84 & 5.87 &  &          1.55 \\
  \ulong &         \distuniform &          1.13 & 2.10 & 2.14 & 2.80 & 3.32 &  9.57 & 1.59 &  1.41 & 1.49 &  & \textbf{1.03} \\

  \hline
  Total  & &

  \textbf{1.05} & 1.79 & 2.20 & 2.91 & 3.28 & 7.54 & 1.51 & 6.17 & 4.07 &  & 1.38 \\

  Rank & &
  1 & 4 & 5 & 6 & 7 & 10 & 3 & 9 & 8 &  & 2 \\\hline\hline
  
  \uint &        \distsorted & \textbf{1.77} & 10.03 & 2.77 & 11.64 &  14.68 &          1.91 & 5.28 &  7.86 &  &  &          4.98 \\
  \uint & \distreversesorted &          1.51 &  1.84 & 2.46 &  2.03 &  11.96 &          5.17 & 1.22 &  1.44 &  &  & \textbf{1.17} \\
  \uint &          \distones &          1.59 & 15.94 & 1.95 & 19.35 & 286.17 & \textbf{1.18} & 1.50 & 73.11 &  &  &          1.20 \\

  \hline\hline
  
  \uint &            \distexpo &          1.31 & 2.85 & 2.34 & 3.68 & 4.55 & 17.62 & 1.57 & 2.02 &  &  & \textbf{1.02} \\
  \uint &            \distzipf & \textbf{1.05} & 2.54 & 2.06 & 3.22 & 4.05 & 15.68 & 1.33 & 6.39 &  &  &          1.41 \\
  \uint &  \distduplicatesroot & \textbf{1.09} & 1.78 & 2.26 & 2.62 & 3.92 &  6.16 & 1.37 & 7.50 &  &  &          1.42 \\
  \uint & \distduplicatestwice &          1.40 & 3.18 & 2.32 & 3.59 & 4.35 &  9.10 & 1.24 & 1.83 &  &  & \textbf{1.02} \\
  \uint & \distduplicateseight &          1.23 & 2.84 & 2.26 & 3.41 & 4.24 & 16.24 & 1.33 & 1.84 &  &  & \textbf{1.08} \\
  \uint &    \distalmostsorted &          1.38 & 2.08 & 2.63 & 5.66 & 3.22 &  4.54 & 1.32 & 1.62 &  &  & \textbf{1.08} \\
  \uint &         \distuniform &          1.41 & 3.26 & 2.28 & 3.68 & 4.45 & 14.52 & 1.36 & 1.61 &  &  & \textbf{1.03} \\

  \hline
  Total  & &

  1.26 & 2.59 & 2.30 & 3.60 & 4.09 & 10.75 & 1.36 & 2.49 &  &  & \textbf{1.14} \\

  Rank & &
  2 & 6 & 4 & 7 & 8 & 9 & 3 & 5 &  &  & 1 \\\hline\hline
  
  \pair &        \distsorted &          1.39 &  9.38 & 1.82 & 15.05 &  15.50 & \textbf{1.03} & 5.75 & 20.15 & 52.30 &  & 8.02 \\
  \pair & \distreversesorted & \textbf{1.09} &  1.47 & 2.06 &  2.22 &  10.46 &          3.15 & 1.35 &  3.21 &  8.24 &  & 1.77 \\
  \pair &          \distones &          1.66 & 14.10 & 1.77 & 15.21 & 118.30 & \textbf{1.08} & 1.21 & 11.71 & 54.52 &  & 1.16 \\

  \hline\hline
  
  \pair &            \distexpo &          1.12 & 1.77 & 2.22 & 2.76 & 3.09 & 6.92 & 1.92 & \textbf{1.07} &  9.52 &  &          1.39 \\
  \pair &            \distzipf & \textbf{1.00} & 1.62 & 2.04 & 2.53 & 2.79 & 6.30 & 1.62 &          7.35 &  9.87 &  &          1.77 \\
  \pair &  \distduplicatesroot & \textbf{1.01} & 1.58 & 2.08 & 2.81 & 3.84 & 4.88 & 1.58 &          4.35 & 11.76 &  &          1.52 \\
  \pair & \distduplicatestwice & \textbf{1.02} & 1.67 & 2.02 & 2.44 & 2.96 & 4.10 & 1.43 &          4.88 &  7.54 &  &          1.48 \\
  \pair & \distduplicateseight & \textbf{1.02} & 1.59 & 2.05 & 2.41 & 2.83 & 6.01 & 1.40 &          6.98 &  8.81 &  &          1.57 \\
  \pair &    \distalmostsorted & \textbf{1.05} & 1.95 & 2.69 & 5.67 & 3.24 & 3.88 & 1.37 &          4.27 & 10.94 &  &          1.65 \\
  \pair &         \distuniform &          1.08 & 1.81 & 2.12 & 2.62 & 2.93 & 6.15 & 1.67 &          1.20 &  5.36 &  & \textbf{1.04} \\

  \hline
  Total  & &

  \textbf{1.04} & 1.71 & 2.16 & 2.90 & 3.08 & 5.35 & 1.56 & 3.46 & 8.87 &  & 1.47 \\\hline\hline

  Rank & &
  1 & 4 & 5 & 6 & 7 & 9 & 3 & 8 & 10 &  & 2 \\
  
  \quartet & \distuniform & \textbf{1.01} & 1.29 & 2.08 & 2.40 & 2.93 & 4.42 &  &  &  &  &  \\

  \hline

  Rank & &
  1 & 2 & 3 & 4 & 5 & 6 &  &  &  &  &  \\\hline\hline
  
  \bytes & \distuniform & \textbf{1.05} & 1.14 & 2.14 & 2.35 & 3.18 & 3.55 &  &  &  &  &  \\

  \hline

  Rank & &
  1 & 2 & 3 & 4 & 5 & 6 &  &  &  &  &  \\\hline\hline
\end{tabular}

  }
  \caption{ Average slowdowns of parallel algorithms for different
    data types and input distributions.  The slowdowns average
    over the machines and input sizes with at least $2^{21}t$
    bytes.
  }
  \label{tab:slowdown par}
\end{table}

\Cref{tab:slowdown par} shows average slowdowns of parallel algorithms
for different data types and input distributions aggregated over all
machines and input sizes with at least $2^{21}t$ bytes.
\ifarxiv For the average slowdowns separated by machine, we
refer to \cref{app:more measurements},
\cref{tab:slowdown par 128,tab:slowdown par
  132,tab:slowdown par 133,tab:slowdown par 135}.
\else For the average slowdowns separated by machine, we
refer to the extended version of this paper~\cite{axtmann2020ips4oarxiv}.
\fi
In this section, an \emph{instance} describes the inputs of a
specific data type and input distribution.
We say that ``algorithm A is faster than algorithm B
(by a factor of C) for some instances'' if the average slowdown
of B is larger than the average slowdown of A (by a factor of C) for these instances.

Overall, the results
show that \compiparassssort is much faster than its competitors in
most cases except some ``easy'' instances and some instances with \uint data
types.  Except for some \uint instances, \compiparassssort is even
significantly faster than its fastest radix sort competitor
\radixregion.  This indicates that parallel sorting algorithms are
memory bound for most inputs, except for data types that only have a
few bytes.  In most cases, \compiparassssort also outperforms our
radix sorter \compiparassrsort.  \compiparassrsort is faster
for some instances with \uint data types and, as expected,
\compiparassrsort is faster for \distuniform instances.
\compiparassrsort has a better ranking than our fastest in-place radix
sort competitor \radixregion.  Thus, our approach of sorting data with
parallel block permutations seems to perform better than the
graph-based approach of \radixregion.\\

\emph{Comparison to \compiparassssort.}
\compiparassssort is the fastest algorithm for $30$ out of $42$ instances.
\compiparassssort is outperformed for $8$ instances having
``easy'' input distributions, i.e, \distsorted,
\distreversesorted and \distones. For now on, we consider only these instances:
\compptbb detects \distsorted and
\distones inputs as sorted and returns immediately. \radixregion
detects that the elements of \distones inputs only have zero bits, and
thus, also return immediately for \distones inputs. It is therefore
not surprising that \compptbb and \radixregion sort easy inputs very
fast. \compptbb (\radixregion) is for $4$ (for $3$) \distones 
instances better than \compiparassssort. \compptbb is also better for
$3$ \distsorted instances, i.e., with \double, \ulong, and
\pair data types. Also, \radixregion is faster
than \compiparassssort for the \distreversesorted \uint instance. Our
algorithm also detects these  instances but with a slightly
larger overhead.

In this paragraph, we do not consider ``easy'' instances. \compiparassssort
is significantly faster than our competitors for $23$ out of $30$ instances ($>1.15$). For $3$
instances, \compiparassssort performs similar ($\pm 0.06$) to
\textbf{\radixregion} (\distalmostsorted distributed \uint instance and
\distuniform distributed \uint instance) and \textbf{\radixppbbr} (\distexpo
distributed \pair instance). \radixregion is the only competitor that is noticeably faster than \compiparassssort, at least for one instance, i.e., \distduplicatestwice distributed \uint
inputs (factor $ 1.13$).
Overall, \compiparassssort is faster than its respectively fastest competitor
by a factor of $1.2$, $1.4$, $1.6$, and $1.8$ for
$22$, $13$, $8$, and $5$ noneasy  instances, respectively. If we
only consider comparison-based competitors, \compiparassssort is
faster by a factor of $1.2$, $1.4$, $1.6$, and $1.8$ for $29$, $28$,
$22$, and $10$ noneasy  instances, respectively. The values become
even better when we only consider in-place comparison-based
competitors. In this case, the \compiparassssort is faster by a factor of
$2.15$ for all noneasy  instances.

\compiparassssort is much faster than \textbf{\compmyparassssaxtmann}.
The only difference between these
algorithms is that \compiparassssort implements the partitioning
routine in-place whereas \compmyparassssaxtmann is non-in-place.  We
note that the algorithms share most of their code, even the decision
tree is the same.  The reason why \compmyparassssaxtmann is slower
than \compiparassssort is that \compiparassssort is more cache
efficient than \compmyparassssaxtmann: For example, \compmyparassssaxtmann
has about $46$\,\% more L3-cache misses than \compiparassssort for
\distuniform distributed \ulong inputs with $2^{27}$ elements whereas
the number of instructions and the number of branch (misses) of
\compmyparassssaxtmann are similar to the ones of \compiparassssort.
The sequential results presented in \cref{sec:sequential comparison}
support this conjecture as the gap between the sequential versions is
smaller than the gap between the parallel versions.

\emph{Comparison to \compiparassrsort.}
Our in-place radix sorter \compiparassrsort performs slightly better than our
fastest competitor, \textbf{\radixregion}. \compiparassrsort is faster than \radixregion for $11$ out of
$21$ noneasy instances.  In particular, \compiparassrsort is faster by a
factor of $1.2$, $1.4$, and $1.6$ for $9$, $4$, and $1$ noneasy instances
and the factor is never smaller than $0.8$.

\compiparassrsort outperforms \textbf{\compiparassssort} for instances
with \uint data types and some \distuniform distributed instances.
For \uint instances, \compiparassrsort is faster than \compiparassssort 
by a factor of $1.37$.  Interestingly, \compiparassrsort is not
much faster than \compiparassssort for \distuniform instances with
more than $32$-bit elements.  This indicates that the evaluation of
the branchless decision tree is not a limiting factor for these data
types in \compiparassssort.  For the remaining instances (data types with more
than $32$-bit elements and instances which noneasy distributions other than
\distuniform) \compiparassssort is significantly faster than
\compiparassrsort.
\\

From now on, we do not present results for \comppaspas, \compptbb,
\compmyparassssaxtmann, and \comppsort.  In regard to
non-in-place comparison-based competitors, the
algorithms \comppaspas, \compmyparassssaxtmann, and \comppsort perform
worse than \compppbbs.  For non-in-place comparison-based competitors,
the parallel quicksort algorithm \compptbb is
for noneasy  instances slower than the quicksort implementation
\comppbalancedsort.

\input{extern/ips4o-benchmark-suite-plots/benchmark/running_times/running_time_random_parallel_pretty.tex}

\subsubsection{Running Times for \distuniform Input}\label{sec:par uniform input}

In this section,
we compare \compiparassssort and \compiparassrsort to their closest
parallel competitors for \distuniform distributed \ulong inputs.
\Cref{fig:rt rand par} depicts the running times separately for each
machine.  The results of the algorithms obtained for \double inputs
are similar to the running times obtained for \ulong inputs.  We
decided to present results for \ulong inputs as our closest parallel
competitors for data types with ``primitive'' keys, i.e.,
\radixregion, \radixppbbr, and \radixraduls, do not support \double
inputs.

We outperform all comparison-based algorithms significantly
for medium and large input sizes, e.g., by a factor of $1.49$ to
$2.45$ for the largest inputs depending on the machine.  For in-place
competitors, the factor is even $2.55$ to $3.71$.  In general, all
competitors are very inefficient for small input sizes except the
non-in-place competitors \compppbbs and \radixppbbr.  However, the
performance of \compppbbs and \radixppbbr significantly decreases for
larger inputs.  Exploiting integers in \compiparassrsort slightly
improves the performance for medium and large input sizes compared to
\compiparassssort.  For small input sizes, exploiting integers makes
\compiparassrsort more efficient than \compiparassssort.  Our radix
sorters \radixraduls and \radixregion are only competitive for large
input sizes. Still, they are very inefficient even for these input
sizes on \pcintellargefour, our largest machine. In particular, they
are $2.72$ respectively $3.08$ times slower than \compiparassrsort for
the largest input size on this machine.

We sort twice as much data as our
non-in-place competitors (\compppbbs, \radixraduls, and \radixppbbr)
which run out of memory for $2^{32}$ elements on \pcamd.  Also, the
results in \cref{tab:slowdown par}, \cref{sec:par slowdown}, show that
inputs with \ulong \distuniform inputs are ``best case'' inputs
for \radixraduls.  Other input distributions and data types are sorted
by \radixraduls much less efficient.

\input{extern/ips4o-benchmark-suite-plots/benchmark/running_times/speedup_random_parallel_pretty.tex}

\subsubsection{Speedup Comparison and Strong Scaling}\label{sec:par uniform speedup}

The goal of the speedup benchmark is to examine the performance of
the parallel algorithms with increasing availability of cores.
Benchmarks with $2i$ threads are executed on the first $i$ cores,
starting at the first NUMA node until it is completely used. Then we
continue using the cores of the next NUMA node, and so on.  Here, we
mean by cores ``physical cores'' that run two hardware threads on our
machines and we use NUMA nodes as a synonym for CPUs.\footnote{Many
Linux tools interpret a CPU with two hardware threads per core as two
distinct NUMA nodes -- one contains the first hardware thread of each
core and the other contains the second hardware threads of each core.%
}  In result, the benchmark always takes advantage
of the ``full capacity'' of a core with hyper-threading.  Preliminary
experiments have shown that all algorithms are slowed down when we use
only one thread per core.

% ------------------------------------------------------------------------------

\Cref{fig:par speedup} depicts the speedup of parallel algorithms
executed on different numbers of cores relative to our sequential
implementation \compissssort on our machines for \distuniform inputs.
We first compare our algorithms to the non-in-place radix sorter
\textbf{\radixraduls}.  This competitor is fast for \distuniform
inputs but it is slow for inputs with skewed key distributions and
inputs with duplicated keys (see \cref{tab:slowdown par} in
\cref{sec:par slowdown}).  On the machines with one CPU, \pcamd and
\pcintelfour, \radixraduls is faster when we use only a fraction of
the available cores.
When we further increase the available cores on these machines, the
speedup of \radixraduls stagnates and our algorithms,
\compiparassssort and \compiparassrsort, catch up until they have
about the same speedup.
\radixraduls also outperforms all algorithms on
our machine with four CPUs, \pcintellargefour, when the algorithms use
only one CPU.
On the same machine, the performance of \radixraduls stagnates when we expand the algorithm to more than one CPU. When \radixraduls uses all CPUs, it is even a factor of $2.54$ slower than our algorithm \compiparassssort.
We have seen the same performance characteristics when we executed \compiparassssortnts on this machine.
\compiparassssort solved this problem of \compiparassssortnts with a more sophisticated memory and task management.
Thus, we conclude that the same problems also result in performance problems for \radixraduls.
Our
algorithms \compiparassssort and \compiparassrsort use the memory on
this machine more efficiently and do not get memory-bound  -- the
speedup of our algorithms increases on \pcintellargefour linearly.

% ------------------------------------------------------------------------------

The in-place radix sorter \textbf{\radixregion} seems to have similar problems as \radixraduls
on \pcintellargefour.
Even worse, the
speedup of \radixregion stagnates on three out of four machines when
the available cores increase.  When all cores are used, the speedup of
\radixregion is a factor of $1.11$ to $2.70$ smaller than the speedup
of \compiparassrsort.  On three out of four machines, our radix sorter
\textbf{\compiparassrsort} has a larger speedup than our samplesort
algorithm \textbf{\compiparassssort} when we use only a few cores.  For more
cores, their speedups converge on two machines, even though \compiparassrsort performs
significantly fewer instructions.

On the machines with one CPU, \pcamd and \pcintelfour,
\compiparassssort has a speedup of $8.37$ respectively $40.92$.  This
is a factor of $1.46$ respectively $1.85$ more than the fastest
\textbf{comparison-based competitor}.  On the machine with four CPUs,
\pcintellargefour, and on the machine with two CPUs, \pcinteltwo, the
speedup of \compiparassssort is $20.91$ respectively $17.49$. This is even a
factor of $2.27$ respectively $2.17$ more than the fastest comparison-based
competitor.

In conclusion, our in-place algorithms outperform their
comparison-based competitors significantly on all machines
independently of the number of assigned cores. For example,
\compiparassssort yields a speedup of $40.92$ on the machine
\pcintelfour whereas \compppbbs only obtains a speedup of $22.17$. As
expected, the fastest competitors for the (\distuniform) input used in
this experiment are radix sorters.  The fastest 
radix sort competitor, non-in-place \radixraduls, starts with a very large speedup
when only a few cores are in use.  For more cores, \radixraduls
remains faster than our algorithms on one machine (\pcinteltwo).  On
two machines (\pcintelfour and \pcamd), the speedup of \radixraduls
converges to the speedups of our algorithms.  And, on our largest
machine with four CPUs (\pcintellargefour), the memory management of
\radixraduls seems to be not practical at all.  On this machine,
\radixraduls is even a factor of $2.54$ slower than \compiparassssort.
The in-place radix sort competitor
\radixregion is in all cases significantly slower than our algorithms.
The speedup of \compiparassrsort is larger
than the one of \compiparassssort when they
use only a few cores of the machine.  However, the speedup levels out
when the number of cores increases in most cases.

\subsubsection{Input Distributions and Data Types}\label{sec:par distr types}

\begin{figure}[tbp]
\input{extern/ips4o-benchmark-suite-plots/benchmark/running_times/running_time_parallel_allgen_one_machine_pretty_i10pc136.tex}
  \caption{
    Running times of parallel algorithms on different input distributions and data types of size $D$ executed on machine \pcintelfour.
    The radix sorters \radixppbbr, \radixraduls, \radixregion, and \compiparassrsort does not support the data types \double and \bytes.
  }
  \label{fig:par rt distr types}
\end{figure}

In this section, we compare our algorithms to our competitors for different input distributions and data types
by scaling the input size.  We show results of \distuniform inputs for
the data types \double, \pair, and \bytes. For a discussion of
\distuniform distributed \ulong data types, we refer to \cref{sec:par
  uniform input}. For the remaining input distributions, we use the
data type \ulong as a convenient example: In contrast to \double, \ulong is
supported by all algorithms in \cref{fig:par rt distr
  types}.  Additionally, we assume that \ulong is more interesting
than \uint in practice.  We decided to present in \cref{fig:par rt
  distr types} results obtained on machine \pcintelfour as our
competitors have the smallest absolute running time on this
machine.
\ifarxiv For more
details, we refer to \cref{fig:par rt distr types 128,fig:par rt distr
  types 132,fig:par rt distr types 133,fig:par rt distr types 135} in
\cref{app:more measurements} which report the results separately for each machine.
\else For more
details, we refer to the extended version of this
paper~\cite{axtmann2020ips4oarxiv} which reports the
results separately for each machine.
\fi

For many inputs, our \compiparassssort is faster than \compiparassrsort.  For most inputs,
\compiparassssort (and to some extend \compiparassrsort) is much
faster than \radixregion, our closest competitor.  For example, \compiparassssort is up to a
factor of $1.61$ faster for the largest inputs ($n=2^{37}/D$) and up
to a factor of $1.78$ for inputs of medium size ($n=2^{29}/D$).
The results show that radix sorters are often slow for
inputs with many duplicates or skewed key distributions (i.e.,
\distzipf, \distexpo, \distduplicateseight, \distduplicatesroot).
Yet, our algorithm seems to be the least affected by this.
Our algorithms outperform their comparison-based
competitors significantly for all input distributions and data types with $n\geq 2^{28}/D$.
For example, \compiparassssort outperforms \compppbbs by a
factor of $1.25$ to $2.20$ for the largest inputs.
Only for small inputs, where the algorithms are inefficient anyway, our algorithms are consistently outperformed by one algorithm (non-in-place \compppbbs).
The remainder of this section compares our algorithms and their competitors in detail.\\

The non-in-place comparison-based \textbf{\compppbbs} is slower than \compiparassssort for small inputs ($n\leq 2^{27}/D$). We note that all
algorithms are inefficient for these small inputs.  However, for
inputs where the algorithms become efficient and for large inputs,
\compiparassssort significantly outperforms \compppbbs. For example,
\compppbbs is a factor of $1.25$ to $2.20$ slower than
\compiparassssort for the largest input size.  The difference between
\compppbbs and \compiparassssort is the smallest for \bytes inputs.
This input has very large elements which are moved only twice by
\compppbbs due to its $\sqrt{n}$-way partitioning strategy.  We see
this as an important turning point.  While the previous
state-of-the-art comparison-based algorithm worked non-in-place, it is
now robustly outperformed by our in-place algorithms for inputs that
are sorted efficiently by parallel comparison-based sorting
algorithms.

The in-place comparison-based \textbf{\comppbalancedsort} is significantly slower than our
algorithms for all inputs.  For example, \comppbalancedsort
is a factor of $2.46$ to $3.87$ slower than \compiparassssort for the
largest input size. We see this improvement as a major contribution of our paper.

The non-in-place radix sorter \textbf{\radixppbbr} is tremendously slow for
all inputs with skewed inputs and inputs with identical keys.
In particular, its running times exceed the limits of \cref{fig:par rt distr
  types} for \distalmostsorted, \distduplicatesroot,
\distduplicatestwice, and \distzipf inputs.
Exceptions are
\distuniform inputs with \pair data type: For these inputs, \radixppbbr is
faster than our algorithms for small input sizes and performs similar
for medium and large inputs. However, this advantage disappears for other uniformly distributed
inputs (see \cref{tab:slowdown par} in \cref{sec:par slowdown}).
  
The non-in-place radix sorter \textbf{\radixraduls} is a factor of $2.20$ to $2.72$ slower than
\compiparassssort for the largest input size. For smaller inputs, its
performance is even worse for almost all inputs.

Even though \compiparassrsort
outperforms the in-place radix sorter \textbf{\radixregion} for almost all inputs,
\compiparassssort is even faster.  Thus, we concentrate our analysis
on comparing \radixregion to \compiparassssort rather than
\compiparassrsort.
For input data types supported by \radixregion, i.e., integer keys, it
is our closest competitor.  Overall, we see that the efficiency of
\radixregion slightly degenerates for inputs larger than $n>2^{32}$.
The performance of \compiparassssort remains the same for these large
input sizes.  \radixregion performs the best for \distalmostsorted and
\distduplicatestwice distributed inputs.  For these inputs,
\radixregion is competitive to \compiparassssort in most cases.
However, \radixregion performs much worse than \compiparassssort for
the remaining inputs, e.g., random inputs (\distuniform),
skewed inputs (\distexpo and \distzipf), and inputs with many
duplicates (e.g., \distduplicatesroot).  For these distributions, \radixregion is
slower than \compiparassssort by a factor of $1.17$ to $1.61$ for the
largest input size and becomes even less efficient for smaller inputs,
e.g., \radixregion is slower than \compiparassssort by factors
of $1.29$ to $1.68$ for $n=2^{27}$.

\textbf{\compiparassssort} is
competitive or faster than \textbf{\compiparassrsort} for all inputs.
\compiparassssort and \compiparassrsort perform similarly for inputs
of medium input size which are \distuniform, \distduplicatestwice,
\distduplicatesroot, and \distalmostsorted distributed.  Still, for
these inputs, the performance of \compiparassrsort (significantly)
decreases for large inputs ($n>2^{32}$) in most cases.  For inputs
with very skewed key distributions, i.e., \distexpo and \distzipf,
\compiparassssort is significantly faster than \compiparassrsort.

\begin{figure}[tbp]
  \begin{tikzpicture}
    \begin{groupplot}[
      group style={
        group size=4 by 1,
        x descriptions at=edge bottom,
        y descriptions at=edge left,
        horizontal sep=0.4cm,
      },
      width=0.315\textwidth,
      height=0.2\textheight,
      xmin=1,
      xmax=3,
      ymin=0.0,
      ymax=1.05,
      plotstyleparallelperf,
      xmajorgrids=true,
      ymajorgrids=true,
      subtitle/.style={title=#1},
      ]
      
      \nextgroupplot[
      ylabel={Fraction of inputs},
      ]
      \addplot coordinates { (1,0.975793) (1.25,0.997496) (1.5,0.999165) (1.75,1) (2,1) (2.25,1) (2.5,1) (2.75,1) (3,1) (4.5,1) };
      \addlegendentry{\compiparassssort};
      \addplot coordinates { (1,0.024207) (1.25,0.123539) (1.5,0.338898) (1.75,0.516694) (2,0.656928) (2.25,0.77379) (2.5,0.872287) (2.75,0.912354) (3,0.937396) (4.5,1) };
      \addlegendentry{\compppbbs};
      \addplot coordinates { (4.5,1) };
      \addlegendentry{\comppbalancedsort};
      \addplot coordinates { (4.5,1) };
      \addlegendentry{\radixregion};
      \addplot coordinates { (4.5,1) };
      \addlegendentry{\compiparassrsort};

      \legend{}
      
      \coordinate (c1) at (rel axis cs:0,1);
      
      \nextgroupplot[
      every legend/.append style={at=(ticklabel cs:1.1)},
      legend style={at={($(0,0)+(1cm,1cm)$)},legend columns=6,fill=none,draw=black,anchor=center,align=center},
      ]
      \addplot coordinates { (1,0.998331) (1.25,1) (1.5,1) (1.75,1) (2,1) (2.25,1) (2.5,1) (2.75,1) (3,1) (4.5,1) };
      \addlegendentry{\compiparassssort};
      \addplot coordinates { (4.5,1) };
      \addlegendentry{\compppbbs};
      \addplot coordinates { (1,0.00166945) (1.25,0.00500835) (1.5,0.0175292) (1.75,0.033389) (2,0.0542571) (2.25,0.102671) (2.5,0.199499) (2.75,0.373957) (3,0.530885) (4.5,1) };
      \addlegendentry{\comppbalancedsort};
      \addplot coordinates { (4.5,1) };
      \addlegendentry{\radixregion};
      \addplot coordinates { (4.5,1) };
      \addlegendentry{\compiparassrsort};

      \legend{}
      
      \nextgroupplot[
      ]
      \addplot coordinates { (1,0.741667) (1.25,0.913095) (1.5,0.969048) (1.75,0.986905) (2,0.99881) (2.25,1) (2.5,1) (2.75,1) (3,1) (4.5,1) };
      \addlegendentry{\compiparassssort};
      \addplot coordinates { (4.5,1) };
      \addlegendentry{\compppbbs};
      \addplot coordinates { (4.5,1) };
      \addlegendentry{\comppbalancedsort};
      \addplot coordinates { (1,0.258333) (1.25,0.554762) (1.5,0.710714) (1.75,0.791667) (2,0.829762) (2.25,0.852381) (2.5,0.882143) (2.75,0.913095) (3,0.936905) (4.5,1) };
      \addlegendentry{\radixregion};
      \addplot coordinates { (4.5,1) };
      \addlegendentry{\compiparassrsort};

      \legend{}
      
      \nextgroupplot[
      every legend/.append style={at=(ticklabel cs:1.1)},
      legend style={at={($(0,0)+(1cm,1cm)$)},legend columns=6,fill=none,draw=black,anchor=center,align=center},
      legend to name=legendparpair
      ]
      \addplot coordinates { (1,0.62381) (1.25,0.840476) (1.5,0.941667) (1.75,0.983333) (2,0.996429) (2.25,1) (2.5,1) (2.75,1) (3,1) (4.5,1) };
      \addlegendentry{\compiparassssort};
      \addplot coordinates { (4.5,1) };
      \addlegendentry{\compppbbs};
      \addplot coordinates { (4.5,1) };
      \addlegendentry{\comppbalancedsort};
      \addplot coordinates { (4.5,1) };
      \addlegendentry{\radixregion};
      \addplot coordinates { (1,0.37619) (1.25,0.685714) (1.5,0.846429) (1.75,0.869048) (2,0.880952) (2.25,0.890476) (2.5,0.90119) (2.75,0.915476) (3,0.919048) (4.5,1) };
      \addlegendentry{\compiparassrsort};

      \coordinate (c2) at (rel axis cs:1,1);

    \end{groupplot}

    \coordinate (c23s) at ($(group c2r1.south)!.5!(group c3r1.south)$);
    \node[yshift=-0.75cm] at (c23s) {{Running time ratio relative to best}};

    \coordinate (c12n) at ($(group c1r1.north)!.5!(group c2r1.north)$);
    \node[yshift=0.3cm] at (c12n) {{All Data Types}};

    \coordinate (c34n) at ($(group c3r1.north)!.5!(group c4r1.north)$);
    \node[yshift=0.3cm] at (c34n) {{Unsigned Integer Key}};

    \coordinate (c3) at ($(c1)!.5!(c2)$);
    \node[below] at (c3 |- current bounding box.south)
    {\pgfplotslegendfromname{legendparpair}};
  \end{tikzpicture}
  \caption{Pairwise performance profiles of \compiparassssort to
    \compppbbs, \comppbalancedsort, \radixregion, and
    \compiparassrsort. The performance plots with \compppbbs and
    \comppbalancedsort use all data types. The performance plots with
    \radixregion and \compiparassrsort use inputs with
    unsigned integer keys (\uint, \ulong, and \pair data types). The
    results were obtained on all machines for all input distributions
    with at least $2^{21}t$ bytes except \distsorted,
    \distreversesorted, and \distones.}\label{fig:perf par}
\end{figure}

\subsubsection{Comparison of Performance Profiles}\label{sec:par perf profiles}

In this section, we compare the pairwise performance profiles of \compiparassssort with the (non)in-place comparison-based \comppbalancedsort (\compppbbs), and the radix sorter \radixregion as well as the pairwise performance profiles of \compiparassrsort and \radixregion. The profiles are shown in \cref{fig:perf par}. We do not compare our algorithms to the radix sorters \radixppbbr and \radixraduls as these are non-in-place and as their profiles are much worse than the profiles of the in-place radix sorter \radixregion.
Overall, the performance of \compiparassssort is much better than the performance of any other sorting algorithm. When we only consider radix sorters, the performance profile of \compiparassrsort is better than the one of \radixregion.

\compiparassssort performs significantly better than \textbf{\compppbbs}.
For example, \compppbbs sorts only  $2.4$\,\% of the inputs at least at fast as \compiparassssort.
Also, there is virtually no input for which \compppbbs is at least a factor of $1.50$ faster than \compiparassssort.
In contrast, \compiparassssort sorts $66$\,\% of the inputs at least a factor of $1.50$ faster than \compppbbs.

The performance profile of \textbf{\comppbalancedsort} is even worse than the one of \compppbbs.
\compiparassssort is faster than \comppbalancedsort for virtually any inputs.
\compiparassssort is even three times faster than \comppbalancedsort for almost $50$\,\% of the inputs.

The performance of \compiparassssort is also
significantly better the performance of \textbf{\radixregion}. 
For example, \compiparassssort sorts $74$\,\% of the inputs faster than \radixregion.
Also, \radixregion sorts only $9$\,\% of the inputs at least a factor of $1.25$ faster than \compiparassssort.
In contrast, \compiparassssort sorts $44$\,\% of the inputs at least a factor of $1.25$ faster than \radixregion.

Among all pairwise performance profiles, the profiles of \textbf{\compiparassssort and \compiparassrsort} are the closest.
Still, \compiparassssort performs better than \compiparassrsort.
For example, \compiparassssort sorts $62$\,\% of the inputs faster than \compiparassrsort.
Also, \compiparassrsort outperforms \compiparassssort for $16$\,\% of the inputs by a factor of $1.25$ or more.
On the other hand, \compiparassssort outperforms \compiparassrsort for $31$\,\% of the inputs by a factor of $1.25$ or more.

\begin{table}
  \resizebox*{0.40\textwidth}{!}{

\begin{tabular}{|@{}l@{}|@{}c@{}|}
\hline
  \begin{tabular}{l}
    \multirow{2}{*}{Distribution}\\
    \\\hline
  
         \distsorted \\
  \distreversesorted \\
           \distones \\

    \hline\hline
    
             \distexpo \\
             \distzipf \\
   \distduplicatesroot \\
  \distduplicatestwice \\
  \distduplicateseight \\
     \distalmostsorted \\
          \distuniform \\
    \hline
    Total\\
    Rank\\\hline
    
  \end{tabular}
                &
  \begin{tabular}{cc}
  \multicolumn{2}{c}{\pcinteltwo}\\
  \compiparassssort & \imsdradix
  \\\hline
  
  \textbf{1.00} & 77.00 \\
  \textbf{1.00} &  7.23 \\
  \textbf{1.00} &       \\

  \hline\hline
  
  \textbf{1.00} &  2.31 \\
  \textbf{1.00} & 35.64 \\
  \textbf{1.00} &  7.10 \\
  \textbf{1.00} & 41.03 \\
  \textbf{1.00} & 44.58 \\
  \textbf{1.00} & 10.47 \\
  \textbf{1.00} &  1.68 \\

  \hline

  \textbf{1.00} & 10.95 \\

  1 & 2 \\\hline
  \end{tabular}
\end{tabular}

  }
  \caption{
    Average slowdowns of \compiparassssort and \imsdradix for \pair data types, different input distributions, and machine \pcinteltwo with inputs containing at least $2^{21}t$ bytes.
    \imsdradix breaks for \distones input.
  }
  \label{tab:slow ips4o vs inplace msd radix}
\end{table}

\subsubsection{Comparison to \imsdradix}\label{sec:comp inpl MSD radix sort}

We compare our algorithm \compiparassssort to the in-place radix sorter
\imsdradix~\cite{polychroniou2014comprehensive}
separately as the available implementation works only in rather
special circumstances -- $64$ threads, $n>2^{26}$, integer key-value
pairs with values stored in a separate array.  Also, that
implementation is not in-place and requires a very specific input
array: On a machine with $m$ NUMA nodes, the input array must consist
of $m$ subarrays with a capacity of $1.2n/m$ each. The experiments in
\cite{GithubOrestis2020comprehensive} pin subarray $i$ to NUMA node $i$.
We nevertheless see the comparison as important since
\imsdradix uses a similar basic approach to block
permutation as \compiparassssort.

\Cref{tab:slow ips4o vs inplace msd radix} shows the average slowdowns
of \compiparassssort and \imsdradix for different input
distributions executed on \pcinteltwo.  We did not run
\imsdradix on \pcintelfour, \pcintellargefour, and \pcamd as
these machines do not have $64$ hardware theads.  The results show
that \compiparassssort is much faster than \imsdradix for
all input distributions.  For example, the average slowdown ratio of
\imsdradix to \compiparassssort is $7.10$ for
\distduplicatesroot input on \pcinteltwo.  Note that
\imsdradix breaks for \distones input and its average
slowdown are between $35.64$ and $44.58$ for some input distributions
with duplicated keys (\distduplicatestwice, \distduplicateseight, and
\distzipf) and with a skewed key distribution (\distzipf).  For
\distsorted input, \imsdradix is also much slower because
\compiparassssort detects sorted inputs.

\input{extern/ips4o-benchmark-suite-plots/benchmark/running_times/subroutine_times.tex}

\subsection{Phases of Our Algorithms}\label{sec:algorithm phases}

\Cref{fig:subroutines} shows the running times of the sequential samplesort algorithm \compissssort and the sequential radix sorter \compissrsort as well as their parallel counterparts \compiparassssort and \compiparassrsort.
The running times are split into the four phases of the partitioning step (sampling,  classification, permutation, and cleanup), the time spent in the base case algorithm, and  overhead for the remaining components of the algorithm such as initialization and scheduling.
In the following discussion of the sequential and parallel execution times,
we report numbers for the largest input size unless stated otherwise.

\subsubsection*{Sequential Algorithms}

The running time curves of  \compissrsort are less smooth than those for \compissssort because this code currently lacks the same careful adaptation of the distribution degree $k$

The time for \textbf{sampling} in the partitioning steps of \compissssort is relatively small, i.e., $7.88$\,\% of the total running time.
For \compissrsort, no sampling is performed.

The \textbf{classification phase} of \compissssort takes up about half of the total running time.
The \textbf{permutation phase} is a factor of about eight faster than its classification phase.
As both phases transfer about the same data volume and as the classification phase performs a factor of $\Th{\log k}$ more local work ($\log k = 8$), we conclude that the classification phase is bounded by its local work. It is interesting to note that this was very different in 2004. In the 2004 paper \cite{sanders2004super}, data distribution dominated element classification.
Since then, peak memory bandwidth of high-performance processors has increased much faster than
internal speed of a single core. The higher investment in memory bandwidth was driven by the need to accommodate the memory traffic of multiple cores. Indeed, we will see below that for parallel execution, memory access once more becomes crucial.
Classification and permutation phases of \compissrsort behave similarly as for \compissssort. Since the local work for classification is much lower for radix sort, the running time ratio between these two phases is smaller yet, with $2.43$ still quite high.

The \textbf{cleanup} takes less than five percent of the total running time of \compissssort and less than two percent of \compissrsort.
The sequential algorithms spend a significant part of the running time in the \textbf{base case}.
The base case takes $36.71$\,\% of the total running time.
For \compissrsort the base case even dominates the overall running time ($70.29$\,\%) because
it performs less work in the partitioning steps and because it uses larger base cases.
The \textbf{overhead} caused by the data structure construction and task scheduling is negligible in the sequential case.

\subsubsection*{Parallel Algorithms}

The partitioning steps of the parallel algorithms are significantly slower than the ones of the sequential algorithms. In particular, the work needed for the permutation phase increases by a factor of $11.59$ for \compiparassssort and $19.88$ for \compiparassrsort. Since the permutation phase does little else than copying rather large blocks, the difference is mainly caused by memory bottlenecks.

Since memory access costs now dominate the running time, 
the performance advantage of radix sort over samplesort decreases when executed in parallel instead of sequentially.
For other input distributions as well as other data types, the parallel radix sort is even slower than parallel sample sort (see \cref{sec:par slowdown}). In other words, the price paid for \emph{less local work} in classification (for radix sort) is more data transfers due to \emph{less accurate classification}. In the parallel setting, this tradeoff is harmful except for uniformly distributed keys and small element sizes.

When the input size below $n= 2^{26}$, the time for the classification phase and the additional overhead dominates the total running time.
We can avoid this overhead.
In our implementation, a ``sorter object'' and an appropriate constructor allows us to separate the data structure creation from the sorting.
When we exclude the time to create and initialize the sorter object, sorting with \compiparassssort becomes significantly faster, i.e., a factor of $3.74$ for $n=2^{24}$.

\section{Conclusion and Future Work}\label{s:conclusion}

In-place Super Scalar Samplesort (\compiparassssort) and In-place
Super Scalar Radix Sort (\compiparassrsort) are among the fastest
sorting algorithms both sequentially and on multi-core machines. The
algorithms can also be used for data distribution and local sorting in
distributed memory parallel algorithms (e.g.,~\cite{ABSS15}).

Both algorithms are the fastest known
algorithms for a wide range of machines, data types, input sizes, and
data distributions. Exceptions are small inputs (which fit into
cache), a limitation to a fraction of the available cores (which profit from nonportable SIMD
instructions), and almost sorted inputs (which profit from sequential adaptive
sorting algorithms).  Even in those exceptions, our algorithms, which
were not designed for these purposes, are surprisingly close to more
specialized implementations. One reason is that for large inputs,
memory access costs overwhelmingly dominate the total cost of a
parallel sorting algorithm so that saving elsewhere has little effect.

Our comparison based
algorithm parallel algorithm \compiparassssort\ even mostly outperforms the integer
sorting algorithms, despite having a logarithmic
factor overhead with respect to executed instructions.
Memory access efficiency of our algorithms is also the reason for the
initially surprising observation that our in-place algorithms outperform
algorithms that are allowed to use additional space.

Both algorithms significantly outperform the fastest parallel
comparison-based competitor, \compppbbs, on almost all inputs.
They are also significantly better than the fastest sequential
comparison-based competitor, \compspdq, except for sorted and
almost-sorted inputs.

The fastest radix sort competitors are \radixsska (sequential) and
\radixregion (parallel).  Our radix sorter is
significantly faster than \radixsska and competitive to \radixregion.
Also, our parallel samplesort
algorithm is significantly faster than \radixregion
for all inputs.  Exceptions are some $32$-bit inputs.
Our parallel samplesort algorithm even sorts uniform distributed inputs significantly faster than
\radixregion if the keys contain more than $32$-bits.

Radix sorters which take advantage of non-portable hardware features,
e.g., \radixipp (vector instructions) and \radixraduls (non-temporal
writes), are very fast for small (\distuniform distributed) data types. \radixipp for example
sorts 32-bit unsigned integers very fast and \radixraduls is very fast
for 64-bit unsigned integers.
However, the interesting methods developed for these algorithms have little impact on larger data types and ``hard'' input distributions and thus, we perform better overall.

We compare the algorithms for input arrays with various NUMA memory
layouts.  With our new locality aware task scheduler,
\compiparassssort is robustly fast for all NUMA memory layouts.

\subsection*{Future Work}

Several improvements of our algorithms can be
considered which address the remaining cases where our algorithms are
outperformed.  For small inputs, not in-place variants of our
algorithms with preallocated data structures, smaller values of the
distribution factor $k$ and smaller block sizes could be faster. For small inputs, the
base case sorter becomes also more relevant. Here we could profit from
several results on fast sorting for very small inputs
\cite{bingmann2020engineering,bramas2017novel,CodishCNS17network}.
Also, we would like to speed up the branchless
decision tree with vector instructions.  Preliminary results have
shown improvements of up to a factor of $1.25$ for \compissssort with
a decision tree using AVX-512 instructions.  However, a general
challenge remains how data-parallel instructions can be harnessed for
sorting data with large keys and associated information and how to
balance portability and efficiency.

With respect to the volume of accessed memory, which isa main
distinguishing feature of our algorithms, further improvements are
conceivable. One option is to reconsider the approach from most radix
sort implementations and of the original super scalar samplesort
\cite{sanders2004super} to first determine exact bucket sizes.  This
is particularly attractive for radix sorters since computing bucket
indices is very fast. Then one could integrate the classification phase
and the permutation phase of \compiparassssort. To make this efficient, one should still work with
blocks of elements moved between local buffer blocks and the
input/output array.
For samplesort, one would approximate bucket sizes using the sample and a cleanup would be required.
Another difficulty may be a robust parallel implementation
that avoids contention for all input distributions.

A more radical approach to reducing memory access volume would be to
implement the permutation phase in sublinear time by using the
hardware support for virtual memory.  For large inputs, one could make
data blocks correspond to virtual memory pages. One could then move
around blocks by just changing their virtual addresses. It is unclear
to us though whether this is efficiently (or even portably) supported
by current operating systems. Also, the output might have an
unexpected mapping to NUMA nodes which might affect the performance of
subsequently processing the sorted array.

Our radix sorter \compiparassrsort\ is currently a prototype meant for
demonstrating the usefulness of our scheduling and data movement
strategies independently of a comparison based sorter. It could be
made more robust by adapting the function for extracting bucket
indices to various input distributions (which can be approximated
analyzing a sample of the input).  This could in particular entail
various compromises between the full-fledged search tree of \compiparassssort and the plain byte extraction of \compiparassrsort.
For example, one could accelerate the search tree traversal of super
scalar samplesort by precomputing a lookup table of starting nodes that are addressed by the most significant bits of the key.
One could also consider the approach from the \radixlearned algorithm
\cite{kristo2020caseimpl} which addresses a large number of buckets
using few linear functions.  Perhaps, approximate
distribution-learning approaches can be replaced by fast and accurate
computational-geometry algorithms. Existing geometry algorithms
\cite{imai1986optimal,douglas1973algorithms} might have to be adapted
to use a cost function that optimizes the information gain from using
a small number of piece-wise linear functions.

Adaptive sorting algorithms are an intriguing area of research in
algorithms~\cite{CastroW92adaptive}.  However, implementations such as
\compstim\ currently cannot compete with the best nonadaptive
algorithms except for some extreme cases. Hence, it would be
interesting to engineer adaptive sorting algorithms to take the
performance improvements of fast nonadaptive algorithms (such as ours)
into account.

The measurements reported in this paper were performed using somewhat
non-portable implementations that use a 128-bit compare-and-swap
instruction specific to x86 architectures (see also
Section~\ref{s:details}). Our portable variants currently use locks
that incur noticeable overheads for inputs with only very few different keys.
Different approaches can avoid locks without noticeable overhead but these would lead to more complicated source code.

Coming back to the original motivation for an alternative to quicksort
variants in standard libraries, we see \compiparassssort\ as an
interesting candidate. The main remaining issue is code
complexity.  When code size matters (e.g., as indicated by a compiler
flag like {\tt -Os}), one could use \compiparassssort with fixed $k$ and
a larger base case size.
Formal verification of the correctness of the implementation might
help to increase trust in the remaining cases.\\

\subparagraph*{\bf Acknowledgements.}  We would like to thank the authors
of~\cite{shun2012brief,synergy2019aspas,obeya2019github,refresh2017raduls2,edelkamp2016github,peters2015pdq,mcfadden2014wikisort,goro2011timsort,Lorenz2016ssss,skarupke2016github,kristo2020caseimpl,GithubOrestis2020comprehensive} for sharing their code for evaluation.  Timo
Bingmann and Lorenz Hübschle-Schneider~\cite{Lorenz2016ssss} kindly
provided an initial implementation of the branchless decision tree
that was used as a starting point in our implementation.

\bibliographystyle{ACM-Reference-Format}
\bibliography{paper}

\ifarxiv

\newpage
\appendix

\section{Details of the Analysis}\label{app:analysis}

\subsection{Limit Number of Recursions}\label{app:analysis num recursions}

In this section, we prove the following theorem:

\begin{theorem}\label{thm:ips4olevel}
  Let $M\geq 1$ be a constant.
  Then, after $\Oh{\log_k \frac{n}{M}}$ recursion levels, all non-equality buckets of \compiparassssort have size $M$ with a probability of at least $1 - n/M$ for an oversampling ratio of $\alpha=\Th{c\log k}$.
\end{theorem}

We first show the \cref{lem:oprobsucc,lem:probmultilevel} which are used to prove \cref{thm:ips4olevel}.
Let $e$ be an arbitrary but fixed element of a task with $n$ elements in \compiparassssort.
A ``successful recursion step'' of $e$ is a recursion step that assigns the element to a bucket of size $3n/k$.

\begin{lemma}\label{lem:oprobsucc}
  The probability of a successful recursion step of an arbitrary but fixed element is at least $1-2k^{-c/12}$ for an oversampling ratio of $\alpha=c\log k$. 
\end{lemma}

\begin{proof}
  We bound the probability that a task of $n$ elements assigns an arbitrary but fixed element $e_j$ to a bucket containing at most $3n/k$ (a successful recursion step).
  Let $\oset{e_1}{e_{n}}$ be the input of the task in sorted order, let $R_r = \oset{e_r}{e_{r+1.5n/k - 1}}$ be the set containing $e_k$ and the $1.5n/k\}$'th larger elements, and let $\oset{s_{1} }{ s_{\alpha k}}$ be the selected samples.
  The boolean indicator $X_{ik}$ that sample $s_i$ is an element of $R_k$ is defined as
  \[
  X_{ij}=\left\{
    \begin{array}{lr}
      1, & s_i\in R_k\\
      0, & \text{else.}
    \end{array}\right.
    \]
    The probability $\Pr[X_{ik} = 1] = 1.5\frac{n}{k}\cdot\frac{1}{n}=\frac{1.5}{k}$ is independent of the sample $s_i$ as the samples are selected with replacement.
    Thus, the expected value of the number of samples selected from $R_k$ is $X_{k} = \sum_{i=1}^{\alpha k}X_{ik}$ is $E\left[X_{k}\right]=1.5/k \cdot \alpha k = 1.5\alpha$.
    We use the Chernoff bound to limit the probability of less than $\alpha$ samples in $R_k$ to $\Pr[X_{k} < \alpha] = \Pr[X_k < (1 - 1/3)E[X_k]] < e^{-1/2\left(1/3\right)^2E[X_k]} = e^{-1/12\alpha}$.
    When $R_j$ as well as $R_{j - 1.5n/k}$ both provide at least $\alpha$ samples, $R_j$ as well as $R_{j - 1.5n/k}$ provide a splitter and $e_j$ is in a bucket containing at most $3n/k$ elements.
    The probability is $\Pr[X_j \geq S \land X_{j - 1.5n/k} \geq S] = 1-\Pr[X_j < S \lor X_{j - 1.5n/k} < S] \geq  1-\Pr[X_j < S] -\Pr[X_{j - 1.5n/k} < S] > 1 - 2e^{-1/12\alpha} = 1 - 2k^{-1/12c}$.
\end{proof}

\begin{lemma}\label{lem:probmultilevel}
  Let $c$ be a constant, let $\alpha=c\log k$ be the oversampling ratio of \compiparassssort ($c \geq 36 - 2.38/\log(0.34 \cdot k)$), and let \compiparassssort execute $2\log_{k/3}\frac{n}{M}$ recursion levels. 
  Then, an arbitrary but fixed input element of \compiparassssort passes at least $\log_{k/3}\log \frac{n}{M}$ successful recursion levels with a probability of at least $1-\left(n/M\right)^{-2}$.
\end{lemma}

\begin{proof}
  We execute \compiparassssort $2\log_{k/3}\frac{n}{M}$ recursion levels and bound the probability that an arbitrary but fixed input element passes at least $\log_{k/3}\log\frac{n}{M}$ successful recursion levels.
    This experiment is a Bernoulli trial as we have exactly two possible outcomes, ``successful recursion step'' and ``non-successful recursion step'', and the probability of success is the same on each level.
    Let denote the random variable $X$ as the number of non-successful recursion steps after $2\log_{k/3}\frac{n}{M}$ recursion levels, $p$ the probability of a non-successful recursion step, and let $c \geq 36 - 2.38/\log(0.34 \cdot k)$.
    The probability $I$
  \begin{equation}
    \begin{aligned}
      I &= \probability{X > 2\log\frac{n}{M} - \log\frac{n}{M}}
       \leq \probability{X > \log\frac{n}{M}}\\
      &\leq \sum_{j>\log\frac{n}{M}} \binom{2\log\frac{n}{M}}{j}p^j\left(1-p\right)^{2\log\frac{n}{M} - j}
      \leq \sum_{j>\log\frac{n}{M}} \left(\frac{2e\log\frac{n}{M}}{j}\right)^jp^j\\
      &\leq \sum_{j>\log\frac{n}{M}} \left(\frac{2e\log\frac{n}{M}}{\log\frac{n}{M}}\right)^jp^j
      \leq \sum_{j>\log\frac{n}{M}} \left(2e\right)^j\left(2k^{-1/12c}\right)^j\\
      &\leq \sum_{j>\log\frac{n}{M}} \left(4ek^{-1/12c}\right)^j
      = \frac{\left(4ek^{-1/12c}\right)^{\log\frac{n}{M} + 1}}{1 - 4ek^{-1/12c}}\\
      &\leq \frac{\left(\frac{n}{M}\right)^{-1/12c + \log(4e)}}{1 - 4ek^{-1/12c}}
      \leq \left(\frac{n}{M}\right)^{-2}
    \end{aligned}
  \end{equation}
  defines an upper bound of the probability that a randomly selected input element passes $2\log_{k/3}n/M$ recursion levels without passing $\log_{k/3}\frac{n}{M}$ successful recursion levels.
  For the sake of simplicity, all logarithms of the equation above are to the base of $k/3$.
  The third ``$\leq$'' uses $\binom{n}{k} \leq \left(en/k\right)^k$, the fifth ``$\leq$'' uses \cref{lem:oprobsucc} and the ``$=$'' uses the geometric series.
\end{proof}

\begin{proof}[Proof of \cref{thm:ips4olevel}]
  We first assume that $M\geq k^2n_0$ holds.
  In this case, we select $kc\log k$ samples.
  Let $l = \log_{k/3}\frac{n}{M}$ and let $e$ be an arbitrary but fixed input element of \compiparassssort after $2l$ recursion levels.
  \cref{lem:probmultilevel} tells us that $e$ has passed at least $l$ successful recursion steps with a probability of at least $1-\left(n/M\right)^{-2}$ when \compiparassssort has performed $2\log_{k/3}\frac{n}{M}$ recursion levels.
  Element $e$ is, in this case, in a bucket containing more than $n\left(3/k\right)^l = M$ elements as each successful recursion step shrinks the bucket by a factor of at least $3/k$.
  Let $E = \oset{e_1}{e_{n}}$ be the input elements of \compiparassssort in sorted order and let $Q = \{e_{iM} | 1 \leq i < n/M \wedge i \in \mathbb{N}\}$ every $n/M$'th element.
  We now examine buckets containing elements in $Q$ after $2l$ recursion levels.
  The probability that any element $Q$ is in a bucket containing more than $M$ elements is less than $n/M\cdot\left(n/M\right)^{-2} = \left(n/M\right)^{-1}$ -- this follows from the former insight and the Boole's inequality.
  In other words, the probability that all elements in $Q$ are in buckets containing less than $M$ elements is larger than $1 - M / n$.
  As this holds for all elements in $Q$, every $n/M$'th element in the sorted output, the probability that all elements after $2l$ recursion level are in buckets containing less than $M$ elements is larger than $1 - M / n$.
\end{proof}

\subsection{Comparing the I/O  Volume of \compissssort\ and \compssssort.}\label{app:io volume analysis}

\begin{table}[!ht]
  \centering
  \begin{tabular}{l|l|l|l}
    Subroutine               & Types                     & Reps      & Sum in $n$ Bytes      \\
    \hline
    \textbf{\compssssort}                                                                    \\
    \hline\hline
    Copy back                & r + w + {\color{green}wa} & once      & 16 + {\color{green}8} \\
    Base Case                & r + w                     & once      & 16                    \\
    Init Temp Array + Oracle & w                         & once      & 9                     \\
    Classification: Oracle   & w + {\color{green}wa}     & per level & 1 + {\color{green}1}  \\
    Classification: Array    & r                         & per level & 8                     \\
    Redistribution: Oracle   & r                         & per level & 1                     \\
    Redistribution: Array    & r + w + {\color{green}wa} & per level & 16+{\color{green}8}   \\
    \hline
    \textbf{\compissssort}                                                                   \\
    \hline\hline
    Base Case                & r + w                     & once      & 16                    \\
    Classification           & r + w                     & per level & 16                    \\
    Redistribution           & r + w                     & per level & 16                    \\
  \end{tabular}
  \caption{I/O volume of read (\emph{r}) and write (\emph{w}) operations broken down into subroutines of \compissssort and \compssssort. Additionally, potential write allocate operations (\emph{wa}) are listed.}
  \label{tab:io comparison}
\end{table}

We compare the first level of \compissssort\ and \compssssort for inputs with $8$-byte input elements. We assume a oracle with $1$-byte entries for \compssssort.
Furthermore, we assume that the input does not fit into the private cache.

Both algorithms read and write the data once for the base case -- $16n$ bytes of I/O volume.
Each level of \compissssort\ reads and writes all data once in the classification phase and once in the permutation phase -- $32n$ bytes per level.
Each level of \compssssort\ reads the elements twice and writes them once only in its distribution phase -- $24n$ bytes per level.

Additionally, \compssssort\ writes an oracle sequence that indicates the bucket for each element in the classification phase and reads the oracle sequence in the distribution phase -- $2n$ bytes per level.
The algorithm also has to allocate the temporary arrays. For security reasons, that memory is zeroed by the operating system -- $9n$ bytes.%
\footnote{In current versions of the Linux kernel this is done by a single thread and thus results in a huge scalability bottleneck.}
If the number of levels is odd, \compssssort\ has to copy the sorted result back to the input array -- $16n$ bytes.
For now, \compissssort\ (\compssssort) has an I/O volume of $32n$ ($26n$) byte per level and $16n$ ($41n$) bytes once.

When \compssssort writes to the temporary arrays or during copying back, cache misses happen when an element is written to a cache block that is currently not in memory.
Depending on the cache replacement algorithm, a \emph{write allocate} may be performed -- the block is read from the memory to the cache even though none of the data in that block will ever be read.
Detecting that the entire cache line will be overwritten is difficult as \compssssort writes to the target buckets element by element.
This amounts to an I/O volume of up to $9n$ bytes per level and $8n$ bytes once.
\compissssort\ does not perform write allocates.
The classification phase essentially sweeps a window of size $\Th{bk}$ through the memory by reading elements from the right border of the window and writing elements to the left border.
The permutation phase reads a block from the memory and replaces the ``empty'' memory bock with a cached block afterwards.
Finally, we get for \compissssort\ (\compssssort) a total I/O volume of $32n$ ($35n$) byte per level and $16n$ ($49n$) bytes once -- \compssssort\ with one level has a factor of $1.75$ more I/O volume than \compissssort. \Cref{tab:io comparison} shows the I/O volume of the subroutines in detail.

Furthermore, \compssssort\ may suffer more conflict misses\note{also called collision misses} than \compissssort\ due to the mapping of data to cache lines.
In the distribution phase,  \compssssort\ reads the input from left to right but writes elementwise to positions in the buckets which are not coordinated.
For the same reasons, \compssssort\ may suffer more TLB misses.
\compissssort, on the other hand, essentially writes elements to cached buffer blocks (classification) and swaps blocks of size $b$ within the input array (block permutation).
For an average case analysis on scanning multiple sequences, we refer to \cite{MehSan03}.

Much of this overhead can be reduced using measures that are non-portable (or hard to make portable).
In particular, non-temporal writes eliminate the write allocates and also help to eliminate the conflict misses.
One could also use a base case sorter that does the copying back as a side-effect when the number of recursion levels is odd.
When sorting multiple times within an application, one can keep the temporary arrays without having to reallocate them. However, this may require a different interface to the sorter.
Overall, depending on many implementation details, \compssssort\ may require slightly or significantly more I/O volume.

%......................................................................
\section{From In-Place to Strictly In-Place}\label{ss:strictly}
We now explain how the space consumption of \compiparassssort can be made independent of $n$ in a rather simple way by adapting the strictly in-place approach of quicksort.
We do not consider the space requirement for storing parallel tasks as those tasks are processed immediately.
However, we require the (sequential and parallel) partitioning steps to mark the beginning of each subtask by storing the largest element of a subtask in its first position.
When a thread has finished its last parallel task, the elements covered by its sequential tasks can be described with two indices $c_l$ and $c_r$:

\begin{lemma}\label{lem:invariant local stack}
  When thread $i$ starts its sequential phase, the sequential tasks assigned to thread $i$ cover a consecutive subarray of the input array, i.e., there is no gap between the tasks in the input array.
\end{lemma}
For reasons of better readability, we appended the proof of \cref{lem:invariant local stack} to the end of this chapter.
The proof of this lemma also implicitly describes the technique to calculate the values of $c_l$ and $c_r$.

In the sequential phase, the thread has to sort the elements $A[c_l,c_r-1]$, which are partitioned into its sequential tasks.
The boundaries of the sequential tasks are implicitly represented by the largest element at the beginning of each task.
Starting a search at the leftmost task, the first element larger than the first element of a task defines the first element of the next task.
Note that the time required for the search is only logarithmic to the task size when using an exponential/binary search.
We assume that the corresponding function $\mathit{searchNextLargest}$ returns $n+1$ if no larger elements exist -- this happens for the last task.
The function $\mathit{onlyEqualElements}$ checks whether a task only contains identical elements.
We have to skip these ``equal'' tasks to avoid an infinite loop.
The following pseudocode uses this approach to emulate recursion in constant space on the sequential tasks.
\begin{code}
  $n := e - b$\RRem{total size of sequential tasks}\\
  $i := b$\RRem{first element of current task}\\
  $j := \mathit{searchNextLargest}(A[i],A,i+1,n)$\RRem{first element of next task}\\
  \textbf{while} $i<n$ \textbf{do}\+\\
  \textbf{else if} $\mathit{onlyEqualElements}(A,i,j-1)$; \textbf{then} $i := j$\RRem{skip equal tasks}\\
  \textbf{else if} $j-i<n_0$ \textbf{then} $\mathit{smallSort}(a,i,j-1)$;\quad $i := j$\RRem{base case}\\
  \textbf{else} $\mathit{partition}(a,i,j-1)$\RRem{partition first unsorted task}\\
  $j := \mathit{searchNextLargest}(A[i],A,i+1,n)$\RRem{find beginning of next task}
\end{code}

The technique which we described -- making the space consumption of \compiparassssort independent of $n$ -- used the requirement that the sequential tasks of a thread cover the consecutive subarray $\VarArray[c_l, c_r-1]$ for some $c_l$ and $c_r$.
In the following, we show that this requirement holds.

\begin{proof}[Proof of \cref{lem:invariant local stack}]\label{proof:invariant local stack}
  Let thread $i$ process a parallel task $T[l, r)$ with thread group $\excloset{\tbegin }{ \tend}$ in one of the following states:
  \begin{itemize}
  \item \emph{Left}. Thread $i$ is the leftmost thread of the thread group, i.e., $i = \tbegin$. 
  \item \emph{Right}. Thread $i$ is the rightmost thread of the thread group, i.e., $i = \tend-1$. 
  \item \emph{Middle}. Thread $i$ is not the leftmost or rightmost of the thread group, i.e., $\tbegin < i < \tend - 1$. 
  \end{itemize}
  
  We claim that the sequential tasks assigned to thread $i$ fulfill the following propositions when (and directly before) thread $i$ processes $T[l, r)$:
  \begin{itemize}
  \item \emph{$T[l, r)$ was processed in state Left (Right)}. Thread $i$ does not have sequential tasks or its sequential tasks cover a consecutive subarray of the input array, i.e., there is no gap between the tasks.
    In the latter case, the rightmost task ends at position $l - 1$ with $l-1\in (in/t,(i+1)n/t-1)$ (leftmost task begins at position $r$ with $r\in (in/t,(i+2)n/t-1)$).
  \item \emph{$T[l, r)$ was processed in state Middle}. Thread $i$ does not have sequential tasks.
  \end{itemize}
  
  Assume for now that these propositions hold -- we will prove them later.
  We use the proposition to show that \cref{lem:invariant local stack} holds when a thread $i$ starts processing its sequential tasks:
  Let $T[l, r)$ be the last parallel task of thread $i$, executed with the thread group $\excloset{\tbegin }{ \tend}$.
  No matter in which state $T[l, r)$ has been executed, the sequential tasks of thread $i$ cover a consecutive subarray of the input array at the beginning of its sequential phase.
  \begin{itemize}
  \item \emph{$T[l, r)$ was processed in state Right}.
    As $i = \tend-1$, we add all subtasks of $T[l, r)$ which start in $\VarArray[in/t-1,r-1]$ to thread $i$.
    No gap can exist between these subtasks as they cannot be interrupted by a parallel task.
    Also, before we assign the new subtasks to thread $i$, the leftmost sequential task of thread $i$ begins at position $r$ (see proposition).
    Then, the rightmost sequential subtask of $T[l, r)$ which start in $\VarArray[in/t-1,r-1]$ ends at $\VarArray[r-1]$.
    Thus, after the subtasks were added, there is no gap between the sequential tasks of thread $i$.
  \item \emph{$T[l, r)$ was processed in state Left}.
    As $i = \tbegin$, we add all subtasks of $T[l, r)$ which start in $\VarArray[l,(i+1)n/t-1]$ to thread $i$.
    No gap can exist between these subtasks as they cannot be interrupted by a parallel task.
    Also, before thread $i$ adds the new subtasks, the rightmost sequential task ends at position $l-1$ (see proposition).
    Then, the leftmost subtask of $T[l, r)$ which starts in $\VarArray[l,(i+1)n/t-1]$ begins at $\VarArray[l]$.
    Thus, after the subtasks were added, there is no gap between the sequential tasks of thread $i$.
  \item \emph{$T[l, r)$ processed in state middle}.
    We add all sequential subtasks to thread $i$ which start in $\VarArray[i n/t, (i + 1)n/t-1]$.
    No gap can exist between these subtasks as they cannot be interrupted by a parallel task.
    Also, the subtasks are added in sorted order from left to right.
  \end{itemize}
  
  We will now prove the propositions by induction.
  
  \emph{\underline{Base case.}}
  When a thread $i$ processes its first parallel task $T[0,n)$, thread $i$ does not have any sequential tasks.

  \emph{\underline{Inductive step.}}
  We assume that the induction hypothesis holds when thread $i$ was executing the parallel task $T[l, r)$ with the thread group $\excloset{\tbegin }{ \tend}$.
  We also assume thread $i$ and others execute the next parallel task $T[l_s, r_s)$.
  We note that $T[l_s, r_s)$ is a subtask of $T[l, r)$.
  We have to prove that the induction hypothesis still holds after we have added thread $i$'s sequential subtasks of $T[l, r)$ to the thread, i.e., when thread $i$ executes the subtask $T[l_s, r_s)$.
  
  \begin{itemize}
  \item \emph{Thread $i$ has executed $T[l, r)$ in the state Middle and thread $i$ is executing $T[l_s, r_s)$ in the state Middle.}
    From the induction hypothesis, we know that thread $i$ did not have sequential tasks when $T[l, r)$ was executed.
    We have to show that thread $i$ did not get sequential subtask of $T[l, r)$, after $T[l, r)$ has been executed.
    A sequential subtask $T[a,b)$ would have been added to thread $i$, if $i=\min(at/n, t-1)$.
    We show that no $T[a,b)$ with this property exists.
    As thread $i$ is not the rightmost thread of $T[l,r)$, we have $i < \tend-1$.
    This means that a sequential subtask $T[a,b)$ is only assigned to thread $i$ if $i=\lfloor at/n\rfloor$ holds, i.e., $a\in [in/t,(i+1)n/t)$ is required.
    However, there is no sequential subtask of $T[l, r)$ which begins in the $i$'th stripe of the input array:
    As thread $i$ is not the leftmost thread of $T[l_s, r_s)$, the parallel subtask $T[l_s, r_s)$ contains the subarray $\VarArray[in/t, (i+1)n/t-1]$ completely (see \cref{lem:par task range}).
    Thus, a second (sequential) subtask $T[a,b)$ with $in/t\leq a<(i+1)n/t$ cannot exist.
  \item \emph{Thread $i$ has executed $T[l, r)$ in the state Middle and thread $i$ is executing $T[l_s, r_s)$ in the state Right.}
    From the induction hypothesis, we know that thread $i$ did not have sequential tasks when $T[l, r)$ was executed.
    As thread $i$ was not the rightmost thread of $T[l,r)$, we have $i < \tend-1$.
    This means that a sequential subtask $T[a,b)$ of $T[l,r)$ is only assigned to thread $i$ if $i=\lfloor at/n\rfloor$ holds, i.e., $a\in [in/t,(i+1)n/t)$ is required.
    However, as thread $i$ is not the leftmost thread of $T[l_s, r_s)$,  $T[l_s, r_s)$ completely contains $\VarArray[in/t,(i+1)n/t-1]$.
    Thus, there is no sequential subtask $T[a,b)$ with $a\in [in/t(i+1)n/t)$ -- we do not add sequential tasks of $T[l, r)$ to thread $i$.
  \item \emph{Thread $i$ has executed $T[l, r)$ in the state Middle and thread $i$ is executing $T[l_s, r_s)$ in the state Left.}
    From the induction hypothesis, we know that thread $i$ did not have sequential tasks when $T[l, r)$ was executed.
    Also, as thread $i$ is not the rightmost thread of $T[l,r)$, we have $i < \tend-1$.
    This means that a sequential subtask $T[a,b)$ is only assigned to thread $i$ if $i=\lfloor at/n\rfloor$ holds, i.e., $a\in [in/t,(i+1)n/t)$ is required.
    Thus, if there is no sequential subtask $T[a,b)$ of $T[l, r)$ with $a\in [in/t,(i+1)n/t)$, the thread $i$ does not get sequential subtasks and the induction step is completed in the case here.
    Otherwise, if sequential subtasks $T[a,b)$ exist with $a\in in/t,(i+1)n/t)$, they are added to thread $i$ and we have to show that the propositions hold afterwards:
    All subtasks $T[a,b)$ which begin in $\VarArray[i n/t,(i+1)n/t]$ are sequential subtasks, except one parallel subtask, $T[l_s, r_s)$.
    Thus, there is no gap between these sequential subtasks.
    As thread $i$ is the leftmost thread of $T[l_s, r_s)$, we know that $l_s\in (in/t,(i+1)n/t)$ and that $r_s \geq (i+2)n/t$.
    Thus, the rightmost sequential subtask ends at $\VarArray[l_s - 1]$ with $l_s-1\in (in/t,(i+1)n/t-1)$.
  \item \emph{Thread $i$ has executed $T[l, r)$ in the state Left (Right) and thread $i$ executes $T[l_s, r_s)$ in the state Left (Right).}
    From the induction hypothesis, we know that $l-1$ (that $r$) is narrowed by $l-1\in (in/t,(i+1)n/t)$ (by $r\in [in/t,(i+2)n/t)$) before tasks of $T[l, r)$ are added to thread $i$.
    As thread $i$ is the leftmost (rightmost) thread of $T[l_s, r_s)$, we can narrow the begin $l_s$ (end $r_s$) of $T[l_s, r_s)$ also by $l_s-1\in (in/t,(i+1)n/t)$ (by $r_s\in [in/t,(i+2)n/t)$).
    Thus,  $T[l, r)$ creates subtasks whereof one subtask starts at $\VarArray[l]$ (at $\VarArray[r_s]$), one subtask ends at $\VarArray[l_s-1]$ (at $\VarArray[r-1]$), and subtasks cover the remaining elements in between without gaps.
    These subtasks are sequential subtasks as $\lfloor (l_s-1)t/n\rfloor -\lfloor lt/n \rfloor \leq \lfloor ((i+1)n/t-1)t/n\rfloor -\lfloor (in/t) t/n \rfloor = 0$
    (as $\lfloor (r-1)t/n\rfloor -\lfloor r_st/n \rfloor \leq \lfloor ((i+2)n/t-1)t/n\rfloor -\lfloor (in/t) t/n \rfloor = 1$).
    And, these sequential subtasks are all added to thread $i$, as they start in the subarray $\VarArray[in/t,(i+1)n/t-1]$ (in the subarray $\VarArray[in/t,(i+2)n/t-1]$, the $+2$ is used as thread $i$ is the rightmost thread of $T[l_s, r_s)$).
    Note that, in the penultimate sentence, we used the inequality $l\leq in/t$ (the inequality $r \leq (i+2)n/t$) from the induction hypothesis.
    When these subtasks were added to thread $i$, the sequential tasks of thread $i$ still cover a consecutive sequence of elements:
    On the one hand, the leftmost (rightmost) sequential subtask starts at $\VarArray[l]$ (ends at $\VarArray[r-1]$,) and the new sequential subtasks have no gaps in between.
    On the other hand, we know from the induction hypothesis that the rightmost (leftmost) sequential task of thread $i$ had ended at position $l-1$ (had started at position $r$) and that the old sequential tasks of thread $i$ had not had gaps in between.
  \end{itemize}
\end{proof}

\section{More Measurements}\label{app:more measurements}

% ------------------------------------------------------------------------------
% - Running times of sequential algorithms for \uint inputs
% ------------------------------------------------------------------------------

\input{extern/ips4o-benchmark-suite-plots/benchmark/running_times/running_time_random_sequential_some_machines_uint.tex}

% ------------------------------------------------------------------------------
% - Slowdowns of sequential algorithms for each machine.
% ------------------------------------------------------------------------------

\begin{table}[!ht]
  \resizebox*{!}{0.93\textheight}{

\begin{tabular}{ll|rrrrrrrrr|rrrr}
  Type
  & Distribution
  & \rotatebox[origin=c]{90}{\compissssort}
  & \rotatebox[origin=c]{90}{\compspdq}
  & \rotatebox[origin=c]{90}{\compblock}
  & \rotatebox[origin=c]{90}{\compmyssssaxtmann}
  & \rotatebox[origin=c]{90}{\compsyaros}
  & \rotatebox[origin=c]{90}{\compssort}
  & \rotatebox[origin=c]{90}{\compstim}
  & \rotatebox[origin=c]{90}{\compsmergequick}
  & \rotatebox[origin=c]{90}{\compswiki}
  & \rotatebox[origin=c]{90}{\radixsska}
  & \rotatebox[origin=c]{90}{\radixipp}
  & \rotatebox[origin=c]{90}{\radixlearned}
  & \rotatebox[origin=c]{90}{\compiparassrsort}\\\hline
  \double &        \distsorted &          1.11 & 1.77 & 20.31 & 1.04 & 11.54 & 17.14 & \textbf{1.03} & 51.78 & 2.90 & 16.90 & 49.53 & 46.50 &  \\
  \double & \distreversesorted & \textbf{1.04} & 1.87 & 15.20 & 1.08 &  5.57 &  6.40 &          1.07 & 25.97 & 6.08 &  9.18 & 24.82 & 22.38 &  \\
  \double &          \distones &          1.14 & 1.93 & 17.48 & 1.11 &  1.29 & 13.56 & \textbf{1.02} &  2.93 & 3.72 & 13.09 & 18.38 & 13.69 &  \\

  \hline\hline
  
  \double &            \distexpo & \textbf{1.05} &          1.10 & 1.24 & 1.28 & 2.31 & 2.56 & 4.15 & 3.79 & 3.98 &          1.27 &          1.12 & 2.33 &  \\
  \double &            \distzipf &          1.19 &          1.32 & 1.44 & 1.45 & 2.86 & 3.07 & 4.77 & 4.23 & 4.84 &          1.26 & \textbf{1.14} & 3.14 &  \\
  \double &  \distduplicatesroot & \textbf{1.10} &          1.39 & 1.73 & 1.62 & 1.51 & 2.50 & 1.50 & 5.43 & 2.81 &          1.70 &          2.43 & 3.33 &  \\
  \double & \distduplicatestwice &          1.21 &          1.33 & 1.40 & 1.42 & 2.50 & 2.73 & 2.93 & 3.27 & 3.12 & \textbf{1.11} &          1.14 & 3.15 &  \\
  \double & \distduplicateseight & \textbf{1.02} &          1.08 & 1.36 & 1.27 & 2.46 & 2.78 & 4.28 & 4.29 & 4.07 &          1.21 &          1.35 & 2.38 &  \\
  \double &    \distalmostsorted &          2.22 & \textbf{1.05} & 1.87 & 2.79 & 1.57 & 1.62 & 1.25 & 5.85 & 2.26 &          2.05 &          4.22 & 4.39 &  \\
  \double &         \distuniform &          1.13 &          1.24 & 1.27 & 1.32 & 2.47 & 2.57 & 3.62 & 2.94 & 3.56 &          1.14 & \textbf{1.07} & 2.13 &  \\

  \hline
  Total  & &

  1.23 & \textbf{1.21} & 1.46 & 1.53 & 2.19 & 2.51 & 2.89 & 4.15 & 3.43 & 1.36 & 1.55 & 3.04 &  \\

  Rank & &
  2 & 1 & 4 & 5 & 7 & 8 & 9 & 12 & 11 & 3 & 6 & 10 &  \\\hline\hline
           
  \ulong &        \distsorted &          1.10 & 1.77 & 17.70 & 1.06 & 8.93 & 16.60 & \textbf{1.06} & 42.12 & 2.82 & 17.88 & 61.80 & 79.74 & 11.03 \\
  \ulong & \distreversesorted & \textbf{1.02} & 1.73 & 14.04 & 1.08 & 4.92 &  6.25 &          1.03 & 22.35 & 6.16 & 10.28 & 35.04 & 40.14 &  6.71 \\
  \ulong &          \distones &          1.10 & 1.50 & 15.49 & 1.04 & 1.08 & 12.05 & \textbf{1.04} &  2.35 & 3.34 & 12.65 & 16.71 & 15.46 &  1.29 \\

  \hline\hline
  
  \ulong &            \distexpo &          1.09 &          1.22 & 1.35 & 1.40 & 2.65 & 2.84 & 4.69 & 3.74 & 4.62 & 1.23 & 1.37 & 2.50 & \textbf{1.05} \\
  \ulong &            \distzipf &          1.45 &          1.71 & 1.92 & 1.93 & 3.54 & 3.82 & 6.08 & 5.02 & 6.23 & 1.60 & 1.50 & 3.07 & \textbf{1.04} \\
  \ulong &  \distduplicatesroot & \textbf{1.06} &          1.44 & 1.77 & 1.70 & 1.43 & 2.55 & 1.64 & 5.16 & 3.24 & 1.70 & 2.28 & 3.33 &          1.08 \\
  \ulong & \distduplicatestwice &          1.55 &          1.84 & 1.89 & 2.00 & 3.34 & 3.57 & 4.02 & 4.07 & 4.29 & 1.37 & 2.32 & 3.07 & \textbf{1.00} \\
  \ulong & \distduplicateseight &          1.20 &          1.32 & 1.56 & 1.58 & 2.77 & 3.15 & 5.07 & 4.76 & 5.03 & 1.49 & 2.68 & 2.71 & \textbf{1.02} \\
  \ulong &    \distalmostsorted &          2.13 & \textbf{1.06} & 1.85 & 3.03 & 1.52 & 1.71 & 1.35 & 5.36 & 2.41 & 2.37 & 6.14 & 8.39 &          1.23 \\
  \ulong &         \distuniform &          1.28 &          1.47 & 1.51 & 1.63 & 2.84 & 2.97 & 4.24 & 3.18 & 4.32 & 1.17 & 1.57 & 4.86 & \textbf{1.05} \\

  \hline
  Total  & &

  1.36 & 1.42 & 1.68 & 1.84 & 2.45 & 2.86 & 3.42 & 4.40 & 4.14 & 1.52 & 2.24 & 4.99 & \textbf{1.06} \\

  Rank & &
  2 & 3 & 5 & 6 & 8 & 9 & 10 & 12 & 11 & 4 & 7 & 13 & 1 \\\hline\hline

  \uint &        \distsorted & 2.84 & 4.29 & 49.59 & 2.87 & 24.36 & 60.07 & \textbf{1.94} & 121.52 & 6.43 & 35.09 & 48.04 & 263.84 & 27.84 \\
  \uint & \distreversesorted & 1.55 & 2.27 & 20.24 & 1.46 &  6.38 & 11.16 & \textbf{1.01} &  31.74 & 5.86 &  9.79 & 29.71 &  61.66 &  8.44 \\
  \uint &          \distones & 2.56 & 3.97 & 48.98 & 2.53 &  2.26 & 33.81 & \textbf{1.94} &   6.54 & 9.05 & 20.41 & 12.16 &  37.89 &  3.12 \\

  \hline\hline
  
  \uint &            \distexpo & 1.54 & 1.85 & 2.07 & 1.89 & 4.37 & 4.57 &          7.00 & 5.93 & 6.71 & 1.47 & \textbf{1.03} &  4.73 &          1.18 \\
  \uint &            \distzipf & 1.89 & 2.31 & 2.65 & 2.40 & 5.27 & 5.67 &          8.57 & 7.40 & 8.91 & 1.33 &          1.20 &  5.23 & \textbf{1.18} \\
  \uint &  \distduplicatesroot & 1.19 & 1.55 & 1.97 & 1.85 & 1.63 & 2.76 &          1.44 & 5.98 & 3.15 & 1.23 &          1.52 &  3.86 & \textbf{1.11} \\
  \uint & \distduplicatestwice & 1.93 & 2.46 & 2.50 & 2.46 & 5.05 & 5.07 &          5.07 & 5.20 & 5.47 & 1.22 &          1.46 &  5.22 & \textbf{1.10} \\
  \uint & \distduplicateseight & 1.34 & 1.64 & 1.99 & 1.77 & 4.17 & 4.56 &          6.50 & 5.74 & 6.43 & 1.22 &          1.83 &  3.19 & \textbf{1.01} \\
  \uint &    \distalmostsorted & 2.65 & 1.25 & 2.21 & 3.50 & 1.83 & 2.74 & \textbf{1.14} & 6.79 & 2.45 & 2.08 &          4.92 & 10.36 &          1.33 \\
  \uint &         \distuniform & 1.75 & 2.05 & 2.06 & 2.04 & 4.10 & 4.23 &          5.89 & 4.55 & 5.91 & 1.41 & \textbf{1.00} &  5.72 &          1.32 \\

  \hline
  Total  & &

  1.70 & 1.83 & 2.19 & 2.21 & 3.46 & 4.09 & 4.09 & 5.88 & 5.15 & 1.40 & 1.58 & 6.56 & \textbf{1.17} \\

  Rank & &
  4 & 5 & 6 & 7 & 8 & 10 & 9 & 12 & 11 & 2 & 3 & 13 & 1 \\\hline\hline
  
  \pair &        \distsorted & 1.12 & 1.57 & 13.51 &          1.04 & 7.57 & 12.35 & \textbf{1.02} & 28.08 & 2.31 & 13.08 &  &  & 8.61 \\
  \pair & \distreversesorted & 1.11 & 1.41 &  9.31 & \textbf{1.01} & 3.78 &  4.63 &          1.05 & 14.28 & 7.20 &  6.49 &  &  & 5.00 \\
  \pair &          \distones & 1.16 & 1.65 & 10.91 &          1.05 & 1.08 & 10.21 & \textbf{1.03} &  1.97 & 2.74 &  9.02 &  &  & 1.22 \\

  \hline\hline
  
  \pair &            \distexpo & 1.15 &          2.05 & 1.29 & 1.38 & 2.11 & 2.45 & 4.26 & 3.19 & 4.18 & 1.25 &  &  & \textbf{1.05} \\
  \pair &            \distzipf & 1.45 &          2.75 & 1.67 & 1.82 & 2.69 & 2.84 & 4.82 & 3.73 & 5.27 & 1.48 &  &  & \textbf{1.02} \\
  \pair &  \distduplicatesroot & 1.20 &          1.46 & 1.68 & 1.71 & 1.44 & 2.30 & 1.86 & 4.39 & 3.78 & 1.60 &  &  & \textbf{1.03} \\
  \pair & \distduplicatestwice & 1.74 &          3.04 & 1.83 & 2.01 & 2.90 & 3.10 & 3.74 & 3.56 & 4.41 & 1.47 &  &  & \textbf{1.00} \\
  \pair & \distduplicateseight & 1.30 &          2.39 & 1.53 & 1.65 & 2.28 & 2.65 & 4.51 & 3.97 & 4.93 & 1.46 &  &  & \textbf{1.01} \\
  \pair &    \distalmostsorted & 2.73 & \textbf{1.02} & 2.29 & 3.40 & 1.86 & 2.06 & 2.29 & 5.47 & 3.91 & 2.58 &  &  &          1.48 \\
  \pair &         \distuniform & 1.41 &          2.54 & 1.47 & 1.71 & 2.46 & 2.48 & 3.82 & 2.88 & 4.22 & 1.24 &  &  & \textbf{1.00} \\

  \hline
  Total  & &

  1.50 & 2.06 & 1.66 & 1.88 & 2.20 & 2.53 & 3.43 & 3.81 & 4.36 & 1.54 &  &  & \textbf{1.07} \\

  Rank & &
  2 & 6 & 4 & 5 & 7 & 8 & 9 & 10 & 11 & 3 &  &  & 1 \\\hline\hline
  
  \quartet & \distuniform & 1.06 & 1.91 & 1.26 & 1.39 & 1.92 & 1.78 & 3.08 & 2.01 & 3.22 & \textbf{1.04} &  &  &  \\

  \hline

  Rank & &
  2 & 6 & 3 & 4 & 7 & 5 & 9 & 8 & 10 & 1 &  &  &  \\\hline\hline
           
  \bytes & \distuniform & 1.21 & 1.16 & 1.13 & 1.51 & 1.52 & 1.21 & 2.02 & 1.55 & 2.65 & \textbf{1.09} &  &  &  \\

  \hline

  Rank & &
  4 & 3 & 2 & 6 & 7 & 5 & 9 & 8 & 10 & 1 &  &  &  \\\hline\hline
\end{tabular}

  }
  \caption{
    Average slowdowns of sequential algorithms for different data types and input distributions on \pcintellargefour.
    The slowdowns average over input sizes with at least $2^{18}$ bytes.
    \label{tab:slowdown seq 135}
  }
\end{table}
\begin{table}[!ht]
  \resizebox*{!}{0.93\textheight}{

\begin{tabular}{ll|rrrrrrrrr|rrrr}
  Type
  & Distribution
  & \rotatebox[origin=c]{90}{\compissssort}
  & \rotatebox[origin=c]{90}{\compspdq}
  & \rotatebox[origin=c]{90}{\compblock}
  & \rotatebox[origin=c]{90}{\compmyssssaxtmann}
  & \rotatebox[origin=c]{90}{\compsyaros}
  & \rotatebox[origin=c]{90}{\compssort}
  & \rotatebox[origin=c]{90}{\compstim}
  & \rotatebox[origin=c]{90}{\compsmergequick}
  & \rotatebox[origin=c]{90}{\compswiki}
  & \rotatebox[origin=c]{90}{\radixsska}
  & \rotatebox[origin=c]{90}{\radixipp}
  & \rotatebox[origin=c]{90}{\radixlearned}
  & \rotatebox[origin=c]{90}{\compiparassrsort}\\\hline
  \double &        \distsorted & \textbf{1.01} & 1.76 & 35.11 &          1.09 & 15.38 & 21.09 & 1.10 & 76.15 & 2.58 & 28.91 & 72.53 & 65.62 &  \\
  \double & \distreversesorted &          1.02 & 1.87 & 15.19 & \textbf{1.01} &  5.04 &  5.53 & 1.10 & 26.62 & 6.54 & 10.39 & 25.48 & 22.04 &  \\
  \double &          \distones & \textbf{1.03} & 1.83 & 31.36 &          1.09 &  1.25 & 17.36 & 1.10 &  2.50 & 3.36 & 21.07 & 32.62 & 18.57 &  \\

  \hline\hline
  
  \double &            \distexpo & \textbf{1.02} & 1.29 & 1.59 & 1.23 & 2.45 & 2.69 &          4.36 & 4.42 & 4.48 & 1.46 &          1.38 & 1.86 &  \\
  \double &            \distzipf & \textbf{1.05} & 1.46 & 1.73 & 1.26 & 2.70 & 2.86 &          4.64 & 4.43 & 4.86 & 1.34 &          1.27 & 2.00 &  \\
  \double &  \distduplicatesroot & \textbf{1.13} & 1.99 & 2.51 & 1.75 & 1.61 & 2.43 &          1.34 & 6.93 & 3.54 & 2.44 &          2.98 & 3.22 &  \\
  \double & \distduplicatestwice &          1.11 & 1.38 & 1.46 & 1.27 & 2.40 & 2.46 &          2.82 & 3.45 & 3.11 & 1.11 & \textbf{1.08} & 2.27 &  \\
  \double & \distduplicateseight & \textbf{1.00} & 1.38 & 1.73 & 1.25 & 2.60 & 3.09 &          4.62 & 5.25 & 4.82 & 1.60 &          1.70 & 3.02 &  \\
  \double &    \distalmostsorted &          2.20 & 1.48 & 2.47 & 2.90 & 1.58 & 1.58 & \textbf{1.01} & 6.79 & 2.58 & 2.51 &          3.53 & 4.61 &  \\
  \double &         \distuniform &          1.06 & 1.25 & 1.34 & 1.22 & 2.36 & 2.39 &          3.58 & 3.00 & 3.61 & 1.18 & \textbf{1.05} & 2.03 &  \\

  \hline
  Total  & &

  \textbf{1.18} & 1.45 & 1.79 & 1.48 & 2.20 & 2.45 & 2.78 & 4.69 & 3.77 & 1.58 & 1.66 & 2.90 &  \\

  Rank & &
  1 & 2 & 6 & 3 & 7 & 8 & 9 & 12 & 11 & 4 & 5 & 10 &  \\\hline\hline
           
  \ulong &        \distsorted & 1.27 & 1.84 & 35.98 & \textbf{1.00} & 15.29 & 23.94 &          1.40 & 70.86 & 3.38 & 39.54 & 100.69 & 121.56 & 16.53 \\
  \ulong & \distreversesorted & 1.01 & 1.79 & 13.80 &          1.01 &  4.27 &  5.30 & \textbf{1.01} & 20.68 & 6.63 & 11.26 &  28.76 &  34.85 &  6.19 \\
  \ulong &          \distones & 1.24 & 1.79 & 35.08 & \textbf{1.00} &  1.19 & 15.68 &          1.41 &  2.31 & 4.18 & 24.10 &  36.80 &  17.26 &  1.35 \\

  \hline\hline
  
  \ulong &            \distexpo & 1.06 & 1.42 & 1.72 & 1.31 & 2.32 & 2.72 &          4.82 & 4.20 & 4.88 & 1.29 & 1.79 & 2.13 & \textbf{1.04} \\
  \ulong &            \distzipf & 1.78 & 2.52 & 3.14 & 2.24 & 4.21 & 4.77 &          7.84 & 6.78 & 8.28 & 2.37 & 2.38 & 3.43 & \textbf{1.00} \\
  \ulong &  \distduplicatesroot & 1.62 & 2.81 & 3.93 & 2.90 & 2.23 & 3.47 &          2.32 & 9.39 & 5.50 & 3.38 & 3.92 & 4.93 & \textbf{1.00} \\
  \ulong & \distduplicatestwice & 2.05 & 2.81 & 3.01 & 2.44 & 4.45 & 4.86 &          5.68 & 6.00 & 6.44 & 2.14 & 3.06 & 4.39 & \textbf{1.00} \\
  \ulong & \distduplicateseight & 1.42 & 1.72 & 2.48 & 1.70 & 3.15 & 3.97 &          6.55 & 6.13 & 6.40 & 2.25 & 4.15 & 2.71 & \textbf{1.02} \\
  \ulong &    \distalmostsorted & 2.19 & 1.26 & 2.45 & 3.20 & 1.56 & 1.65 & \textbf{1.13} & 5.94 & 2.94 & 2.89 & 6.59 & 8.09 &          1.18 \\
  \ulong &         \distuniform & 1.44 & 1.95 & 2.07 & 1.74 & 3.10 & 3.43 &          5.37 & 4.19 & 5.44 & 1.38 & 2.09 & 5.14 & \textbf{1.02} \\

  \hline
  Total  & &

  1.61 & 1.98 & 2.60 & 2.13 & 2.84 & 3.37 & 4.11 & 5.88 & 5.47 & 2.13 & 3.13 & 5.44 & \textbf{1.03} \\

  Rank & &
  2 & 3 & 6 & 4 & 7 & 9 & 10 & 13 & 12 & 5 & 8 & 11 & 1 \\\hline\hline

  \uint &        \distsorted & 2.15 & 3.19 & 67.26 & \textbf{2.10} &         31.01 & 47.01 &          2.18 & 149.90 & 5.53 & 62.99 & 38.38 & 262.40 & 32.72 \\
  \uint & \distreversesorted & 1.24 & 2.08 & 18.46 &          1.32 &          6.12 &  7.30 & \textbf{1.07} &  29.95 & 6.22 & 11.91 & 16.88 &  49.89 &  7.38 \\
  \uint &          \distones & 2.32 & 3.12 & 81.68 &          2.37 & \textbf{2.02} & 33.24 &          2.34 &   5.00 & 8.11 & 28.89 & 16.21 &  37.73 &  2.66 \\

  \hline\hline
  
  \uint &            \distexpo & 1.49 & 1.99 & 2.59 & 1.91 & 3.63 & 4.11 &          7.14 &  6.41 & 7.00 & 1.55 & \textbf{1.05} &  3.52 &          1.05 \\
  \uint &            \distzipf & 1.93 & 3.06 & 3.89 & 2.60 & 5.45 & 5.94 &          9.91 &  8.60 & 9.80 & 2.04 &          1.28 &  4.50 & \textbf{1.06} \\
  \uint &  \distduplicatesroot & 1.74 & 3.34 & 4.51 & 3.14 & 2.60 & 4.03 &          2.14 & 11.16 & 5.37 & 2.89 &          2.20 &  5.97 & \textbf{1.00} \\
  \uint & \distduplicatestwice & 2.27 & 3.18 & 3.51 & 2.88 & 5.32 & 5.86 &          6.77 &  7.21 & 7.00 & 1.69 &          1.24 &  6.96 & \textbf{1.02} \\
  \uint & \distduplicateseight & 1.55 & 2.17 & 2.84 & 1.93 & 3.92 & 4.67 &          7.66 &  7.48 & 7.41 & 1.82 &          2.13 &  4.46 & \textbf{1.02} \\
  \uint &    \distalmostsorted & 2.82 & 1.69 & 2.91 & 4.26 & 1.96 & 1.92 & \textbf{1.00} &  7.71 & 2.97 & 2.72 &          4.37 & 10.75 &          1.39 \\
  \uint &         \distuniform & 1.75 & 2.33 & 2.66 & 2.31 & 4.27 & 4.50 &          6.87 &  5.22 & 6.57 & 1.67 & \textbf{1.02} &  5.87 &          1.16 \\

  \hline
  Total  & &

  1.89 & 2.46 & 3.21 & 2.62 & 3.67 & 4.21 & 4.74 & 7.50 & 6.25 & 2.00 & 1.66 & 7.24 & \textbf{1.09} \\

  Rank & &
  3 & 5 & 7 & 6 & 8 & 9 & 10 & 13 & 11 & 4 & 2 & 12 & 1 \\\hline\hline
  
  \pair &        \distsorted & 1.03 & 1.77 & 23.06 & \textbf{1.01} & 12.03 & 17.28 & 1.04 & 44.20 & 2.29 & 24.60 &  &  & 13.90 \\
  \pair & \distreversesorted & 1.05 & 1.18 &  8.67 & \textbf{1.04} &  3.83 &  4.29 & 1.04 & 13.90 & 7.30 &  7.52 &  &  &  5.51 \\
  \pair &          \distones & 1.02 & 1.66 & 18.17 & \textbf{1.02} &  1.09 & 14.29 & 1.03 &  2.20 & 2.84 & 15.14 &  &  &  1.27 \\

  \hline\hline
  
  \pair &            \distexpo & 1.13 &          2.04 & 1.28 & 1.17 & 1.92 & 2.18 & 3.83 & 3.30 & 4.46 & 1.14 &  &  & \textbf{1.08} \\
  \pair &            \distzipf & 1.46 &          2.80 & 1.74 & 1.56 & 2.64 & 2.91 & 5.04 & 3.91 & 5.82 & 1.54 &  &  & \textbf{1.02} \\
  \pair &  \distduplicatesroot & 1.49 &          1.78 & 2.44 & 1.96 & 1.77 & 2.60 & 2.04 & 5.89 & 4.98 & 2.24 &  &  & \textbf{1.00} \\
  \pair & \distduplicatestwice & 1.65 &          2.98 & 1.81 & 1.67 & 2.85 & 2.96 & 3.82 & 3.65 & 4.76 & 1.52 &  &  & \textbf{1.00} \\
  \pair & \distduplicateseight & 1.32 &          2.28 & 1.63 & 1.40 & 2.29 & 2.61 & 4.91 & 4.22 & 5.01 & 1.74 &  &  & \textbf{1.00} \\
  \pair &    \distalmostsorted & 3.63 & \textbf{1.00} & 3.48 & 4.50 & 2.54 & 2.64 & 2.41 & 7.49 & 5.13 & 3.90 &  &  &          2.09 \\
  \pair &         \distuniform & 1.42 &          2.63 & 1.65 & 1.53 & 2.60 & 2.65 & 4.23 & 3.21 & 4.77 & 1.22 &  &  & \textbf{1.05} \\

  \hline
  Total  & &

  1.60 & 2.10 & 1.91 & 1.78 & 2.34 & 2.64 & 3.58 & 4.32 & 4.97 & 1.75 &  &  & \textbf{1.14} \\

  Rank & &
  2 & 6 & 5 & 4 & 7 & 8 & 9 & 10 & 11 & 3 &  &  & 1 \\\hline\hline
  
  \quartet & \distuniform & 1.15 & 1.89 & 1.46 & 1.34 & 1.91 & 1.98 & 3.37 & 2.38 & 3.89 & \textbf{1.01} &  &  &  \\

  \hline

  Rank & &
  2 & 5 & 4 & 3 & 6 & 7 & 9 & 8 & 10 & 1 &  &  &  \\\hline\hline
           
  \bytes & \distuniform & 1.52 & 1.35 & 1.45 & 1.54 & 2.17 & 1.45 & 2.42 & 2.06 & 3.75 & \textbf{1.01} &  &  &  \\

  \hline

  Rank & &
  5 & 2 & 4 & 6 & 8 & 3 & 9 & 7 & 10 & 1 &  &  &  \\\hline\hline
\end{tabular}

  }
  \caption{
    Average slowdowns of sequential algorithms for different data types and input distributions on \pcamd.
    The slowdowns average over input sizes with at least $2^{18}$ bytes.
    \label{tab:slowdown seq 133}
  }
\end{table}
\begin{table}[!ht]
  \resizebox*{!}{0.93\textheight}{

\begin{tabular}{ll|rrrrrrrrr|rrrr}
  Type
  & Distribution
  & \rotatebox[origin=c]{90}{\compissssort}
  & \rotatebox[origin=c]{90}{\compspdq}
  & \rotatebox[origin=c]{90}{\compblock}
  & \rotatebox[origin=c]{90}{\compmyssssaxtmann}
  & \rotatebox[origin=c]{90}{\compsyaros}
  & \rotatebox[origin=c]{90}{\compssort}
  & \rotatebox[origin=c]{90}{\compstim}
  & \rotatebox[origin=c]{90}{\compsmergequick}
  & \rotatebox[origin=c]{90}{\compswiki}
  & \rotatebox[origin=c]{90}{\radixsska}
  & \rotatebox[origin=c]{90}{\radixipp}
  & \rotatebox[origin=c]{90}{\radixlearned}
  & \rotatebox[origin=c]{90}{\compiparassrsort}\\\hline
  \double &        \distsorted &          1.07 & 1.83 & 31.60 & \textbf{1.01} & 14.48 & 28.79 & 1.02 & 92.63 & 3.93 & 27.46 & 94.22 & 75.36 &  \\
  \double & \distreversesorted & \textbf{1.03} & 1.70 & 19.10 &          1.06 &  5.72 &  7.75 & 1.14 & 36.85 & 7.45 & 12.07 & 37.44 & 27.41 &  \\
  \double &          \distones &          1.09 & 1.79 & 26.88 & \textbf{1.01} &  1.15 & 19.12 & 1.05 &  3.02 & 4.87 & 21.95 & 28.47 & 19.26 &  \\

  \hline\hline
  
  \double &            \distexpo & \textbf{1.00} &          1.11 & 1.23 & 1.29 & 2.34 & 2.69 & 4.44 & 4.42 & 4.19 &          1.19 & 1.65 & 2.21 &  \\
  \double &            \distzipf & \textbf{1.08} &          1.23 & 1.38 & 1.48 & 2.79 & 3.11 & 5.01 & 4.70 & 4.79 &          1.08 & 1.39 & 2.40 &  \\
  \double &  \distduplicatesroot & \textbf{1.04} &          1.23 & 1.52 & 1.51 & 1.23 & 2.12 & 1.34 & 6.01 & 2.57 &          1.59 & 2.39 & 2.93 &  \\
  \double & \distduplicatestwice &          1.20 &          1.35 & 1.38 & 1.52 & 2.60 & 2.85 & 3.19 & 3.79 & 3.18 & \textbf{1.02} & 1.45 & 2.69 &  \\
  \double & \distduplicateseight & \textbf{1.00} &          1.06 & 1.31 & 1.33 & 2.43 & 2.90 & 4.64 & 5.02 & 4.35 &          1.16 & 1.61 & 2.49 &  \\
  \double &    \distalmostsorted &          2.29 & \textbf{1.01} & 2.07 & 2.76 & 1.60 & 1.87 & 1.30 & 7.43 & 2.22 &          2.26 & 5.61 & 4.73 &  \\
  \double &         \distuniform &          1.05 &          1.20 & 1.20 & 1.34 & 2.43 & 2.54 & 3.82 & 3.26 & 3.63 & \textbf{1.01} & 1.78 & 2.11 &  \\

  \hline
  Total  & &

  1.18 & \textbf{1.16} & 1.42 & 1.55 & 2.13 & 2.55 & 3.00 & 4.79 & 3.45 & 1.28 & 2.00 & 2.98 &  \\

  Rank & &
  2 & 1 & 4 & 5 & 7 & 8 & 10 & 12 & 11 & 3 & 6 & 9 &  \\\hline\hline
           
  \ulong &        \distsorted & 1.11 & 1.83 & 28.42 & \textbf{1.02} & 13.30 & 26.46 &          1.05 & 84.18 & 3.33 & 34.08 & 116.92 & 130.41 & 16.69 \\
  \ulong & \distreversesorted & 1.06 & 1.73 & 17.76 & \textbf{1.02} &  5.28 &  7.00 &          1.04 & 32.52 & 7.72 & 13.81 &  43.72 &  47.28 &  7.59 \\
  \ulong &          \distones & 1.10 & 1.68 & 25.90 &          1.03 &  1.11 & 17.94 & \textbf{1.01} &  2.73 & 4.42 & 22.45 &  28.05 &  17.99 &  1.47 \\

  \hline\hline
  
  \ulong &            \distexpo & \textbf{1.02} &          1.15 & 1.24 & 1.38 & 2.23 & 2.59 & 4.42 & 4.09 & 4.14 &          1.11 & 2.06 & 2.06 &          1.10 \\
  \ulong &            \distzipf &          1.28 &          1.56 & 1.71 & 1.85 & 3.20 & 3.61 & 5.76 & 5.23 & 5.72 &          1.27 & 1.75 & 2.05 & \textbf{1.01} \\
  \ulong &  \distduplicatesroot & \textbf{1.06} &          1.22 & 1.52 & 1.63 & 1.16 & 1.95 & 1.39 & 5.35 & 2.73 &          1.43 & 2.26 & 2.94 &          1.35 \\
  \ulong & \distduplicatestwice &          1.49 &          1.67 & 1.65 & 1.90 & 2.99 & 3.29 & 3.76 & 4.19 & 3.85 &          1.21 & 2.26 & 2.65 & \textbf{1.00} \\
  \ulong & \distduplicateseight &          1.16 &          1.20 & 1.47 & 1.58 & 2.51 & 3.00 & 4.95 & 5.04 & 4.74 &          1.36 & 2.54 & 2.27 & \textbf{1.02} \\
  \ulong &    \distalmostsorted &          2.37 & \textbf{1.02} & 1.91 & 3.08 & 1.60 & 1.76 & 1.41 & 6.81 & 2.39 &          2.69 & 7.30 & 8.47 &          1.21 \\
  \ulong &         \distuniform &          1.25 &          1.41 & 1.38 & 1.61 & 2.61 & 2.80 & 4.09 & 3.43 & 4.01 & \textbf{1.00} & 2.80 & 3.71 &          1.12 \\

  \hline
  Total  & &

  1.32 & 1.30 & 1.54 & 1.80 & 2.21 & 2.64 & 3.25 & 4.77 & 3.79 & 1.37 & 2.66 & 4.32 & \textbf{1.11} \\

  Rank & &
  3 & 2 & 5 & 6 & 7 & 8 & 10 & 13 & 11 & 4 & 9 & 12 & 1 \\\hline\hline

  \uint &        \distsorted & 2.82 & 4.45 & 83.97 &          2.32 & 35.39 & 95.00 & \textbf{2.00} & 234.86 &  8.93 & 61.73 & 92.69 & 426.73 & 42.36 \\
  \uint & \distreversesorted & 1.47 & 2.38 & 26.97 &          1.51 &  7.54 & 12.80 & \textbf{1.00} &  49.47 &  7.45 & 14.09 & 51.86 &  86.09 & 10.10 \\
  \uint &          \distones & 2.51 & 4.09 & 80.92 & \textbf{1.99} &  2.49 & 54.23 &          2.11 &   7.59 & 12.12 & 34.89 & 17.80 &  58.18 &  4.03 \\

  \hline\hline
  
  \uint &            \distexpo & 1.29 & 1.53 & 1.67 & 1.62 & 3.20 & 3.56 &          5.92 & 5.77 & 5.61 &          1.13 & 1.11 &  2.80 & \textbf{1.05} \\
  \uint &            \distzipf & 1.67 & 2.10 & 2.40 & 2.18 & 4.76 & 5.14 &          8.11 & 7.72 & 7.94 & \textbf{1.09} & 1.34 &  3.13 &          1.22 \\
  \uint &  \distduplicatesroot & 1.35 & 1.43 & 1.85 & 1.70 & 1.41 & 2.42 &          1.26 & 6.66 & 2.76 & \textbf{1.09} & 1.67 &  3.61 &          1.61 \\
  \uint & \distduplicatestwice & 1.86 & 2.16 & 2.22 & 2.15 & 4.12 & 4.40 &          4.86 & 5.67 & 4.89 & \textbf{1.04} & 1.78 &  3.97 &          1.18 \\
  \uint & \distduplicateseight & 1.28 & 1.46 & 1.77 & 1.59 & 3.29 & 3.76 &          6.00 & 6.28 & 5.70 & \textbf{1.04} & 1.91 &  2.41 &          1.08 \\
  \uint &    \distalmostsorted & 2.77 & 1.22 & 2.43 & 3.56 & 1.89 & 2.76 & \textbf{1.09} & 8.40 & 2.23 &          2.30 & 7.36 & 10.79 &          1.30 \\
  \uint &         \distuniform & 1.35 & 1.62 & 1.56 & 1.63 & 3.20 & 3.25 &          4.79 & 4.01 & 4.53 & \textbf{1.01} & 1.12 &  3.99 &          1.20 \\

  \hline
  Total  & &

  1.59 & 1.61 & 1.96 & 1.98 & 2.91 & 3.51 & 3.68 & 6.22 & 4.45 & \textbf{1.19} & 1.84 & 5.51 & 1.22 \\

  Rank & &
  3 & 4 & 6 & 7 & 8 & 9 & 10 & 13 & 11 & 1 & 5 & 12 & 2 \\\hline\hline
  
  \pair &        \distsorted & 1.08 & 1.72 & 19.73 &          1.02 & 9.96 & 19.02 & \textbf{1.01} & 51.62 & 2.65 & 20.73 &  &  & 11.81 \\
  \pair & \distreversesorted & 1.07 & 1.18 & 10.15 &          1.08 & 3.48 &  4.63 & \textbf{1.07} & 17.97 & 7.55 &  7.46 &  &  &  5.07 \\
  \pair &          \distones & 1.12 & 1.64 & 17.10 & \textbf{1.01} & 1.05 & 16.22 &          1.12 &  2.13 & 3.05 & 13.49 &  &  &  1.15 \\

  \hline\hline
  
  \pair &            \distexpo & \textbf{1.04} &          1.94 & 1.18 & 1.52 & 1.82 & 2.07 & 3.87 & 3.35 & 4.04 &          1.10 &  &  &          1.07 \\
  \pair &            \distzipf &          1.39 &          2.68 & 1.53 & 2.06 & 2.53 & 2.72 & 4.83 & 4.09 & 5.25 &          1.25 &  &  & \textbf{1.00} \\
  \pair &  \distduplicatesroot & \textbf{1.05} &          1.15 & 1.39 & 1.66 & 1.09 & 1.76 & 1.65 & 4.31 & 3.24 &          1.23 &  &  &          1.12 \\
  \pair & \distduplicatestwice &          1.48 &          2.67 & 1.56 & 2.00 & 2.47 & 2.66 & 3.35 & 3.55 & 3.95 &          1.21 &  &  & \textbf{1.02} \\
  \pair & \distduplicateseight &          1.24 &          2.20 & 1.40 & 1.78 & 2.08 & 2.47 & 4.52 & 4.35 & 4.85 &          1.37 &  &  & \textbf{1.00} \\
  \pair &    \distalmostsorted &          3.40 & \textbf{1.00} & 2.62 & 4.12 & 2.17 & 2.50 & 2.80 & 7.94 & 4.55 &          3.17 &  &  &          1.66 \\
  \pair &         \distuniform &          1.29 &          2.39 & 1.35 & 1.76 & 2.24 & 2.34 & 3.75 & 2.98 & 4.00 & \textbf{1.04} &  &  &          1.08 \\

  \hline
  Total  & &

  1.43 & 1.88 & 1.53 & 2.01 & 2.00 & 2.34 & 3.37 & 4.16 & 4.22 & 1.37 &  &  & \textbf{1.12} \\

  Rank & &
  3 & 5 & 4 & 7 & 6 & 8 & 9 & 10 & 11 & 2 &  &  & 1 \\\hline\hline
  
  \quartet & \distuniform & 1.22 & 2.00 & 1.31 & 1.82 & 2.00 & 2.02 & 3.37 & 2.39 & 3.71 & \textbf{1.02} &  &  &  \\

  \hline

  Rank & &
  2 & 5 & 3 & 4 & 6 & 7 & 9 & 8 & 10 & 1 &  &  &  \\\hline\hline
           
  \bytes & \distuniform & 1.51 & 1.30 & 1.29 & 1.88 & 1.84 & 1.38 & 2.39 & 1.98 & 3.40 & \textbf{1.04} &  &  &  \\

  \hline

  Rank & &
  5 & 3 & 2 & 7 & 6 & 4 & 9 & 8 & 10 & 1 &  &  &  \\\hline\hline
\end{tabular}

  }
  \caption{
    Average slowdowns of sequential algorithms for different data types and input distributions on \pcinteltwo.
    The slowdowns average over input sizes with at least $2^{18}$ bytes.
    \label{tab:slowdown seq 132}
  }
\end{table}
\begin{table}[!ht]
  \resizebox*{!}{0.93\textheight}{

\begin{tabular}{ll|rrrrrrrrr|rrrr}
  Type
  & Distribution
  & \rotatebox[origin=c]{90}{\compissssort}
  & \rotatebox[origin=c]{90}{\compspdq}
  & \rotatebox[origin=c]{90}{\compblock}
  & \rotatebox[origin=c]{90}{\compmyssssaxtmann}
  & \rotatebox[origin=c]{90}{\compsyaros}
  & \rotatebox[origin=c]{90}{\compssort}
  & \rotatebox[origin=c]{90}{\compstim}
  & \rotatebox[origin=c]{90}{\compsmergequick}
  & \rotatebox[origin=c]{90}{\compswiki}
  & \rotatebox[origin=c]{90}{\radixsska}
  & \rotatebox[origin=c]{90}{\radixipp}
  & \rotatebox[origin=c]{90}{\radixlearned}
  & \rotatebox[origin=c]{90}{\compiparassrsort}\\\hline
  \double &        \distsorted & \textbf{1.03} & 1.64 & 29.22 &          1.11 & 16.69 & 22.54 & 1.24 & 70.63 & 2.51 & 27.84 & 66.47 & 67.60 &  \\
  \double & \distreversesorted &          1.01 & 1.56 & 11.86 & \textbf{1.00} &  4.75 &  4.98 & 1.04 & 21.65 & 4.68 &  8.75 & 19.44 & 18.32 &  \\
  \double &          \distones & \textbf{1.01} & 1.64 & 21.64 &          1.17 &  1.11 & 16.84 & 1.19 &  2.62 & 3.03 & 18.42 & 32.61 & 16.30 &  \\

  \double &            \distexpo & \textbf{1.03} & 1.10 & 1.20 & 1.24 & 2.62 & 2.92 &          4.78 & 4.32 & 4.79 & 1.36 &          1.32 & 2.17 &  \\
  \double &            \distzipf & \textbf{1.02} & 1.11 & 1.27 & 1.26 & 2.95 & 3.12 &          4.92 & 4.28 & 5.06 & 1.08 &          1.21 & 2.16 &  \\
  \double &  \distduplicatesroot &          1.14 & 1.64 & 1.96 & 1.86 & 1.79 & 2.69 & \textbf{1.14} & 7.04 & 3.71 & 2.32 &          3.24 & 3.98 &  \\
  \double & \distduplicatestwice &          1.18 & 1.26 & 1.31 & 1.37 & 2.80 & 2.95 &          3.25 & 3.48 & 3.54 & 1.09 & \textbf{1.09} & 2.31 &  \\
  \double & \distduplicateseight & \textbf{1.02} & 1.06 & 1.29 & 1.25 & 2.72 & 3.13 &          4.86 & 4.93 & 4.82 & 1.38 &          1.52 & 2.51 &  \\
  \double &    \distalmostsorted &          3.03 & 1.22 & 3.02 & 4.47 & 2.43 & 2.33 & \textbf{1.00} & 9.01 & 3.56 & 3.52 &          4.81 & 6.97 &  \\
  \double &         \distuniform &          1.09 & 1.15 & 1.15 & 1.25 & 2.64 & 2.68 &          3.96 & 3.07 & 3.95 & 1.10 & \textbf{1.04} & 2.13 &  \\

  \hline
  Total  & &

  1.25 & \textbf{1.21} & 1.51 & 1.61 & 2.54 & 2.82 & 2.89 & 4.83 & 4.16 & 1.53 & 1.71 & 3.49 &  \\

  Rank & &
  2 & 1 & 3 & 5 & 7 & 8 & 9 & 12 & 11 & 4 & 6 & 10 &  \\\hline\hline
           
  \ulong &        \distsorted & 1.33 & 1.92 & 29.77 & \textbf{1.00} & 16.75 & 24.57 &          1.03 & 59.38 & 2.81 & 37.63 & 79.94 & 113.03 & 16.30 \\
  \ulong & \distreversesorted & 1.01 & 1.54 & 10.29 & \textbf{1.01} &  4.11 &  4.67 &          1.01 & 15.95 & 4.66 &  9.84 & 20.49 &  27.22 &  5.05 \\
  \ulong &          \distones & 1.37 & 2.06 & 24.96 &          1.10 &  1.18 & 16.98 & \textbf{1.01} &  2.57 & 3.81 & 22.53 & 39.42 &  16.94 &  1.39 \\

  \ulong &            \distexpo & \textbf{1.04} & 1.18 & 1.33 & 1.34 & 2.51 & 2.89 &          5.09 & 3.73 & 5.07 & 1.26 & 1.62 &  2.06 &          1.04 \\
  \ulong &            \distzipf &          1.73 & 2.06 & 2.39 & 2.33 & 4.80 & 5.32 &          9.03 & 6.36 & 9.19 & 2.15 & 2.51 &  2.73 & \textbf{1.00} \\
  \ulong &  \distduplicatesroot &          1.59 & 2.32 & 2.90 & 2.77 & 2.42 & 3.61 &          1.80 & 7.98 & 5.65 & 3.15 & 4.26 &  4.56 & \textbf{1.00} \\
  \ulong & \distduplicatestwice &          2.04 & 2.47 & 2.55 & 2.62 & 4.94 & 5.36 &          6.36 & 5.47 & 6.95 & 1.99 & 3.41 &  4.35 & \textbf{1.00} \\
  \ulong & \distduplicateseight &          1.37 & 1.52 & 1.91 & 1.77 & 3.43 & 4.06 &          6.89 & 5.57 & 6.84 & 2.13 & 3.39 &  2.48 & \textbf{1.00} \\
  \ulong &    \distalmostsorted &          2.93 & 1.19 & 2.94 & 4.64 & 2.39 & 2.41 & \textbf{1.03} & 7.08 & 3.83 & 4.03 & 7.29 & 11.08 &          1.66 \\
  \ulong &         \distuniform &          1.43 & 1.73 & 1.73 & 1.82 & 3.59 & 3.78 &          5.86 & 3.75 & 5.84 & 1.34 & 2.02 &  4.99 & \textbf{1.00} \\

  \hline
  Total  & &

  \textbf{1.65} & 1.71 & 2.17 & 2.30 & 3.30 & 3.78 & 4.17 & 5.51 & 6.00 & 2.12 & 3.13 & 5.84 \\

  Rank & &
  2 & 3 & 5 & 6 & 8 & 9 & 10 & 11 & 13 & 4 & 7 & 12 \\\hline\hline

  \uint &        \distsorted & 2.48 & 4.51 & 67.12 & 2.80 & 37.74 & 54.21 & \textbf{1.93} & 139.72 & 6.14 & 64.10 & 28.28 & 343.53 & 35.51 \\
  \uint & \distreversesorted & 1.42 & 1.88 & 13.03 & 1.55 &  5.38 &  6.34 & \textbf{1.06} &  19.68 & 4.27 &  8.92 &  7.12 &  40.41 &  5.83 \\
  \uint &          \distones & 2.19 & 3.95 & 60.33 & 2.38 &  2.21 & 41.48 & \textbf{1.96} &   6.01 & 7.85 & 27.01 & 17.69 &  42.95 &  2.97 \\

  \uint &            \distexpo & 1.60 & 1.84 & 2.13 & 1.97 & 4.12 & 4.76 &          7.67 &  5.90 &  7.78 & 1.51 & \textbf{1.00} &  3.45 &          1.09 \\
  \uint &            \distzipf & 2.04 & 2.51 & 3.01 & 2.67 & 6.12 & 6.86 &         10.77 &  7.84 & 11.07 & 1.66 &          1.25 &  4.11 & \textbf{1.06} \\
  \uint &  \distduplicatesroot & 1.67 & 2.48 & 3.21 & 2.83 & 2.56 & 4.07 &          1.50 &  9.10 &  5.14 & 2.41 &          2.08 &  5.62 & \textbf{1.00} \\
  \uint & \distduplicatestwice & 2.65 & 3.13 & 3.32 & 3.13 & 6.46 & 7.09 &          7.72 &  6.89 &  8.31 & 1.63 & \textbf{1.09} &  6.71 &          1.10 \\
  \uint & \distduplicateseight & 1.53 & 1.81 & 2.31 & 1.93 & 4.31 & 5.07 &          7.80 &  6.48 &  7.89 & 1.55 &          1.30 &  2.97 & \textbf{1.00} \\
  \uint &    \distalmostsorted & 5.23 & 2.10 & 4.95 & 8.25 & 3.93 & 3.96 & \textbf{1.00} & 12.63 &  5.42 & 5.34 &          5.24 & 23.34 &          2.70 \\
  \uint &         \distuniform & 2.21 & 2.53 & 2.60 & 2.57 & 5.38 & 5.77 &          8.32 &  5.38 &  8.39 & 1.75 & \textbf{1.02} &  7.46 &          1.29 \\

  \hline
  Total  & &

  2.21 & 2.31 & 2.97 & 2.94 & 4.51 & 5.24 & 4.84 & 7.49 & 7.49 & 2.03 & \textbf{1.53} & 10.24 \\

  Rank & &
  4 & 5 & 7 & 6 & 8 & 10 & 9 & 11 & 12 & 3 & 2 & 13 \\\hline\hline
  
  \pair &        \distsorted &          1.03 & 1.65 & 20.71 & 1.04 & 12.67 & 18.52 & \textbf{1.03} & 35.04 & 2.41 & 23.54 &  &  & 12.43 \\
  \pair & \distreversesorted &          1.08 & 1.18 &  6.89 & 1.07 &  3.77 &  3.90 & \textbf{1.05} & 10.76 & 5.49 &  6.82 &  &  &  4.58 \\
  \pair &          \distones & \textbf{1.02} & 1.65 & 13.38 & 1.04 &  1.06 & 12.81 &          1.03 &  2.03 & 2.78 & 12.80 &  &  &  1.24 \\

  \pair &            \distexpo & 1.06 &          2.00 & 1.09 & 1.20 & 1.96 & 2.18 & 4.16 & 2.72 & 4.40 & 1.13 &  &  & \textbf{1.04} \\
  \pair &            \distzipf & 1.53 &          3.18 & 1.63 & 1.74 & 3.06 & 3.31 & 5.96 & 3.77 & 6.50 & 1.58 &  &  & \textbf{1.00} \\
  \pair &  \distduplicatesroot & 1.62 &          1.90 & 2.06 & 2.17 & 1.86 & 2.85 & 1.96 & 5.37 & 5.22 & 2.22 &  &  & \textbf{1.00} \\
  \pair & \distduplicatestwice & 1.67 &          3.23 & 1.69 & 1.86 & 3.10 & 3.34 & 4.26 & 3.32 & 4.99 & 1.55 &  &  & \textbf{1.00} \\
  \pair & \distduplicateseight & 1.24 &          2.37 & 1.34 & 1.44 & 2.31 & 2.70 & 4.96 & 3.64 & 5.16 & 1.72 &  &  & \textbf{1.00} \\
  \pair &    \distalmostsorted & 3.51 & \textbf{1.00} & 3.26 & 4.62 & 2.70 & 2.91 & 1.96 & 6.48 & 5.18 & 4.17 &  &  &          2.03 \\
  \pair &         \distuniform & 1.41 &          2.82 & 1.40 & 1.58 & 2.71 & 2.81 & 4.62 & 2.75 & 4.90 & 1.24 &  &  & \textbf{1.00} \\

  \hline
  Total  & &

  \textbf{1.60} & 2.21 & 1.68 & 1.90 & 2.48 & 2.85 & 3.69 & 3.83 & 5.16 & 1.77 &  &  \\

  Rank & &
  2 & 6 & 3 & 5 & 7 & 8 & 9 & 10 & 11 & 4 &  &  \\\hline\hline
  
  \quartet & \distuniform & 1.13 & 1.82 & 1.24 & 1.28 & 1.96 & 1.84 & 3.04 & 1.95 & 3.67 & \textbf{1.01} &  &  &  \\

  \hline

  Rank & &
  2 & 5 & 3 & 4 & 8 & 6 & 9 & 7 & 10 & 1 &  &  \\\hline\hline
           
  \bytes & \distuniform & 1.53 & 1.38 & 1.40 & 1.68 & 2.10 & 1.42 & 2.25 & 1.79 & 3.50 & \textbf{1.01} &  &  &  \\

  \hline

  Rank & &
  5 & 2 & 3 & 6 & 8 & 4 & 9 & 7 & 10 & 1 &  &  \\\hline\hline
\end{tabular}

  }
  \caption{
    Average slowdowns of sequential algorithms for different data types and input distributions on \pcintelfour.
    The slowdowns average over input sizes with at least $2^{18}$ bytes.
    \label{tab:slowdown seq 128}
  }
\end{table}

% ------------------------------------------------------------------------------
% - Slowdowns of \compmyssssaxtmann and \compssssschneider
% ------------------------------------------------------------------------------

\begin{table}[!ht]
  \centering
  \resizebox*{!}{0.93\textheight}{

\begin{tabular}{ll|rr}
  Type
  & Distribution
  & \compmyssssaxtmann
  & \compssssschneider\\\hline
  \double &        \distsorted & \textbf{1.00} & 35.77 \\
  \double & \distreversesorted & \textbf{1.00} & 16.15 \\
  \double &          \distones & \textbf{1.00} & 17.03 \\

  \hline\hline
  
  \double &            \distexpo & \textbf{1.00} & 1.41 \\
  \double &            \distzipf & \textbf{1.00} & 1.34 \\
  \double &  \distduplicatesroot & \textbf{1.02} & 1.39 \\
  \double & \distduplicatestwice & \textbf{1.00} & 1.23 \\
  \double & \distduplicateseight & \textbf{1.00} & 1.52 \\
  \double &    \distalmostsorted & \textbf{1.02} & 1.13 \\
  \double &         \distuniform & \textbf{1.00} & 1.23 \\

  \hline
  Total  & &

  \textbf{1.01} & 1.32 \\

  Rank & &
  1 & 2 \\\hline\hline
           
  \ulong &        \distsorted & \textbf{1.00} & 37.39 \\
  \ulong & \distreversesorted & \textbf{1.00} & 15.50 \\
  \ulong &          \distones & \textbf{1.00} & 16.41 \\

  \hline\hline
  
  \ulong &            \distexpo & \textbf{1.00} & 1.41 \\
  \ulong &            \distzipf & \textbf{1.00} & 1.31 \\
  \ulong &  \distduplicatesroot & \textbf{1.02} & 1.33 \\
  \ulong & \distduplicatestwice & \textbf{1.00} & 1.24 \\
  \ulong & \distduplicateseight & \textbf{1.00} & 1.48 \\
  \ulong &    \distalmostsorted & \textbf{1.04} & 1.11 \\
  \ulong &         \distuniform & \textbf{1.00} & 1.24 \\

  \hline
  Total  & &

  \textbf{1.01} & 1.30 \\

  Rank & &
  1 & 2 \\\hline\hline

  \uint &        \distsorted & \textbf{1.00} & 39.40 \\
  \uint & \distreversesorted & \textbf{1.01} & 15.16 \\
  \uint &          \distones & \textbf{1.00} & 20.16 \\

  \hline\hline
  
  \uint &            \distexpo & \textbf{1.00} & 1.49 \\
  \uint &            \distzipf & \textbf{1.00} & 1.38 \\
  \uint &  \distduplicatesroot & \textbf{1.01} & 1.45 \\
  \uint & \distduplicatestwice & \textbf{1.00} & 1.25 \\
  \uint & \distduplicateseight & \textbf{1.00} & 1.56 \\
  \uint &    \distalmostsorted & \textbf{1.04} & 1.14 \\
  \uint &         \distuniform & \textbf{1.00} & 1.25 \\

  \hline
  Total  & &

  \textbf{1.01} & 1.35 \\

  Rank & &
  1 & 2 \\\hline\hline
  
  \pair &        \distsorted & \textbf{1.00} & 24.42 \\
  \pair & \distreversesorted & \textbf{1.00} & 10.47 \\
  \pair &          \distones & \textbf{1.00} & 12.16 \\

  \hline\hline
  
  \pair &            \distexpo & \textbf{1.00} & 1.32 \\
  \pair &            \distzipf & \textbf{1.01} & 1.20 \\
  \pair &  \distduplicatesroot & \textbf{1.02} & 1.24 \\
  \pair & \distduplicatestwice & \textbf{1.00} & 1.19 \\
  \pair & \distduplicateseight & \textbf{1.00} & 1.37 \\
  \pair &    \distalmostsorted & \textbf{1.04} & 1.07 \\
  \pair &         \distuniform & \textbf{1.01} & 1.16 \\

  \hline
  Total  & &

  \textbf{1.01} & 1.22 \\

  Rank & &
  1 & 2 \\\hline\hline
  
  \quartet & \distuniform & \textbf{1.01} & 1.12 \\

  \hline

  Rank & &
  1 & 2 \\\hline\hline
           
  \bytes & \distuniform & 1.09 & \textbf{1.04} \\

  \hline

  Rank & &
  2 & 1 \\\hline\hline
\end{tabular}

  }
  \caption{
    Average slowdowns of \compmyssssaxtmann and \compssssschneider for different data types and input distributions.
    The slowdowns average over the machines and input sizes with at least $2^{18}$ bytes.
    \label{tab:slowdown seq s4o ssss}
  }
\end{table}

% ------------------------------------------------------------------------------
% - Parallel running times of a selected set of input instances
% ------------------------------------------------------------------------------

\begin{figure}[!ht]
  \input{extern/ips4o-benchmark-suite-plots/benchmark/running_times/running_time_parallel_allgen_one_machine_pretty_i10pc136.tex}
  \caption{
    Running times of parallel algorithms on different input distributions and data types of size $D$ executed on machine \pcintelfour.
    The radix sorters \radixppbbr, \radixraduls, \radixregion, and \compiparassrsort does not support the data types \double and \bytes.
  }
  \label{fig:par rt distr types 128}
\end{figure}

\begin{figure}[!ht]
\input{extern/ips4o-benchmark-suite-plots/benchmark/running_times/running_time_parallel_allgen_one_machine_pretty_i10pc132.tex}
  \caption{
    Running times of parallel algorithms on different input distributions and data types of size $D$ executed on machine \pcinteltwo.
    The radix sorters \radixppbbr, \radixraduls, \radixregion, and \compiparassrsort does not support the data types \double and \bytes.
  }
  \label{fig:par rt distr types 132}
\end{figure}

\begin{figure}[!ht]
\input{extern/ips4o-benchmark-suite-plots/benchmark/running_times/running_time_parallel_allgen_one_machine_pretty_i10pc133.tex}
  \caption{
    Running times of parallel algorithms on different input distributions and data types of size $D$ executed on machine \pcamd.
    The radix sorters \radixppbbr, \radixraduls, \radixregion, and \compiparassrsort does not support the data types \double and \bytes.
  }
  \label{fig:par rt distr types 133}
\end{figure}

\begin{figure}[!ht]
\input{extern/ips4o-benchmark-suite-plots/benchmark/running_times/running_time_parallel_allgen_one_machine_pretty_i10pc135.tex}
  \caption{
    Running times of parallel algorithms on different input distributions and data types of size $D$ executed on machine \pcintellargefour.
    The radix sorters \radixppbbr, \radixraduls, \radixregion, and \compiparassrsort does not support the data types \double and \bytes.
  }
  \label{fig:par rt distr types 135}
\end{figure}

% ------------------------------------------------------------------------------
% - Parallel Slowdowns for each machine
% ------------------------------------------------------------------------------

\begin{table}
  \resizebox*{!}{0.93\textheight}{

\begin{tabular}{ll|rrrrrr|rrrrrrr}
    Type
  & Distribution
  & \rotatebox[origin=c]{90}{\compiparassssort} 
  &  \rotatebox[origin=c]{90}{\compppbbs}
  & \rotatebox[origin=c]{90}{\compmyparassssaxtmann} 
  & \rotatebox[origin=c]{90}{\comppsort}
  & \rotatebox[origin=c]{90}{\comppbalancedsort} 
  & \rotatebox[origin=c]{90}{\compptbb} 
  & \rotatebox[origin=c]{90}{\radixregion}  
  & \rotatebox[origin=c]{90}{\radixppbbr}
  & \rotatebox[origin=c]{90}{\radixraduls}
  & \rotatebox[origin=c]{90}{\comppaspas}
  & \rotatebox[origin=c]{90}{\compiparassrsort} \\\hline
  \double &        \distsorted &          1.04 & 14.24 & 1.36 & 17.74 &  22.71 & \textbf{1.01} &  &  &  & 65.29 &  \\
  \double & \distreversesorted & \textbf{1.09} &  1.21 & 1.73 &  1.40 &  15.82 &          3.40 &  &  &  &  6.39 &  \\
  \double &          \distones &          1.04 & 12.29 & 1.30 & 19.88 & 319.78 & \textbf{1.00} &  &  &  & 64.20 &  \\

  \hline\hline
  
  \double &            \distexpo & \textbf{1.00} & 1.82 & 1.87 & 2.45 & 3.37 & 15.64 &  &  &  & 5.57 &  \\
  \double &            \distzipf & \textbf{1.00} & 1.89 & 1.98 & 2.51 & 3.25 & 16.17 &  &  &  & 5.99 &  \\
  \double &  \distduplicatesroot & \textbf{1.00} & 1.45 & 2.02 & 2.20 & 3.74 &  6.33 &  &  &  & 6.78 &  \\
  \double & \distduplicatestwice & \textbf{1.00} & 1.90 & 1.73 & 2.23 & 2.92 &  7.01 &  &  &  & 4.85 &  \\
  \double & \distduplicateseight & \textbf{1.00} & 1.84 & 1.94 & 2.29 & 3.34 & 15.16 &  &  &  & 5.63 &  \\
  \double &    \distalmostsorted & \textbf{1.00} & 1.50 & 2.12 & 4.12 & 2.57 &  3.43 &  &  &  & 7.15 &  \\
  \double &         \distuniform & \textbf{1.00} & 1.96 & 1.70 & 2.33 & 3.00 & 12.65 &  &  &  & 4.77 &  \\

  \hline
  Total  & &

  \textbf{1.00} & 1.75 & 1.90 & 2.53 & 3.15 & 9.57 &  &  &  & 5.76 &  \\

  Rank & &
  1 & 2 & 3 & 4 & 5 & 7 &  &  &  & 6 &  \\\hline\hline
  
  \ulong &        \distsorted & 1.17 & 12.63 & 1.31 & 17.38 &  23.67 & \textbf{1.00} &          7.19 & 78.17 & 44.41 &  & 10.76 \\
  \ulong & \distreversesorted & 1.17 &  1.24 & 1.96 &  1.59 &  18.27 &          3.96 & \textbf{1.07} &  7.83 &  4.24 &  &  1.37 \\
  \ulong &          \distones & 1.09 & 12.45 & 1.33 & 19.68 & 317.45 & \textbf{1.00} &          1.02 & 69.97 & 39.00 &  &  1.32 \\

  \hline\hline
  
  \ulong &            \distexpo & \textbf{1.02} & 1.61 & 1.97 & 2.39 & 3.32 & 14.06 &          1.63 &  1.51 & 2.61 &  &          1.29 \\
  \ulong &            \distzipf & \textbf{1.00} & 1.66 & 2.03 & 2.41 & 3.44 & 14.10 &          1.39 & 19.54 & 5.47 &  &          1.23 \\
  \ulong &  \distduplicatesroot & \textbf{1.00} & 1.35 & 2.06 & 2.20 & 3.62 &  7.54 &          1.33 &  9.04 & 5.70 &  &          1.23 \\
  \ulong & \distduplicatestwice & \textbf{1.03} & 1.74 & 1.87 & 2.25 & 3.11 &  7.34 &          1.11 & 10.11 & 3.43 &  &          1.08 \\
  \ulong & \distduplicateseight & \textbf{1.00} & 1.61 & 2.00 & 2.24 & 3.46 & 13.40 &          1.28 & 13.38 & 4.38 &  &          1.23 \\
  \ulong &    \distalmostsorted &          1.11 & 1.58 & 2.46 & 4.72 & 3.22 &  4.39 & \textbf{1.05} &  9.69 & 5.47 &  &          1.30 \\
  \ulong &         \distuniform &          1.11 & 1.96 & 2.01 & 2.59 & 3.15 & 13.15 &          1.40 &  1.26 & 1.27 &  & \textbf{1.03} \\

  \hline
  Total  & &

  \textbf{1.04} & 1.63 & 2.05 & 2.59 & 3.33 & 9.77 & 1.30 & 6.25 & 3.64 &  & 1.20 \\

  Rank & &
  1 & 4 & 5 & 6 & 7 & 10 & 3 & 9 & 8 &  & 2 \\\hline\hline
  
  \uint &        \distsorted & \textbf{1.23} &  9.48 & 1.74 & 10.22 &  17.28 &          2.09 &          4.87 &  7.39 &  &  & 4.96 \\
  \uint & \distreversesorted &          1.67 &  1.84 & 2.56 &  1.87 &  16.85 &          8.19 & \textbf{1.06} &  1.39 &  &  & 1.16 \\
  \uint &          \distones &          1.09 & 13.43 & 1.35 & 22.64 & 474.99 & \textbf{1.00} &          1.01 & 89.40 &  &  & 1.38 \\

  \hline\hline
  
  \uint &            \distexpo & 1.27 & 2.60 & 2.12 & 3.43 & 4.24 & 24.27 & 1.37 & 1.78 &  &  & \textbf{1.00} \\
  \uint &            \distzipf & 1.06 & 2.32 & 1.94 & 3.07 & 3.90 & 22.24 & 1.16 & 6.03 &  &  & \textbf{1.02} \\
  \uint &  \distduplicatesroot & 1.11 & 1.61 & 2.13 & 2.43 & 3.89 &  7.71 & 1.18 & 6.98 &  &  & \textbf{1.08} \\
  \uint & \distduplicatestwice & 1.46 & 3.09 & 2.27 & 3.53 & 4.61 & 12.22 & 1.07 & 1.59 &  &  & \textbf{1.00} \\
  \uint & \distduplicateseight & 1.24 & 2.66 & 2.13 & 3.21 & 3.99 & 23.06 & 1.16 & 1.54 &  &  & \textbf{1.04} \\
  \uint &    \distalmostsorted & 1.51 & 1.99 & 2.60 & 5.20 & 3.69 &  5.49 & 1.12 & 1.52 &  &  & \textbf{1.01} \\
  \uint &         \distuniform & 1.46 & 3.18 & 2.24 & 3.77 & 4.73 & 21.36 & 1.21 & 1.50 &  &  & \textbf{1.01} \\

  \hline
  Total  & &

  1.29 & 2.43 & 2.20 & 3.44 & 4.13 & 14.54 & 1.18 & 2.30 &  &  & \textbf{1.02} \\

  Rank & &
  3 & 6 & 4 & 7 & 8 & 9 & 2 & 5 &  &  & 1 \\\hline\hline
  
  \pair &        \distsorted &          1.05 & 12.66 & 1.32 & 16.25 &  23.98 & \textbf{1.00} & 6.54 & 29.00 & 75.74 &  & 9.68 \\
  \pair & \distreversesorted & \textbf{1.11} &  1.28 & 1.85 &  1.57 &  16.77 &          3.21 & 1.12 &  2.82 &  7.53 &  & 1.39 \\
  \pair &          \distones &          1.06 & 14.76 & 1.29 & 19.14 & 283.98 & \textbf{1.00} & 1.04 & 15.63 & 74.08 &  & 1.35 \\

  \hline\hline
  
  \pair &            \distexpo &          1.21 & 1.59 & 2.27 & 2.33 & 3.41 & 8.89 & 1.93 & \textbf{1.01} &  9.32 &  & 1.62 \\
  \pair &            \distzipf & \textbf{1.00} & 1.35 & 1.90 & 1.93 & 2.91 & 7.31 & 1.38 &          7.94 &  8.41 &  & 1.26 \\
  \pair &  \distduplicatesroot & \textbf{1.02} & 1.24 & 1.87 & 2.00 & 3.49 & 5.12 & 1.26 &          3.38 &  9.61 &  & 1.28 \\
  \pair & \distduplicatestwice & \textbf{1.02} & 1.37 & 1.88 & 1.86 & 2.84 & 4.37 & 1.24 &          4.69 &  6.20 &  & 1.17 \\
  \pair & \distduplicateseight & \textbf{1.02} & 1.36 & 1.96 & 1.89 & 3.10 & 7.29 & 1.31 &          8.00 &  7.57 &  & 1.29 \\
  \pair &    \distalmostsorted & \textbf{1.07} & 1.74 & 2.59 & 4.25 & 3.64 & 3.94 & 1.09 &          4.01 & 10.36 &  & 1.31 \\
  \pair &         \distuniform &          1.08 & 1.50 & 1.98 & 2.04 & 2.94 & 7.26 & 1.47 & \textbf{1.05} &  4.39 &  & 1.07 \\

  \hline
  Total  & &

  \textbf{1.06} & 1.44 & 2.05 & 2.23 & 3.17 & 6.07 & 1.36 & 3.30 & 7.70 &  & 1.28 \\

  Rank & &
  1 & 4 & 5 & 6 & 7 & 9 & 3 & 8 & 10 &  & 2 \\\hline\hline
  
  \quartet & \distuniform & \textbf{1.03} & 1.14 & 1.99 & 1.83 & 2.71 & 4.68 &  &  &  &  &  \\

  \hline

  Rank & &
  1 & 2 & 4 & 3 & 5 & 6 &  &  &  &  &  \\\hline\hline
  
  \bytes & \distuniform & \textbf{1.05} & 1.11 & 2.04 & 1.70 & 2.53 & 3.51 &  &  &  &  &  \\

  \hline

  Rank & &
  1 & 2 & 4 & 3 & 5 & 6 &  &  &  &  &  \\\hline\hline
\end{tabular}

  }
  \caption{
    Average slowdowns of parallel algorithms for different data types and input distributions obtained on machine \pcintelfour.
    The slowdowns average input sizes with at least $2^{21}t$ bytes.
  }
  \label{tab:slowdown par 128}
\end{table}

\begin{table}
  \resizebox*{!}{0.93\textheight}{

\begin{tabular}{ll|rrrrrr|rrrrrrr}
    Type
  & Distribution
  & \rotatebox[origin=c]{90}{\compiparassssort} 
  &  \rotatebox[origin=c]{90}{\compppbbs}
  & \rotatebox[origin=c]{90}{\compmyparassssaxtmann} 
  & \rotatebox[origin=c]{90}{\comppsort}
  & \rotatebox[origin=c]{90}{\comppbalancedsort} 
  & \rotatebox[origin=c]{90}{\compptbb} 
  & \rotatebox[origin=c]{90}{\radixregion}  
  & \rotatebox[origin=c]{90}{\radixppbbr}
  & \rotatebox[origin=c]{90}{\radixraduls}
  & \rotatebox[origin=c]{90}{\comppaspas}
  & \rotatebox[origin=c]{90}{\compiparassrsort} \\\hline
  \double &        \distsorted &          2.47 & 18.27 & 2.81 & 22.46 &  19.43 & \textbf{1.02} &  &  &  & 60.33 &  \\
  \double & \distreversesorted & \textbf{1.05} &  1.20 & 1.78 &  1.40 &   8.79 &          2.24 &  &  &  &  4.29 &  \\
  \double &          \distones &          2.17 & 17.91 & 2.30 & 26.99 & 292.91 & \textbf{1.03} &  &  &  & 60.63 &  \\

  \hline\hline
  
  \double &            \distexpo & \textbf{1.00} & 2.06 & 1.92 & 2.76 & 2.79 & 11.65 &  &  &  & 4.17 &  \\
  \double &            \distzipf & \textbf{1.00} & 2.36 & 2.08 & 3.16 & 3.33 & 13.48 &  &  &  & 4.65 &  \\
  \double &  \distduplicatesroot & \textbf{1.00} & 1.63 & 2.29 & 2.62 & 3.53 &  5.71 &  &  &  & 5.71 &  \\
  \double & \distduplicatestwice & \textbf{1.00} & 2.15 & 1.79 & 2.60 & 2.64 &  5.37 &  &  &  & 3.52 &  \\
  \double & \distduplicateseight & \textbf{1.00} & 2.10 & 1.98 & 2.67 & 2.93 & 11.29 &  &  &  & 4.29 &  \\
  \double &    \distalmostsorted & \textbf{1.00} & 1.47 & 2.07 & 4.21 & 1.57 &  2.72 &  &  &  & 5.05 &  \\
  \double &         \distuniform & \textbf{1.00} & 2.23 & 1.78 & 2.70 & 2.65 &  9.35 &  &  &  & 3.43 &  \\

  \hline
  Total  & &

  \textbf{1.00} & 1.97 & 1.98 & 2.92 & 2.71 & 7.54 &  &  &  & 4.34 &  \\

  Rank & &
  1 & 2 & 3 & 5 & 4 & 7 &  &  &  & 6 &  \\\hline\hline
  
  \ulong &        \distsorted & 2.32 & 16.98 & 2.91 & 21.77 &  18.67 & \textbf{1.01} &          8.19 & 93.89 & 50.76 &  & 12.20 \\
  \ulong & \distreversesorted & 1.34 &  1.43 & 2.28 &  1.75 &  11.17 &          2.83 & \textbf{1.00} &  8.30 &  4.24 &  &  1.47 \\
  \ulong &          \distones & 1.62 & 18.42 & 2.30 & 26.88 & 291.98 & \textbf{1.07} &          1.09 & 85.76 & 52.94 &  &  1.09 \\

  \hline\hline
  
  \ulong &            \distexpo & \textbf{1.04} & 1.98 & 2.06 & 2.79 & 3.01 & 11.02 &          1.58 &  1.45 & 1.96 &  &          1.08 \\
  \ulong &            \distzipf & \textbf{1.00} & 2.17 & 2.09 & 2.99 & 3.24 & 12.17 &          1.32 & 18.30 & 5.64 &  &          1.30 \\
  \ulong &  \distduplicatesroot & \textbf{1.00} & 1.55 & 2.27 & 2.53 & 3.50 &  5.64 &          1.17 &  9.66 & 6.40 &  &          1.26 \\
  \ulong & \distduplicatestwice &          1.19 & 2.36 & 2.10 & 2.91 & 3.12 &  6.18 & \textbf{1.09} & 11.02 & 3.39 &  &          1.15 \\
  \ulong & \distduplicateseight & \textbf{1.05} & 2.02 & 2.11 & 2.66 & 2.98 & 10.68 &          1.14 & 14.02 & 4.45 &  &          1.15 \\
  \ulong &    \distalmostsorted &          1.23 & 1.73 & 2.62 & 5.24 & 1.99 &  3.42 & \textbf{1.04} &  9.95 & 5.18 &  &          1.26 \\
  \ulong &         \distuniform &          1.21 & 2.50 & 2.15 & 3.09 & 3.11 & 10.31 &          1.36 &  1.54 & 1.15 &  & \textbf{1.06} \\

  \hline
  Total  & &

  \textbf{1.10} & 2.02 & 2.19 & 3.08 & 2.95 & 7.80 & 1.23 & 6.43 & 3.49 &  & 1.18 \\

  Rank & &
  1 & 4 & 5 & 7 & 6 & 10 & 3 & 9 & 8 &  & 2 \\\hline\hline
  
  \uint &        \distsorted & 3.33 & 14.36 & 3.55 & 14.59 &  18.15 & \textbf{1.96} &          6.28 &   8.67 &  &  & 6.47 \\
  \uint & \distreversesorted & 1.94 &  2.12 & 2.80 &  2.07 &  12.90 &          5.32 & \textbf{1.02} &   1.28 &  &  & 1.14 \\
  \uint &          \distones & 1.97 & 19.35 & 1.99 & 32.52 & 473.11 & \textbf{1.06} &          1.08 & 105.42 &  &  & 1.09 \\

  \hline\hline
  
  \uint &            \distexpo & 1.46 & 3.38 & 2.43 & 4.20 & 4.33 & 19.28 &          1.32 & 1.93 &  &  & \textbf{1.00} \\
  \uint &            \distzipf & 1.10 & 2.99 & 2.03 & 3.73 & 3.90 & 17.67 & \textbf{1.05} & 5.83 &  &  &          1.10 \\
  \uint &  \distduplicatesroot & 1.21 & 1.94 & 2.41 & 2.96 & 3.48 &  6.35 & \textbf{1.01} & 7.00 &  &  &          1.36 \\
  \uint & \distduplicatestwice & 1.64 & 3.79 & 2.48 & 4.08 & 4.38 &  9.68 & \textbf{1.04} & 2.12 &  &  &          1.06 \\
  \uint & \distduplicateseight & 1.38 & 3.48 & 2.43 & 3.95 & 4.29 & 18.54 &          1.10 & 1.84 &  &  & \textbf{1.09} \\
  \uint &    \distalmostsorted & 1.72 & 2.25 & 2.76 & 6.11 & 2.87 &  4.77 &          1.17 & 1.34 &  &  & \textbf{1.01} \\
  \uint &         \distuniform & 1.53 & 3.63 & 2.23 & 3.95 & 4.09 & 14.73 &          1.09 & 1.73 &  &  & \textbf{1.08} \\

  \hline
  Total  & &

  1.42 & 2.98 & 2.39 & 4.06 & 3.87 & 11.54 & 1.11 & 2.43 &  &  & \textbf{1.09} \\

  Rank & &
  3 & 6 & 4 & 8 & 7 & 9 & 2 & 5 &  &  & 1 \\\hline\hline
  
  \pair &        \distsorted & 2.17 & 14.22 & 2.93 & 19.64 &  17.52 & \textbf{1.03} &          7.26 & 34.12 & 95.13 &  & 11.46 \\
  \pair & \distreversesorted & 1.14 &  1.31 & 2.05 &  1.74 &   9.42 &          2.76 & \textbf{1.03} &  3.11 &  8.53 &  &  1.53 \\
  \pair &          \distones & 1.95 & 20.27 & 2.40 & 25.29 & 197.76 & \textbf{1.03} &          1.06 & 19.66 & 97.93 &  &  1.06 \\

  \hline\hline
  
  \pair &            \distexpo & \textbf{1.06} & 1.49 & 2.02 & 2.36 & 2.45 & 6.71 &          1.52 & 1.07 &  8.42 &  &          1.20 \\
  \pair &            \distzipf & \textbf{1.00} & 1.58 & 1.93 & 2.38 & 2.55 & 6.87 &          1.35 & 7.99 &  9.45 &  &          1.35 \\
  \pair &  \distduplicatesroot & \textbf{1.01} & 1.34 & 2.05 & 2.31 & 3.06 & 4.89 &          1.16 & 4.30 & 10.97 &  &          1.13 \\
  \pair & \distduplicatestwice & \textbf{1.05} & 1.65 & 1.99 & 2.30 & 2.48 & 4.08 &          1.15 & 4.92 &  6.91 &  &          1.17 \\
  \pair & \distduplicateseight & \textbf{1.02} & 1.50 & 2.04 & 2.22 & 2.48 & 6.44 &          1.16 & 7.57 &  8.34 &  &          1.21 \\
  \pair &    \distalmostsorted &          1.06 & 1.67 & 2.60 & 4.49 & 2.12 & 3.25 & \textbf{1.03} & 3.97 & 11.00 &  &          1.33 \\
  \pair &         \distuniform &          1.06 & 1.73 & 2.01 & 2.41 & 2.39 & 6.19 &          1.36 & 1.32 &  4.74 &  & \textbf{1.01} \\

  \hline
  Total  & &

  \textbf{1.04} & 1.56 & 2.08 & 2.56 & 2.49 & 5.31 & 1.24 & 3.55 & 8.26 &  & 1.20 \\

  Rank & &
  1 & 4 & 5 & 7 & 6 & 9 & 3 & 8 & 10 &  & 2 \\\hline\hline
  
  \quartet & \distuniform & \textbf{1.00} & 1.26 & 2.01 & 2.20 & 2.28 & 4.50 &  &  &  &  &  \\

  \hline

  Rank & &
  1 & 2 & 3 & 4 & 5 & 6 &  &  &  &  &  \\\hline\hline
  
  \bytes & \distuniform & \textbf{1.01} & 1.16 & 1.94 & 1.93 & 2.16 & 3.33 &  &  &  &  &  \\

  \hline

  Rank & &
  1 & 2 & 4 & 3 & 5 & 6 &  &  &  &  &  \\\hline\hline
\end{tabular}

  }
  \caption{
    Average slowdowns of parallel algorithms for different data types and input distributions obtained on machine \pcinteltwo.
    The slowdowns average input sizes with at least $2^{21}t$ bytes.
  }
  \label{tab:slowdown par 132}
\end{table}

\begin{table}
  \resizebox*{!}{0.93\textheight}{

\begin{tabular}{ll|rrrrrr|rrrrrrr}
    Type
  & Distribution
  & \rotatebox[origin=c]{90}{\compiparassssort} 
  &  \rotatebox[origin=c]{90}{\compppbbs}
  & \rotatebox[origin=c]{90}{\compmyparassssaxtmann} 
  & \rotatebox[origin=c]{90}{\comppsort}
  & \rotatebox[origin=c]{90}{\comppbalancedsort} 
  & \rotatebox[origin=c]{90}{\compptbb} 
  & \rotatebox[origin=c]{90}{\radixregion}  
  & \rotatebox[origin=c]{90}{\radixppbbr}
  & \rotatebox[origin=c]{90}{\radixraduls}
  & \rotatebox[origin=c]{90}{\comppaspas}
  & \rotatebox[origin=c]{90}{\compiparassrsort} \\\hline
  \double &        \distsorted &          1.08 & 11.04 & 1.26 & 14.97 & 15.74 & \textbf{1.00} &  &  &  & 47.54 &  \\
  \double & \distreversesorted & \textbf{1.01} &  1.15 & 1.45 &  1.53 &  3.63 &          1.48 &  &  &  &  5.11 &  \\
  \double &          \distones &          1.10 &  7.43 & 1.23 & 18.54 & 48.16 & \textbf{1.00} &  &  &  & 47.74 &  \\

  \hline\hline
  
  \double &            \distexpo & \textbf{1.00} & 1.52 & 1.41 & 2.22 & 2.88 & 4.35 &  &  &  & 5.06 &  \\
  \double &            \distzipf & \textbf{1.00} & 1.55 & 1.46 & 2.23 & 2.92 & 4.21 &  &  &  & 4.97 &  \\
  \double &  \distduplicatesroot & \textbf{1.00} & 1.40 & 1.39 & 2.24 & 3.19 & 2.82 &  &  &  & 5.62 &  \\
  \double & \distduplicatestwice & \textbf{1.00} & 1.62 & 1.43 & 2.11 & 2.81 & 2.68 &  &  &  & 4.64 &  \\
  \double & \distduplicateseight & \textbf{1.00} & 1.52 & 1.46 & 2.29 & 2.85 & 4.42 &  &  &  & 5.20 &  \\
  \double &    \distalmostsorted & \textbf{1.00} & 1.52 & 1.96 & 4.39 & 2.24 & 1.86 &  &  &  & 6.73 &  \\
  \double &         \distuniform & \textbf{1.00} & 1.71 & 1.43 & 2.20 & 2.78 & 3.96 &  &  &  & 4.63 &  \\

  \hline
  Total  & &

  \textbf{1.00} & 1.55 & 1.50 & 2.44 & 2.80 & 3.33 &  &  &  & 5.22 &  \\

  Rank & &
  1 & 3 & 2 & 4 & 5 & 6 &  &  &  & 7 &  \\\hline\hline
  
  \ulong &        \distsorted & 1.05 & 10.76 & 1.21 & 14.83 & 15.67 & \textbf{1.00} &          6.04 & 42.27 & 27.60 &  & 7.78 \\
  \ulong & \distreversesorted & 1.07 &  1.23 & 1.58 &  1.64 &  3.86 &          1.58 & \textbf{1.04} &  5.03 &  3.12 &  & 1.27 \\
  \ulong &          \distones & 1.06 &  7.41 & 1.17 & 18.31 & 47.38 & \textbf{1.00} &          1.01 & 34.98 & 30.32 &  & 1.06 \\

  \hline\hline
  
  \ulong &            \distexpo & \textbf{1.03} & 1.46 & 1.50 & 2.25 & 2.91 & 3.99 &          1.40 & 1.63 & 1.91 &  &          1.24 \\
  \ulong &            \distzipf & \textbf{1.00} & 1.47 & 1.49 & 2.22 & 2.94 & 3.82 &          1.29 & 6.92 & 3.55 &  &          1.27 \\
  \ulong &  \distduplicatesroot & \textbf{1.00} & 1.36 & 1.39 & 2.23 & 3.16 & 2.80 &          1.22 & 5.67 & 4.18 &  &          1.15 \\
  \ulong & \distduplicatestwice & \textbf{1.04} & 1.57 & 1.53 & 2.21 & 2.95 & 2.71 &          1.14 & 5.64 & 2.80 &  &          1.11 \\
  \ulong & \distduplicateseight & \textbf{1.03} & 1.45 & 1.50 & 2.29 & 2.90 & 4.06 &          1.22 & 7.00 & 3.27 &  &          1.22 \\
  \ulong &    \distalmostsorted &          1.08 & 1.66 & 2.16 & 4.78 & 2.42 & 2.05 & \textbf{1.07} & 6.63 & 4.23 &  &          1.18 \\
  \ulong &         \distuniform &          1.05 & 1.67 & 1.54 & 2.31 & 2.97 & 3.84 &          1.25 & 1.44 & 1.19 &  & \textbf{1.02} \\

  \hline
  Total  & &

  \textbf{1.03} & 1.52 & 1.57 & 2.51 & 2.88 & 3.23 & 1.22 & 4.08 & 2.79 &  & 1.17 \\

  Rank & &
  1 & 4 & 5 & 6 & 8 & 9 & 3 & 10 & 7 &  & 2 \\\hline\hline
  
  \uint &        \distsorted & 1.14 & 14.38 & 1.33 & 16.51 & 14.74 & \textbf{1.13} &          5.53 & 16.98 &  &  & 6.20 \\
  \uint & \distreversesorted & 1.25 &  1.49 & 1.68 &  1.67 &  4.32 &          1.96 & \textbf{1.02} &  1.75 &  &  & 1.12 \\
  \uint &          \distones & 1.12 &  8.33 & 1.27 & 19.15 & 56.60 & \textbf{1.00} &          1.02 & 48.85 &  &  & 1.05 \\

  \hline\hline
  
  \uint &            \distexpo &          1.10 & 2.20 & 1.60 & 2.70 & 3.18 & 6.38 & 1.27 & 1.64 &  &  & \textbf{1.07} \\
  \uint &            \distzipf & \textbf{1.02} & 2.07 & 1.52 & 2.58 & 3.14 & 5.97 & 1.21 & 4.53 &  &  &          1.20 \\
  \uint &  \distduplicatesroot & \textbf{1.01} & 1.68 & 1.53 & 2.31 & 3.21 & 3.51 & 1.19 & 5.33 &  &  &          1.13 \\
  \uint & \distduplicatestwice &          1.12 & 2.36 & 1.57 & 2.61 & 3.14 & 3.75 & 1.10 & 1.75 &  &  & \textbf{1.00} \\
  \uint & \distduplicateseight & \textbf{1.04} & 2.08 & 1.53 & 2.56 & 3.06 & 6.01 & 1.17 & 2.00 &  &  &          1.09 \\
  \uint &    \distalmostsorted &          1.19 & 1.99 & 2.17 & 5.53 & 2.56 & 2.28 & 1.10 & 2.32 &  &  & \textbf{1.04} \\
  \uint &         \distuniform &          1.17 & 2.56 & 1.57 & 2.79 & 3.17 & 5.69 & 1.13 & 1.43 &  &  & \textbf{1.00} \\

  \hline
  Total  & &

  1.09 & 2.12 & 1.63 & 2.89 & 3.06 & 4.53 & 1.17 & 2.29 &  &  & \textbf{1.07} \\

  Rank & &
  2 & 5 & 4 & 7 & 8 & 9 & 3 & 6 &  &  & 1 \\\hline\hline
  
  \pair &        \distsorted &          1.06 & 12.94 & 1.20 & 14.57 & 15.62 & \textbf{1.00} & 6.13 & 17.80 & 55.45 &  & 7.71 \\
  \pair & \distreversesorted & \textbf{1.06} &  1.51 & 1.54 &  1.68 &  3.87 &          1.51 & 1.10 &  2.11 &  6.43 &  & 1.28 \\
  \pair &          \distones &          1.07 &  9.92 & 1.15 & 18.42 & 44.53 & \textbf{1.00} & 1.01 &  8.84 & 59.46 &  & 1.09 \\

  \hline\hline
  
  \pair &            \distexpo & \textbf{1.04} & 1.48 & 1.48 & 2.19 & 2.83 & 2.91 & 1.44 & 1.04 & 5.90 &  & 1.29 \\
  \pair &            \distzipf & \textbf{1.00} & 1.45 & 1.45 & 2.09 & 2.78 & 2.77 & 1.32 & 3.22 & 6.72 &  & 1.31 \\
  \pair &  \distduplicatesroot & \textbf{1.00} & 1.57 & 1.31 & 2.29 & 3.27 & 2.63 & 1.21 & 2.73 & 7.61 &  & 1.20 \\
  \pair & \distduplicatestwice & \textbf{1.01} & 1.42 & 1.48 & 2.03 & 2.79 & 2.29 & 1.25 & 2.68 & 5.85 &  & 1.11 \\
  \pair & \distduplicateseight & \textbf{1.02} & 1.48 & 1.49 & 2.19 & 2.81 & 3.01 & 1.17 & 3.28 & 6.60 &  & 1.26 \\
  \pair &    \distalmostsorted & \textbf{1.05} & 1.98 & 2.06 & 4.18 & 2.44 & 1.86 & 1.09 & 2.75 & 8.50 &  & 1.15 \\
  \pair &         \distuniform & \textbf{1.04} & 1.50 & 1.55 & 2.15 & 2.89 & 2.80 & 1.31 & 1.12 & 4.40 &  & 1.05 \\

  \hline
  Total  & &

  \textbf{1.02} & 1.54 & 1.53 & 2.37 & 2.82 & 2.58 & 1.25 & 2.20 & 6.39 &  & 1.19 \\

  Rank & &
  1 & 5 & 4 & 7 & 9 & 8 & 3 & 6 & 10 &  & 2 \\\hline\hline
  
  \quartet & \distuniform & \textbf{1.01} & 1.19 & 1.45 & 1.96 & 2.55 & 2.17 &  &  &  &  &  \\

  \hline

  Rank & &
  1 & 2 & 3 & 4 & 6 & 5 &  &  &  &  &  \\\hline\hline
  
  \bytes & \distuniform & \textbf{1.03} & 1.12 & 1.48 & 2.00 & 2.43 & 2.09 &  &  &  &  &  \\

  \hline

  Rank & &
  1 & 2 & 3 & 4 & 6 & 5 &  &  &  &  &  \\\hline\hline
\end{tabular}

  }
  \caption{
    Average slowdowns of parallel algorithms for different data types and input distributions obtained on machine \pcamd.
    The slowdowns average input sizes with at least $2^{21}t$ bytes.
  }
  \label{tab:slowdown par 133}
\end{table}

\begin{table}
  \resizebox*{!}{0.93\textheight}{

\begin{tabular}{ll|rrrrrr|rrrrrrr}
    Type
  & Distribution
  & \rotatebox[origin=c]{90}{\compiparassssort} 
  &  \rotatebox[origin=c]{90}{\compppbbs}
  & \rotatebox[origin=c]{90}{\compmyparassssaxtmann} 
  & \rotatebox[origin=c]{90}{\comppsort}
  & \rotatebox[origin=c]{90}{\comppbalancedsort} 
  & \rotatebox[origin=c]{90}{\compptbb} 
  & \rotatebox[origin=c]{90}{\radixregion}  
  & \rotatebox[origin=c]{90}{\radixppbbr}
  & \rotatebox[origin=c]{90}{\radixraduls}
  & \rotatebox[origin=c]{90}{\comppaspas}
  & \rotatebox[origin=c]{90}{\compiparassrsort} \\\hline
  \double &        \distsorted &          1.38 &  4.36 & 3.57 & 8.71 &   3.91 & \textbf{1.24} &  &  &  & 11.66 &  \\
  \double & \distreversesorted & \textbf{1.08} &  1.99 & 3.54 & 3.50 &  33.14 &          8.28 &  &  &  &  6.39 &  \\
  \double &          \distones &          2.23 & 15.45 & 2.95 & 3.79 & 161.14 & \textbf{1.26} &  &  &  & 11.06 &  \\

  \hline\hline
  
  \double &            \distexpo & \textbf{1.01} & 1.89 & 3.03 & 3.00 & 4.00 & 17.67 &  &  &  & 5.42 &  \\
  \double &            \distzipf & \textbf{1.00} & 2.05 & 3.38 & 3.34 & 5.28 & 20.29 &  &  &  & 6.38 &  \\
  \double &  \distduplicatesroot & \textbf{1.00} & 1.68 & 3.85 & 3.18 & 5.64 &  9.79 &  &  &  & 8.00 &  \\
  \double & \distduplicatestwice & \textbf{1.00} & 2.05 & 2.86 & 2.96 & 3.80 &  9.77 &  &  &  & 5.21 &  \\
  \double & \distduplicateseight & \textbf{1.00} & 1.83 & 2.93 & 2.67 & 3.82 & 15.83 &  &  &  & 5.17 &  \\
  \double &    \distalmostsorted & \textbf{1.01} & 2.83 & 4.03 & 9.66 & 2.62 & 10.32 &  &  &  & 7.06 &  \\
  \double &         \distuniform & \textbf{1.00} & 2.08 & 2.75 & 2.97 & 3.76 & 15.88 &  &  &  & 5.24 &  \\

  \hline
  Total  & &

  \textbf{1.00} & 2.03 & 3.23 & 3.56 & 4.02 & 13.66 &  &  &  & 5.99 &  \\

  Rank & &
  1 & 2 & 3 & 4 & 5 & 7 &  &  &  & 6 &  \\\hline\hline
  
  \ulong &        \distsorted &          1.47 &  4.75 & 2.22 & 9.82 &   4.12 & \textbf{1.45} & 5.51 & 26.12 & 16.95 &  &          5.50 \\
  \ulong & \distreversesorted & \textbf{1.10} &  1.91 & 3.56 & 3.77 &  31.91 &          8.39 & 3.24 & 12.98 &  8.60 &  &          4.15 \\
  \ulong &          \distones &          4.94 & 18.91 & 3.54 & 4.48 & 190.32 &          1.58 & 3.37 & 27.58 & 15.69 &  & \textbf{1.21} \\

  \hline\hline
  
  \ulong &            \distexpo & \textbf{1.08} & 1.96 & 3.20 & 3.14 & 4.87 & 19.91 & 3.05 &  1.80 &  4.86 &  &          1.23 \\
  \ulong &            \distzipf & \textbf{1.00} & 2.00 & 3.51 & 3.24 & 5.38 & 19.31 & 3.08 & 30.36 & 12.53 &  &          4.42 \\
  \ulong &  \distduplicatesroot & \textbf{1.00} & 1.65 & 3.87 & 3.27 & 5.71 &  9.94 & 3.84 & 19.67 & 16.42 &  &          3.43 \\
  \ulong & \distduplicatestwice & \textbf{1.00} & 1.97 & 2.89 & 2.83 & 3.73 &  9.85 & 2.22 & 15.49 &  7.37 &  &          2.55 \\
  \ulong & \distduplicateseight & \textbf{1.01} & 1.69 & 2.87 & 2.49 & 3.84 & 14.73 & 2.14 & 17.58 & 10.19 &  &          2.74 \\
  \ulong &    \distalmostsorted & \textbf{1.00} & 2.83 & 4.07 & 9.64 & 2.77 & 10.81 & 3.29 & 14.71 & 10.28 &  &          3.34 \\
  \ulong &         \distuniform &          1.15 & 2.30 & 3.20 & 3.31 & 4.30 & 16.91 & 2.87 &  1.40 &  3.02 &  & \textbf{1.00} \\

  \hline
  Total  & &

  \textbf{1.03} & 2.03 & 3.35 & 3.58 & 4.26 & 13.92 & 2.87 & 8.75 & 8.12 &  & 2.38 \\

  Rank & &
  1 & 2 & 5 & 6 & 7 & 10 & 4 & 9 & 8 &  & 3 \\\hline\hline
  
  \uint &        \distsorted & \textbf{2.01} &  4.80 & 7.19 & 7.18 &   9.39 & 3.03 & 4.44 &  3.46 &  &  &          2.84 \\
  \uint & \distreversesorted & \textbf{1.21} &  1.93 & 2.91 & 2.68 &  23.08 & 8.82 & 2.15 &  1.42 &  &  &          1.27 \\
  \uint &          \distones &          2.75 & 29.11 & 4.61 & 8.71 & 533.66 & 1.97 & 5.36 & 52.76 &  &  & \textbf{1.33} \\

  \hline\hline
  
  \uint &            \distexpo &          1.45 & 3.34 & 3.61 & 4.60 & 7.78 & 34.24 & 2.88 &  3.03 &  &  & \textbf{1.02} \\
  \uint &            \distzipf & \textbf{1.00} & 2.81 & 3.06 & 3.52 & 5.84 & 26.87 & 2.31 & 10.95 &  &  &          3.26 \\
  \uint &  \distduplicatesroot & \textbf{1.00} & 1.89 & 3.29 & 2.78 & 5.74 &  8.66 & 2.69 & 12.64 &  &  &          2.65 \\
  \uint & \distduplicatestwice &          1.40 & 3.56 & 3.25 & 4.29 & 5.78 & 16.27 & 2.09 &  1.86 &  &  & \textbf{1.03} \\
  \uint & \distduplicateseight &          1.27 & 3.25 & 3.27 & 4.07 & 6.40 & 28.26 & 2.25 &  2.07 &  &  & \textbf{1.11} \\
  \uint &    \distalmostsorted & \textbf{1.14} & 2.06 & 3.05 & 5.78 & 4.15 &  7.46 & 2.25 &  1.52 &  &  &          1.28 \\
  \uint &         \distuniform &          1.53 & 3.73 & 3.45 & 4.34 & 6.69 & 26.29 & 2.50 &  1.80 &  &  & \textbf{1.00} \\

  \hline
  Total  & &

  \textbf{1.24} & 2.86 & 3.28 & 4.11 & 5.96 & 18.42 & 2.41 & 3.04 &  &  & 1.44 \\

  Rank & &
  1 & 4 & 6 & 7 & 8 & 9 & 3 & 5 &  &  & 2 \\\hline\hline
  
  \pair &        \distsorted &          1.52 &  2.93 & 2.33 & 10.32 &  8.15 & \textbf{1.07} & 3.48 & 8.00 & 15.73 &  &          4.38 \\
  \pair & \distreversesorted & \textbf{1.07} &  1.91 & 3.14 &  5.76 & 21.25 &          8.15 & 2.92 & 5.88 & 11.18 &  &          4.00 \\
  \pair &          \distones &          3.63 & 12.24 & 2.80 &  5.22 & 70.67 &          1.32 & 2.08 & 5.97 & 17.38 &  & \textbf{1.15} \\

  \hline\hline
  
  \pair &            \distexpo &          1.20 & 2.90 & 3.64 &  5.12 & 4.04 & 14.27 & 3.53 & \textbf{1.18} & 18.28 &  &          1.53 \\
  \pair &            \distzipf & \textbf{1.00} & 2.27 & 3.31 &  4.44 & 2.97 & 11.95 & 3.06 &         13.93 & 18.25 &  &          4.98 \\
  \pair &  \distduplicatesroot & \textbf{1.01} & 2.58 & 3.81 &  6.39 & 6.71 &  9.15 & 4.01 &          9.30 & 24.77 &  &          3.42 \\
  \pair & \distduplicatestwice & \textbf{1.00} & 2.46 & 3.05 &  4.22 & 4.07 &  7.30 & 2.54 &          9.19 & 13.48 &  &          3.48 \\
  \pair & \distduplicateseight & \textbf{1.02} & 2.23 & 2.97 &  3.84 & 3.02 &  9.60 & 2.28 &         11.55 & 14.89 &  &          3.43 \\
  \pair &    \distalmostsorted & \textbf{1.00} & 2.66 & 3.79 & 14.09 & 6.58 & 10.68 & 3.25 &          7.75 & 14.84 &  &          4.09 \\
  \pair &         \distuniform &          1.13 & 2.81 & 3.37 &  4.69 & 3.77 & 12.17 & 3.20 &          1.32 &  9.45 &  & \textbf{1.04} \\

  \hline
  Total  & &

  \textbf{1.05} & 2.55 & 3.41 & 5.52 & 4.24 & 10.51 & 3.08 & 5.57 & 15.68 &  & 2.79 \\

  Rank & &
  1 & 2 & 5 & 7 & 6 & 9 & 4 & 8 & 10 &  & 3 \\\hline\hline
  
  \quartet & \distuniform & \textbf{1.01} & 1.64 & 3.28 & 4.45 & 5.09 & 8.95 &  &  &  &  &  \\

  \hline

  Rank & &
  1 & 2 & 3 & 4 & 5 & 6 &  &  &  &  &  \\\hline\hline
  
  \bytes & \distuniform & \textbf{1.14} & 1.17 & 3.61 & 4.73 & 8.11 & 7.00 &  &  &  &  &  \\

  \hline

  Rank & &
  1 & 2 & 3 & 4 & 6 & 5 &  &  &  &  &  \\\hline\hline
\end{tabular}

  }
  \caption{
    Average slowdowns of parallel algorithms for different data types and input distributions obtained on machine \pcintellargefour.
    The slowdowns average input sizes with at least $2^{21}t$ bytes.
  }
  \label{tab:slowdown par 135}
\end{table}

\fi % arxiv

\end{document}